\documentclass[11pt]{article}\synctex=1 
\usepackage{amsfonts}
\usepackage{amssymb}
\usepackage{tikz}
\usetikzlibrary{arrows.meta}

\usepackage{pgfplots}
\pgfplotsset{compat=1.18}
\usetikzlibrary{patterns,calc}

\usepackage{amsmath}
\usepackage{amsthm}
\usepackage{geometry}
\usepackage{commath}
\usepackage{booktabs}
\usepackage{lscape}
\usepackage[flushleft]{threeparttable}
\usepackage{bigstrut}
\usepackage{xr}
\usepackage[authoryear]{natbib}
\renewcommand{\theequation}{\thesection.\arabic{equation}}
\newcommand{\myref}[2]{\hyperref[#1]{#2}}
\numberwithin{equation}{section}

\usepackage{latexsym}
\usepackage{graphicx} 
\usepackage{enumerate}
\usepackage{setspace}
\usepackage[toc,title,titletoc,header]{appendix}
\usepackage[bottom]{footmisc}   
\usepackage[pdfborder={0 0 0}]{hyperref}

\usepackage{lineno}
\setpagewiselinenumbers

\usepackage{listings}
\usepackage{xcolor}
\definecolor{ncsured}{HTML}{CC0000}
\definecolor{agentblue}{HTML}{00467F} 

\newtheorem{lemma}{Lemma}

\theoremstyle{definition}

\theoremstyle{remark}

\newcounter{assumptionA}
\def\theassumptionA{\arabic{assumptionA}}
\newenvironment{assumptionA}[1][]{\refstepcounter{assumptionA}\medskip\noindent{\textbf{Assumption \theassumptionA. #1}}}{\medskip}

\newcounter{theorema}

\def\thetheorem{\arabic{theorema}}
\newcounter{definitiona}

\def\thedefinition{\arabic{definitiona}}
\renewcommand{\thelemma}{\arabic{lemma}}

\hypersetup{
  colorlinks=true,linkcolor=blue, citecolor=blue, filecolor=blue, 
  linkbordercolor={0 1 1}, citebordercolor={1 0 0},
  pdfcreator={\LaTeX\ with package \flqq hyperref\frqq},
  pdfsubject={},
}

\begin{document}

{
\title{Designing Agentic AI-Based Screening for Portfolio Investment}
\author{\textsc{Mehmet Caner\thanks{
			North Carolina State University, Nelson Hall, Department of Economics, NC 27695. Email: mcaner@ncsu.edu. }}
	\and \textsc {Agostino Capponi%
		\thanks{Columbia University. Department of Industrial Engineering and Operations Research and Columbia Business School. Email: ac3827@columbia.edu. }}
\and \textsc {Nathan Sun%
		\thanks{Columbia University. Department of Industrial Engineering and Operations Research. Email: nathan.sun@columbia.edu. }}
\and \textsc {Jonathan Y. Tan%
		\thanks{Columbia University. Department of Industrial Engineering and Operations Research. Email: jyt2123@columbia.edu. }}
}

\date{\today}

\maketitle
} 



\vspace*{-0.2cm}
\begin{abstract}

We introduce a new agentic artificial intelligence (AI) platform for portfolio 
management. Our architecture consists of three layers. First, two large language 
model (LLM) agents are assigned specialized tasks: one agent screens for firms 
with desirable fundamentals, while a sentiment analysis agent screens for firms 
with desirable news. Second, these agents deliberate to generate and agree upon 
buy and sell signals from a large portfolio, substantially narrowing the pool of 
candidate assets. Finally, we apply a high-dimensional precision matrix estimation 
procedure to determine optimal portfolio weights. 
We show, through information acquisition theory, that screening with agentic AI can bring utility gains in screening compared with humans. 
We introduce the concept of 
\emph{sensible screening} and establish that, under mild screening errors, the 
squared Sharpe ratio of the screened portfolio consistently estimates its target. Empirically, our method achieves superior Sharpe ratios relative to 
an unscreened baseline portfolio and to conventional screening approaches, 
evaluated on S\&P~500 data over both short and medium terms.
\vspace{0.1cm}


\end{abstract}


\newpage

\begin{quote}
    \textit{``We need new AI models for the real world---quantitative models\dots AI for the quantitative world is something else entirely, focusing on creating novel medical treatments, de novo material science, and advanced risk management and \textbf{portfolio construction}.''}
    
    \hfill --- Jack D.\ Hidary, CEO of Sandbox AQ. \textit{The Wall Street Journal}, Opinion: ``America Needs AI that can do Math,'' February 17, 2026.
\end{quote}

\begin{quote}
    \textit{``Ilya Sutskever has recently been adding his voice to the beyond LLM camp, stating that \textbf{we won't reach artificial general intelligence with large language models alone}.''}
    
    \hfill --- Jack D.\ Hidary, CEO of Sandbox AQ. \textit{The Wall Street Journal}, Opinion: ``America Needs AI that can do Math,'' February 17, 2026.
\end{quote}

    

\section{Introduction} \label{sec_intro}

Portfolio formation is a foundational concept in the finance literature. Several well-known portfolio formations have been studied, including the global minimum variance
portfolio, the Markowitz mean-variance portfolio, and the maximum Sharpe ratio
portfolio. Traditionally, these frameworks were developed and implemented by human
analysts throughout the twentieth century. With the advent of large language models
coupled with deep learning-based estimation in the twenty-first century, there are
now artificial intelligence-driven portfolios. 

Very recently, multi-agent AI systems have been applied to portfolio optimization. 
This paper contributes to this emerging literature by developing a high-dimensional 
portfolio management framework that integrates stock screening with portfolio 
optimization. Every investment decision involves two distinct tasks: selecting 
which stocks to hold and determining their optimal weights. Accordingly, both the 
set of selected stocks and their portfolio weights are treated as endogenous outputs 
of our framework. Rather than applying quantitative methods to the full stock 
universe or assigning equal weights to an arbitrary subset, we propose an initial 
screening stage in which LLM-based agents generate buy and sell signals. A 
quantitative precision matrix estimation method then determines optimal investment 
weights for the screened stocks. This two-stage approach yields portfolios 
concentrated in stocks with strong return-generating potential, achieving superior 
Sharpe ratios relative to standard benchmarks. By incorporating stock screening, 
the framework deliberately departs from index mimicking, aiming instead to 
generate excess returns, analogous to adopting a large tracking error relative 
to a benchmark index in order to enhance risk-adjusted performance.

A useful analogy is the selection process of a sports team. After observing practice
sessions, the manager chooses the best players to form the starting lineup for game
day. The choice of the optimal subset of players can dramatically improve team
performance, and similarly, selecting the optimal subset of stocks may achieve higher
performance than using the entire universe. Note that by setting up tracking error
and weight constraints on portfolios, both industry practitioners and regulators move
away from using all universe-indexing strategies. A recent paper by \cite{cf2026}
outlines how to form portfolios based on data-dependent inequality constraints and
establishes high-dimensional consistency of these portfolios. They demonstrate that
the Sharpe ratio, returns, and collective risk of a subset of stocks can do better
than those of a benchmark index. Because their restrictions are given by mutual fund
prospectuses and government regulators, there is no performance-based screening
involved in forming their portfolios. A study by
\cite{ast2024} presents strong economic and financial theoretical foundations
explaining why a sparse portfolio can perform better than or equal to the market
index. 


In our approach, screening will be done via Agentic-AI. The financial reasoning behind screening is explained in Section \ref{stockscreeningg}. We show that through information acquisition theory in portfolio choice as pioneered by \cite{vv2010} combined with larger information sets agentic systems deliver better results-higher utilities and reach a better optimal solution  than humans. The main reason is the limited capacity of human beings, and also emotional overreactions.

To give an empirical example, by screening via Agentic AI, our annualized return is 37.31\% (Table \ref{tab35-8}, best Sharpe Ratio portfolio)
and its variance is 0.1162, versus a fully diversified portfolio (Table \ref{tab35-1}, best Sharpe Ratio portfolio), with an annualized return of 18.01\% and a variance of 0.0319 between October 2021-April 2024. In the same time period, human analysts screeners achieve annualized return of 3.27\% with a variance of 0.0144 in Table \ref{tab35-4}.

We see that there is a tradeoff between the low variance/large diversification portfolio versus the high return/less diversified portfolio, but ultimately our Agentic-AI portfolio achieves higher Sharpe ratios than the fully diversified case and also human analyst screeners. The rationale for screening can also be illustrated via the preference for active investors' over an ETF-based general index.
These empirical results are in line with theoretical findings of \cite{vv2010} where they show a higher capacity will result in riskier stocks being selected with higher returns. A formal analysis of their framework applied to Agentic AI is shown in Section \ref{stockscreeningg}.

Recently, in a benchmark paper, \cite{lt2026} used ChatGPT to forecast stock prices. In their paper, they conclude that integration of LLMs with  other quantitative models and machine learning techniques can combine the strengths of different approaches.  Our paper contributes in that respect.
Our multi-agent AI framework consists of three coordinated layers. In the first layer, two
specialized LLM agents operate in parallel: an LLM-Strategy agent (LLM-S), which
screens stocks based on fundamental firm characteristics, including log firm size,
book-to-market ratio, and twelve-month momentum, and a FinBERT-based sentiment
agent, which analyzes financial news articles to generate monthly sentiment-driven
signals. These two agents operate on complementary frequencies: LLM-S is retrained
annually, capturing slow-moving structural narratives about firm quality, while
FinBERT is updated monthly to respond to fast-moving news sentiment. In the second
layer, the two agents deliberate and reach a consensus via an intersection-based
decision rule, narrowing the candidate pool from the full S\&P~500 to
a high-conviction subset averaging approximately 20 stocks. In the third layer, a
quantitative optimization algorithm applies state-of-the-art high-dimensional precision matrix
estimation techniques, such as nodewise regression \citep{mein2006}, residual nodewise regression \citep{caner2022},
POET \citep{fan2013}, deep learning-based methods \citep{c-d2025}, and nonlinear shrinkage \citep{lw2017}, to determine optimal
portfolio weights under global minimum variance, mean-variance, and maximum Sharpe
ratio objectives. The winning method-objective combination is selected based on
out-of-sample Sharpe ratio performance.

We show that not every multi-agent AI configuration creates
value, and that the design of the agent team is critical. We demonstrate this systematically
by comparing our full multi-agent system against a rich set of alternatives: purely
quantitative strategies without any screening, single-agent LLM systems (LLM-S alone
or FinBERT alone), conventional screening methods (logistic regression, human analyst
recommendations, and the \cite{nm2013} profitability-and-value screen), and hybrid
systems that incorporate human analyst judgment alongside AI agents. 

In each case,
our full Agentic AI architecture, specifically two specialized LLM screening agents plus a
quantitative weighting method, dominates. We show that incorporating
human analyst recommendations into the AI ensemble consistently \emph{degrades}
performance, a finding we attribute to the behavioral and emotional biases that human
analysts inevitably carry into their recommendations. This result provides a
concrete, quantitative answer to the practical criticism that the financial value of
AI-based systems is difficult to measure.\footnote{\textit{``Early adopters who
rushed into AI pilots and even deployment last year often hit a wall and learned some
hard lessons\dots Critically, they found it was \textbf{hard to measure financial
returns}.''} --- Isabelle Bousquette, \textit{The Wall Street Journal}, Business
Section, page B4: ``AI software proves to be tougher sell than before,'' February
19, 2026.}

We now consider the theoretical reasons why LLMs, specifically Agentic AI, may deliver better results than traditional machine learning (ML) methods. A study by \cite{kl2026} shows that Agentic AI better explains contemporaneous stock returns around earnings announcements compared to traditional ML methods. They build a factor model-based prediction for individual stock return, demonstrating that the abnormal return over the market return is more accurately captured by the predictions of Agentic AI. The primary reason for this is that, within their model, Agentic AI expands the information set of an unknown function by incorporating earlier earnings calls of the same company, information from the earnings calls of related firms, and news announcements that occur between calls. An Agentic AI system can also process qualitative nuances, such as forward-looking guidance, management tone, and question-and-answer sessions, more effectively than a pure quantitative model. Finally, they demonstrate that when combined with a pure quantitative model, Agentic AI accounts for a significantly greater portion of the variation in stock returns around earnings calls.

In this paper, we design an Agentic AI screening method for stocks. Because the expanded information set processed by Agentic AI may not be easily captured by human analysts or simple technical analysis, our method generates a superior selection of candidate stocks to feed into a quantitative (also referred to as quant) agent.

Beyond the empirical results, we make a novel theoretical contribution to
the high-dimensional portfolio literature. A defining feature of our framework is
that the number of assets in the portfolio is not fixed in advance, but is itself a
random variable realized through the screening process. The existing
high-dimensional portfolio literature, including \cite{fan2011}, \cite{caner2022},
and \cite{cf2026}, treats portfolio dimension as a deterministic, known sequence. We
depart from this convention by introducing the concept of \emph{sensible screening}:
a screening process that, even when it makes mild errors in selecting the exact
composition of the optimal portfolio, always includes the optimal stocks when
selecting enough assets and never selects purely suboptimal ones. Under sensible
screening and standard precision matrix consistency conditions, we establish in
Theorem~\ref{thm:sr_convergence} that the squared Sharpe ratio of the screened portfolio consistently
estimates its target, even when the screening process is imperfect.
Hence under these conditions, mild mistakes by the Agentic AI do not negatively impact the downstream task.

In our empirical analysis, we  cover a medium term-time period. 
Our first empirical analysis  covers the S\&P~500 from 
October 2021 to April 2024, a
period encompassing noticeable events such as the 2022 drawdown, and the strong recovery of 2023. To avoid look-ahead bias, we use ChatGPT-3.5, whose knowledge cutoff date is September 2021.
The S\&P 500 achieves a Sharpe ratio of 0.2685, whereas our Agentic-AI system achieves 1.0946 in the same time period. If we add human analysts, Sharpe Ratio declines to 0.1117.

Our second empirical analysis uses ChatGPT-4o for LLM-S agent in the very short-run. The main reason is to address the possibility of look-ahead bias. ChatGPT-4o, a high performing, vintage LLM, is also used by \cite{lt2026} to avoid look-ahead bias. Since the ChatGPT-4o knowledge cutoff date is October 2023, there is no possibility of look-ahead bias by design.  FinBERT's knowledge cutoff date is prior to 2020.

Our out-of-sample dates range from November 2023 to April 2024. The S\&P 500 achieves a Sharpe ratio of 2.09, whereas our Agentic-AI system achieves 8.16 in the same time period. This is a remarkable difference. If we add human analysts to the ensemble of our agents, the Sharpe ratio decreases to 1.39 in the same time period.
Finally, by using the best quantitative method only (i.e. with no screening), our Sharpe Ratio declines to 3.04 from 8.16. This shows the strength of our Agentic AI system.




For our quantitative strategy, we follow the established literature and form weights
according to global minimum variance (GMV), Markowitz mean-variance (MV), and
 maximum Sharpe ratio (MSR) objectives. 
Our quantitative estimation method applies a matrix-based investing strategy, where estimation
techniques form the rows and different portfolio objectives form the columns. The
winner within this matrix is the technique-objective combination that achieves the
highest Sharpe ratio for the out-of-sample time period being evaluated.

The rest of the paper proceeds as follows. Section~\ref{sec:LLM} provides background on large language models and discusses our proposed Agentic AI framework.  Section~\ref{stockscreeningg} introduces
the stock screening framework, including the LLM-S and FinBERT agents, alternative
screening benchmarks, the algorithm, and our theoretical result on screened portfolio.
Section~\ref{quant strategies} describes the quantitative high-dimensional portfolio
weight formation methods. Section~\ref{sec:empirics} presents our empirical results and 
robustness checks. 
We conclude in Section~\ref{sec_conclusion}.
Appendix \ref{sec_A_proofs} contains technical proofs, Appendix~\ref{sec:quant} provides details on the precision matrix estimation techniques,  and Appendix \ref{additional material} presents additional material including 
extended results,  LLM-S prompts and outputs, and Novy-Marx
screening comparisons.

\section{Agentic AI Portfolio Management Framework}\label{sec:LLM}

We provide background on large language models and their applications in
finance. We then describe the architecture of our proposed multi-agent
portfolio management system and explain how it addresses the main 
methodological concerns raised in the literature.

\subsection{Large Language Models}\label{subsec:LLM_background}

Unlike their predecessors, LLMs exhibit a human-like quality of text
generation and deep contextual reasoning. These capabilities stem from
the transformer architecture \citep{vas2017}, which leverages
self-attention mechanisms to capture long-range dependencies in text far
more effectively than earlier recurrent or convolutional
approaches.\footnote{As documented in \cite{kkss2024}, models with fewer
than ten billion parameters perform on near-chance level across a range of
benchmarks, including arithmetic reasoning, multilingual question
answering, the 57-task Massive Multitask Language Understanding (MMLU)
benchmark \citep{mmlu2021}, and semantic understanding evaluations.
Beyond this scale threshold, performance rises sharply, reaching
approximately 30\% on arithmetic tasks, 40--50\% on multilingual
question answering, and 60--70\% on both MMLU and semantic understanding
benchmarks.}

LLMs are trained on massive datasets, optimizing a vast number of
parameters using deep learning estimators within a transformer
structure, leading to significant performance improvements across a
range of tasks \citep[pp.~390--391]{bb2024}. Leveraging self-supervised
learning, LLMs acquire high-quality representations on unlabeled data
and can then be further fine-tuned to improve performance. Fine-tuning
is typically achieved through Low-Rank Adaptation (LoRA) or
Reinforcement Learning from Human Feedback (RLHF). Alternative
approaches include training the LLM from scratch on in-sample data to
predict out-of-sample outcomes, or employing instruction fine-tuning,
where the model is trained on specifically generated prompt-response
pairs. Fine-tuned LLMs can substantially outperform generic open-source
LLMs in specialized domains such as sentiment analysis and financial
classification \citep[Section~4.1]{li2024}.

An important practical consideration is the choice between open-source
models, such as LLaMA, and closed-source models, such as GPT-4 or
Claude. While open-source models offer advantages in transparency,
reproducibility, and data control, \cite{li2024} (Section~3) find that
closed-source models currently achieve superior performance on standard
financial benchmarks, a gap that likely reflects differences in both
model scale and the proprietary nature of their fine-tuning procedures.

\subsection{LLMs in Finance: Related Work}\label{subsec:LLM_literature}

LLMs offer several advantages for stock selection. Because their
pre-trained nature enables robust zero-shot analysis, they eliminate the
strict requirement for supervised learning. Rather than requiring costly
and time-consuming model retraining, LLMs can rapidly execute
simultaneous tasks such as sentiment analysis and keyword extraction,
and their ability to decompose complex tasks into simpler sub-tasks
makes them well suited to processing the large volumes of financial
reports involved in stock selection \citep{li2024}. A prominent
empirical demonstration of this potential is provided by \cite{ckx2023},
who show that integrating LLM-driven news analysis into long-short
portfolios can deliver very large Sharpe ratios over the 2004--2019
period, outperforming both the S\&P~500 index and classical
bag-of-words approaches.

The application of multi-agent systems to portfolio construction has
been considered by \cite{zhao2025}, who design a framework in which
specialized LLM agents collaboratively execute portfolio construction
tasks. Their system comprises three agents, a fundamental equity agent,
a sentiment agent, and a valuation agent, which must collaboratively
discuss, evaluate, and finalize decisions, thereby actively suppressing
hallucinations made by any single agent. This multi-agent design mirrors
the role of coordinated human analyst teams, which must synthesize
financial disclosures, earnings calls, financial ratio analysis, market
news, and research reports. Recent evidence by \cite{z2025} further
supports this approach: context-based LLMs outperformed their baseline
counterparts by 8.6 percentage points on financial reasoning tasks
(Table~2 therein), suggesting that domain adaptation substantially
improves LLM performance.

Despite this promise, \cite{lkcm2025} raise important methodological
concerns. Their central argument is that existing studies demonstrating
superior LLM performance often rely on very short evaluation horizons
(6--24 months) and narrow asset universes (5 to 100 stocks; see their
Table~1). When evaluated over extended horizons and across a larger
cross-section of stocks, they find that LLM advantages disappear,
allowing simple benchmarks such as buy-and-hold to dominate. They
identify three specific sources of bias. First, \textit{look-ahead
bias}, whereby the model inadvertently processes future data during
training or testing. Second, \textit{data snooping}, arising from
repeated testing on a static dataset leading to inflated results. Third,
\textit{survivorship bias}, whereby delisted stocks are omitted,
inflating apparent returns. Evidence of look-ahead bias is also
documented by \cite{lmr2025} and, in the specific context of stock
market predictions using annual US data over 2020--2023, by
\cite{dfs2025}, who find a small but detectable look-ahead bias at
annual frequency and a substantially larger bias at high frequency.

Following \cite{lkcm2025}, we prevent data snooping by using rolling
windows rather than a fixed in-sample training period. We address
survivorship bias by including all stocks in the S\&P~500, including
those subsequently delisted.

To address the short-horizon critique of \cite{lkcm2025}, we analyze
medium and longer periods. The first, between October 2021 and April 2024,
 encompasses a
diverse set of market regimes: the 2022 market
drawdown, and the strong recovery of 2023. For this medium term,
we benefit from vintage LLMs to solve look-ahead bias. The second period  extends the
analysis to a nearly ten-year horizon from 2015 to 2024. 
However, this may be  subject to look-ahead  bias even though we use causal masking. The results are very similar to the medium term exercise, but we do not report in the paper.  These are available from authors on demand. We also analyzed a very short term with Agentic AI, between November 2023-April 2024.
Since the nature of the short term period is a  boom market versus medium term where there is a major drawdown, it is difficult to understand  whether there is really a decline in performance of Agentic AI or is this due to differences in different economic times. But the S\%P 500 in the very short term was 2.0857, and in the medium term this declines to 0.2685, we can say that decline in Agentic AI performance in the medium term to an SR of 1.0946 versus short term SR of 8.1612. So the declines in both S\%P 500 and Agentic AI cases are very similar in medium term compared with very short term, hence it is not obvious Agentic AI performance declines in the medium term.

\subsection{Our Multi-Agent Portfolio Management
  Framework}\label{subsec:our_framework}

Our framework is a two-stage pipeline applied to the full S\&P~500
universe. In the first stage, an AI screening team, consisting of
LLM-S and FinBERT operating in concert, identifies a targeted subset of
stocks by combining fundamental analysis with sentiment analysis of news
articles. In the second stage, a quantitative precision matrix
estimation algorithm assigns optimal portfolio weights to the screened
stocks. This separation of screening and weighting represents a
scalable, end-to-end application of Agentic AI to portfolio
construction, and is the primary architectural distinction from prior
work such as \cite{zhao2025}, whose study evaluates a randomly selected
basket of 15 technology stocks rather than a comprehensive market-wide
pipeline.

We take the following steps to prevent look-ahead bias. 
We follow the route that is recommended by \cite{lt2026} and \cite{llt2026}: we implement our LLM-S agent using ChatGPT-4o, whose knowledge cutoff date is October 2023. Starting then, we can confidently evaluate our out-of-sample dates without any look-ahead bias. \cite{llt2026} recently show that masking, causal masking, and time-stamping, anonymization of labels do not resolve look-ahead bias adequately due to functional form memorization. They propose using LLMs with clear training sample cutoffs, and hence using vintage LLMs with earlier knowledge cutoff dates to definitively prevent look-ahead bias.
Note also that our
FinBERT model's weights were fine-tuned on the Financial PhraseBank
dataset \citep{malo2014good}, and its training ends prior to 2020.


To suppress hallucinations, we deploy a consensus-based decision rule
across agents. For each stock, it is included in the final
recommendation only if both agents independently agree to buy or sell
it; disagreement leads to exclusion. If the agents' decisions are
mutually exclusive across all stocks, we default to the union of their recommendations to prevent a null
recommendation. With
three agents, the rule becomes a two-out-of-three majority vote.

This intersection strategy reduces the
hallucination risk inherent in any single agent while preserving the
complementarity between fundamental and sentiment signals.
\cite{srk2026} provide a systematic comparison of single-agent LLMs
and multi-agent systems, documenting that the defining advantage of
multi-agent architectures lies in dynamic task decomposition, whereby a
high-level objective is broken down and distributed among specialized
agents. Our framework embodies this principle: LLM-S and FinBERT each
address a distinct analytical sub-problem, and their outputs are
combined through a formal consensus mechanism rather than informal
aggregation.

\subsection{Economic Mechanism Behind Multi-Agent Selection Rule}

We now describe in detail the economic-financial reasoning behind intersection and union rule. We show also why a multi-agent system can perform better than single-agent LLMs.  We follow a {\bf sensible screening} rule to choose from a large portfolio of assets as described in detail in Section~\ref{stockscreeningg}. Sensible screening means either the estimated screened set is a subset of the optimal set, or the optimal set
is a subset of the estimated screened set. 

As will be clear from our proof in the appendix, a large deviation from the optimal screened stocks can result in inconsistent Sharpe ratio estimations. For example,  selecting a larger-than-optimal set (compared to a slightly larger set) may affect economic performance, and hence can be detrimental.

\subsubsection {Intersection Rule}

First, we show why the Agentic AI system performs better than single-agent LLMs for the consensus rule when it selects the intersection. Assume that set ${\cal A}_1$ is the set of stocks selected by LLM-S and the set of the stocks selected by FinBERT is ${\cal A}_2$. Let the optimal set of stocks be ${\cal A}_3$. Recall that since the consensus rule selected the intersection, note that ${\cal A}_1 \cap {\cal A}_2 \neq \phi$. We also assume that the truth ${\cal A}_3$ is equal to neither ${\cal A}_1$ nor ${\cal A}_2$, since LLM technology will inevitably make mistakes when screening. However, even though there are LLM mistakes, we assume that we have sensible selection, 
as will be defined in Section~\ref{subsec:sensible_screening}. Sensible selection allows only the two following cases:

{\bf  Case 1}.
The truth is
\[{\cal A}_3 \subset {\cal A}_1 \cap {\cal A}_2.\]

In other words, the optimal set is a subset of intersection of the two agents' choices. Let us analyze what would have  happened if we had used a single agent. If we only had used LLM-S screening, then our choice would have been ${\cal A}_1$, which is larger than the intersection, and hence larger than the optimal ${\cal A}_3 \subset {\cal A}_1 \cap {\cal A}_2$. The same reasoning applies if we had used FinBERT screening. In short, even though single agent selection here is consistent with sensible selection, it makes larger mistakes than the Agentic AI intersection rule. This can be seen in Figure \ref{fig:venn1}.

\begin{figure}[htbp]
\centering
\begin{tikzpicture}[scale=1.2]
    \def\circleA{(0,0) circle (2cm)}
    \def\circleB{(2.5,0) circle (2cm)}

    \draw[thick, fill=blue!20, fill opacity=0.5] \circleA;
    \draw[thick, fill=yellow!20, fill opacity=0.5] \circleB;

    \begin{scope}
        \clip \circleA;
        \fill[green!20, fill opacity=0.6] \circleB;
    \end{scope}

    \draw[thick, fill=green!60!black, fill opacity=0.7] (1.25, 0) circle (0.6cm);

    \node at (-1.5, -1.5) [left] {$A_1$};
    \node at (4.0, -1.5) [right] {$A_2$};
    \node at (1.25, 0) {\textbf{$A_3$}};

\end{tikzpicture}
\caption{Venn diagram showing $A_3 \subseteq (A_1 \cap A_2)$}
\label{fig:venn1}
\end{figure}

{\bf Case 2}. Assume now the truth is 
\[ {\cal A}_1 \cap {\cal A}_2 \subset {\cal A}_3\].

In other words, the intersection is a subset of the truth. There are four sub-possibilities. 

The first possibility is that ${\cal A}_3$ contains ${\cal A}_1 \cap {\cal A}_2$ (but no other elements from ${\cal A}_1$ and ${\cal A}_2$).
Assume, without loss of generality, that we use ${\cal A}_2$ instead of the intersection consensus rule. By selecting ${\cal A}_2$, the screening will not be sensible (as will be defined in Section~\ref{subsec:sensible_screening}), since our choice ${\cal A}_2$ is not a subset of ${\cal A}_3$ nor is ${\cal A}_3$ a subset of ${\cal A}_2$.
But the intersection ${\cal A}_1 \cap {\cal A}_2$ here is a sensible screening choice. The same argument applies if we had chosen only ${\cal A}_1$.

The second possibility is that ${\cal A}_3$ consists of some elements in ${\cal A}_1$ in addition to ${\cal A}_1 \cap {\cal A}_2$. Then by using only LLM-S (${\cal A}_1$), we may do better than the Agentic-AI intersection rule. However, if we choose instead to use FinBERT, this will not be a sensible screening since ${\cal A}_2$ is not a subset of ${\cal A}_3$, nor is ${\cal A}_3$ is a subset of ${\cal A}_2$.  Since it is not clear a priori which agent to choose, the safest bet is the intersection rule.  

The third possibility is that ${\cal A}_3$ consists of some elements in ${\cal A}_2$ in addition to ${\cal A}_1 \cap {\cal A}_2$. The analysis is the same as above: using FinBERT may achieve a better screening than Agentic AI, but LLM-S will not give a sensible screening. A safer bet is the intersection.

The last possibility is that ${\cal A}_3$ is much larger than ${\cal A}_1 \cup {\cal A}_2$. In this scenario, any single-agent rule will dominate Agentic AI, due to a larger deviation from the optimal ${\cal A}_3$. We see this issue in Figure \ref{fig:venn3}.

\begin{figure}[h]
  \centering
  \begin{tikzpicture}

    \draw[blue!60, thick, dashed, rounded corners=12pt]
      (-4.5, -3) rectangle (4.5, 3);
    \node[blue!70, font=\bfseries] at (-3.8, 2.5) {$A_3$};
    \node[blue!50, font=\small] at (-1.8, 2.5)
      {\textit{(supersedes $A_1$ and $A_2$)}};

    \begin{scope}
      \clip (-1.2, 0) circle (1.8);
      \fill[blue!30, opacity=0.6] (-1.2, 0) circle (1.8);
    \end{scope}
    \begin{scope}
      \clip (1.2, 0) circle (1.8);
      \fill[teal!40, opacity=0.6] (1.2, 0) circle (1.8);
    \end{scope}

    \begin{scope}
      \clip (-1.2, 0) circle (1.8);
      \fill[teal!60, opacity=0.7] (1.2, 0) circle (1.8);
    \end{scope}

    \draw[blue!70, thick]  (-1.2, 0) circle (1.8);
    \draw[teal!70, thick]  ( 1.2, 0) circle (1.8);

    \node[font=\bfseries, blue!80]  at (-2.2, 0)   {$A_1$};
    \node[font=\bfseries, teal!80]  at ( 2.2, 0)   {$A_2$};
    \node[font=\small]              at ( 0,  -0.2) {$A_1 \cap A_2$};

  \end{tikzpicture}
  \caption{Sets $A_1$ and $A_2$ share an intersection $A_1 \cap A_2$,
           and are both contained within the superseding set $A_3$.}
  \label{fig:venn3}
\end{figure}

\subsubsection{Union Rule}

Now we analyze the consensus rule when it chooses the union (i.e., when the intersection is empty). We use the same notation as in the previous subsection. As in the previous subsection, we assume that all screenings must be sensible.

{\bf Case 1}. The truth is

\[{\cal A}_3 \subset {\cal A}_1 \cup {\cal A}_2.\]

There are three sub-possibilities. 

First, ${\cal A}_3$ is entirely contained within ${\cal A}_1$. If we use LLM-S and select ${\cal A}_1$, the error will be smaller compared to ${\cal A}_1 \cup {\cal A}_2$. However, if we use FinBERT and select ${\cal A}_2$, it is not a sensible screening, since none of the optimal set is in the FinBERT choice of ${\cal A}_2$. Since it is not clear a priori which agent to choose, the safest choice is the union.

The second possibility is ${\cal A}_3$ is entirely contained within ${\cal A}_2$. The same argument as above holds.

The last possibility is when the truth includes a part of ${\cal A}_1$ and a part of ${\cal A}_2$. If we use any single agent (either ${\cal A}_1$ or ${\cal A}_2$), it is not a sensible screening, since ${\cal A}_1$ will exclude parts of ${\cal A}_3$ that are in ${\cal A}_2$, and likewise with ${\cal A}_2$. 
However, the union of the two sets will always include the correct set of stocks and hence be a sensible screening.

{\bf Case 2}. The truth is

\[{\cal A}_1 \cup {\cal A}_2 \subset {\cal A}_3.\]

Here, it is clear that any single agent will achieve a larger error compared to the union of both. We see this in Figure
\ref{fig:venn-diagram2}.

\begin{figure}[htbp]
    \centering
    \begin{tikzpicture}[thick, every node/.style={text=black, font=\large}]
    
    
        \draw[fill=gray!15, draw=gray!80!black, line width=1.5pt] (0,0) ellipse (4.5cm and 2.8cm);
        \node at (0, 2.2) {\Huge $A_3$};
    
        \draw[blue, fill=blue!10, line width=1.5pt] (-2, -0.2) circle (1.5cm);
        \node[blue] at (-2, -0.2) {\Huge $A_1$};
    
        \draw[orange, fill=orange!10, line width=1.5pt] (2, -0.2) circle (1.5cm);
        \node[orange] at (2, -0.2) {\Huge $A_2$};
    
    \end{tikzpicture}
    \caption{Venn diagram showing $A_1 \cap A_2 = \emptyset$ where both $A_1$ and $A_2$ are subsets of $A_3$.}
    \label{fig:venn-diagram2}
\end{figure}

\subsubsection{Summary of Mechanism Results}

In sum, in all cases, except one sub-possibility, our two-agent intersection/union rule is a better choice than a single-agent rule. To address that sub-possibility in Figure \ref{fig:venn3}, if we have institutional knowledge that the optimal set of assets is very large, then we can instruct our system to use only the union rule instead. However, in general, the intersection/union rule is more likely to provide a sensible screening.

\section{Stock Screening and Weights}\label{stockscreeningg}

Having outlined the architecture of our multi-agent system, we now develop its core component in detail. We formalize the concept of stock screening, describe how our LLM agents implement it, present alternative screening benchmarks, and establish a theoretical result on the Sharpe ratio consistency of screened high-dimensional portfolios.

One argument for stock screening is the identification of potentially high-return assets within a large universe of stocks. Additionally, stock screening can help prevent behavioral biases, which are common among human analysts. The value of screening has been highlighted in several studies, starting with the foundational work of \cite{ff1992}, who demonstrate that screening for small-cap and low price-to-book stocks leads to higher returns in the twentieth century. \cite{jt1993} validate the idea of ``price-performance" screening and consider the ``momentum" effect, showing that stocks performing well over the past 3-12 months tend to continue performing well in the near future. \cite{p2000} introduces F-scores: nine accounting ratios to find high-quality stocks. \cite{nm2013} adds a quality-gross profitability criterion for value stocks.

In a pure economic-theory analysis of portfolio choice, a benchmark article by \cite{vv2010} show why investors can choose a screened portfolio compared to a pure diversified-non-screened portfolio. They analyze information choice and the investment problem jointly. The information acquisition about assets is key to understanding the portfolio choice. They find that when the portfolio choice is preceded by an information choice-acquisition, the fully diversified portfolio is no longer optimal. With information acquisition investors resolve uncertainty about assets earlier than the payoffs are realized. With more substantial information about a subset of asset universe, investors are likely to hold them compared to a fully diversified portfolio. Specifically, they solve investor's utility maximization in four different general cases. They show sub-optimality of the fully diversified portfolio. Namely, the investors use information on a group of risky assets to reduce uncertainty and deviate from a fully diversified-non-screened portfolio.

In sum, screening generates higher returns and/or higher Sharpe ratio portfolios for several reasons. First, due to the limited cognitive capacity of human beings, technical analysis may be more effective on a subset of stocks compared to, for example, a larger universe of a thousand stocks. Second, since investors are susceptible to narrative bias, overweighting compelling stories attached to assets rather than their fundamentals \citep{kt1974, kt1983}, a systematic, quantitative screening approach can help identify and exclude stocks driven primarily by sentiment rather than intrinsic value.

Further, financially distressed firms can be systematically identified by analyzing debt and liquidity ratios and avoiding buying underperforming firms or those close to bankruptcy by using the Altman Z-score (see \cite{a1968}). Eliminating or shorting distressed firms from the portfolio is another avenue to higher returns.

\subsection{Financial Theory Behind Screening and Agentic AI}

Here in this subsection we utilize the paper by \cite{vv2010} to show why an under-diversified, screened portfolio may be a better choice than a fully diversified portfolio. After that we expand that joint analysis of information acquisition with portfolio choice to include Agentic-AI based screening through the paper-ideas
by \cite{kl2026}.

To begin with the analysis of \cite{vv2010}, both information acquisition and investment will be choices. We start, as they do, with Constant Absolute Risk Aversion (CARA) utility with additive information acquisition technology. In their paper, there are other utilities and with entropy based technology acquisition, we briefly discuss them at the end of this subsection. Information is acquired to reduce the conditional variance of asset payoffs. 

\cite{vv2010} uses a model with three periods. Below, we follow their analysis, and notation. In period one, the investor chooses information (signals) to buy subject to a cost. In period two, investor observes signals, and chooses weights for assets to buy. In period three, investor receives payoffs, and has realized the utility of information acquisition with portfolio formation. We introduce the notation. Investor's payoffs are $f$. Investor believes payoffs, $f$, follow a normal distribution with  $N (\mu, \Sigma)$, with $\mu$ being the return vector and $\Sigma$ is the covariance of payoffs. In period one, investor chooses a vector of signals $\eta$ about payoffs, with $\eta = f + e_{\eta}$, and $e_{\eta}$ is distributed with normally $N(0, \Sigma_{\eta})$ with zero mean and covariance of $\Sigma_{\eta}$. In period two, investor combines signal $\eta$ and his prior belief $\mu$ by Bayes' law. Let $\hat{\mu},\hat{\Sigma}$ be the posterior mean and covariance of payoffs conditional on all information on period two. 
\begin{equation}\label{ft-1}
    \hat{\mu}= E [ f / \mu, \eta]= (\Sigma^{-1} + \Sigma_{\eta}^{-1})^{-1}
 ( \Sigma^{-1} \mu + \Sigma_{\eta}^{-1} \eta).
 \end{equation}
\begin{equation}\label{ft-2}
   \hat{\Sigma}= V (f/ \mu, \eta) = (\Sigma^{-1} + \Sigma_{\eta}^{-1})^{-1}.
\end{equation}

Investor forms his portfolio $w_q$ after observing signals $\hat{\mu},\hat{\Sigma}$. Investor faces a budget constraint in choosing portfolio
\begin{equation}\label{ft-3}
    W = W_0 r + w_q' ( f - p r),
\end{equation}
where $W$ is wealth, $W_0$ is the  wealth portioned to risk free asset with return $r$, $p$ is a vector of asset prices, and $w_q'(f-pr)$ the portfolio excess return. As a simplification in their technical analysis, \cite{vv2010} assume that $\Sigma$ is diagonal, a general $\Sigma$ is possible but will have a different interpretation.

There are two choices that investor has to make, first how to allocate information among assets, and then how to weight each asset in a portfolio. Note that these two choices are $\eta, w_q$. The information acquisition choice is simplified in \cite{vv2010}. Hence to choose $\eta$ investor has to choose $\Sigma_{\eta}$, but to simplify notation-choices p.783, Section 1.4 of \cite{vv2010} uses $\hat{\Sigma}$ as the choice variable for information acquisition, and this is a direct result of (\ref{ft-2}). To be specific, this choice is explained by a unique link between covariance of signals, $\Sigma_{\eta}$ with posterior belief covariance matrix, $\hat{\Sigma}$ through (\ref{ft-2}).
Also by choosing $\hat{\Sigma}$, investor chooses the precision matrix about assets $\hat{\Sigma}^{-1}$. So, in sum, the investor will choose precision matrix, which represents information acquisition, in period one, $\hat{\Sigma}^{-1}$, and in period two, investor chooses the quantities to invest $w_q$.

First period utility is linear and $u_1(x)=x$, and the second period utility is $u_2 (W)= - \exp (-\rho W)$, where $\rho $ is the absolute risk aversion constant.  Total utility is represented as
\begin{equation}
    U = - E_1 (\exp (-\rho W)),\label{ft-4}
\end{equation}
where $E_1(.)$ represents the first period expectation.  We see that  (15) of \cite{vv2010} simplifies the utility to be, after substituting the optimal quantity as in (14) of \cite{vv2010} and using the budget constraint (\ref{ft-3}),  the utility only in terms of the precision (information choice) 
\begin{equation}\label{ft-5}
U = \sum_{i=1}^N \log [ \hat{\Sigma}_{ii}^{-1}/\Sigma_{ii}^{-1}],
\end{equation}
where  $\hat{\Sigma}_{ii}^{-1}$ is the $i$-th diagonal precision (posterior) matrix term and $\Sigma_{ii}^{-1}$ is the $i$ th  diagonal precision (prior) term. Note that the utility shows the improvement of the precision of $i$ th asset, and then these  are summed across all assets. Note that both prior and posterior precision matrix is diagonal since $\Sigma, \hat{\Sigma}$ are diagonal. The aim is to find the optimal $\hat{\Sigma}_{ii}^{-1}$ since this will show the optimal information acquisition for $i$ th asset, $i=1,\cdots, N$, where $N $ represents the total number of assets. Since there is a cost to acquire information, the utility is maximized subject to the following capacity $K$ constraint
\begin{equation}\label{ft-6}
    \sum_{i=1}^N (\hat{\Sigma}_{ii}^{-1} - \Sigma_{ii}^{-1}) \le K 
\end{equation}
Also, there is  the no-forgetting constraint
\begin{equation}
\Sigma_{ii}^{-1} \le \hat{\Sigma}_{ii}^{-1} < \infty.\label{ft-7}
    \end{equation}

    To understand the capacity constraint, note that if there is an equality (binding constraint), with zero capacity $K=0$, all assets prior and posterior information-precision are equal, so there is no improvement. However, when $K$ is positive, we can see an improvement of posterior precision over prior for some assets.

    Note that capacity constraint maps prior and posterior precision matrices into capacity, with (\ref{ft-7}), the learning technology makes the posterior precision (variance) higher (lower) than equal to prior precision (variance). To find the optimal information, measured by $\hat{\Sigma}_{ii}^{-1}$, the investor maximizes the utility in (\ref{ft-5}) subject to a capacity constraint (\ref{ft-6}) and the no-forgetting constraint in (\ref{ft-7}). Also see that both prior and posterior precision matrices are non-random as described in Section 1.4 of \cite{vv2010}.  The aim is to choose $\hat{\Sigma}_{ii}^{-1}$ for $i=1,\cdots, N$ to maximize the constrained objective function below with Lagrange multipliers $\psi, \lambda_i $, 
    \[ L =  \sum_{i=1}^N ln [ \hat{\Sigma}_{ii}^{-1}/\Sigma_{ii}^{-1}]
    + \psi [K - \sum_{i=1}^N (\hat{\Sigma}_{ii}^{-1} - \Sigma_{ii}^{-1})] + \sum_{i=1}^N \lambda_i (\hat{\Sigma}_{ii}^{-1} - \Sigma_{ii}^{-1})
    \]
The solution is 
\begin{equation}
    \hat{\Sigma}_{ii,*}^{-1} = max (\bar{\sigma}^{-1}, \Sigma_{ii}^{-1}),\label{sol1}
\end{equation}

where $\bar{\sigma}^{-1}$ is the solution to the following equation
\begin{equation}
 \sum_{i=1}^N max (\bar{\sigma}^{-1}, \Sigma_{ii}^{-1}) = K.\label{sol2}
 \end{equation}

This solution is obtained when assets have different means and variances. These last two equations-solutions form Corollary 1 of \cite{vv2010}. 
It is clear from the solution that investor learns more about the assets that the investor is less certain about. 
The number of assets learned weakly increases with an increase in capacity K. Figure 2 of \cite{vv2010} illustrates a two-asset problem, with expected utility and optimal learning choices (i.e. $\hat{\Sigma}_{jj,*}^{-1}$) for each $j=1,2$. This figure is reproduced in our paper as Figure \ref{fig:cara-additive} since we will link that to Agentic AI below. In Figure  \ref{fig:cara-additive}, the investor is more uncertain about Asset 1, with prior precision set at 0.4, and Asset 2 precision is set at 0.5. The utility curve is  touching the linear capacity constraint at an interior point. The posterior optimal precision, the black dot,  is 0.75 for both assets in our version in Figure \ref{fig:cara-additive}. So compared to prior precision of 0.4, Asset 1 optimal posterior precision is 0.75 and improved more than Asset 2  between prior-posterior precision, which improved 0,25 points between 0.5 to 0.75. The dashed vertical line shows the prior precision of Asset 1, at 0.4, and the horizontal dashed line shows the prior precision of Asset 2 at 0.5.

\begin{figure}[htbp]
    \centering
    \begin{tikzpicture}
    \begin{axis}[
        width=0.7\textwidth, height=0.7\textwidth,
        xmin=0, xmax=2, ymin=0, ymax=2,
        xlabel={Asset 1 posterior precision},
        ylabel={Asset 2 posterior precision},
        axis lines=box,
        xtick={0,0.5,1,1.5,2},
        ytick={0,0.5,1,1.5,2},
        legend style={font=\small, at={(0.97,0.97)}, anchor=north east,
                       draw=black, fill=white},
        legend cell align={left},
    ]

    \addplot[fill=gray!35, draw=none]
        coordinates {(0.4,0.5) (1.0,0.5) (0.4,1.1)} -- cycle;
    \addlegendentry{feasible choice set}

    \addplot[dashed, thick, black] coordinates {(0.4,0) (0.4,2)};
    \addlegendentry{prior variance}
    \addplot[dashed, thick, black, forget plot] coordinates {(0,0.5) (2,0.5)};

    \addplot[dashdotted, thick, black, domain=0.25:1.25] {1.5-x};
    \addlegendentry{additive precision constraint}

    \addplot[thick, black, domain=0.4:2, samples=200,
             restrict y to domain=0:2] ({x},{0.5625/x});
    \addlegendentry{CARA indifference curve}

    \addplot[only marks, mark=*, mark size=2.2pt, black, forget plot]
        coordinates {(0.75,0.75)};

    \end{axis}
    \end{tikzpicture}
    \caption{Indifference curve and information constraints for CARA
    preferences with the additive precision technology. The shaded area
    represents feasible information choices satisfying the capacity and
    no-negative learning constraints; the dot marks the optimal (interior)
    information choice.}
    \label{fig:cara-additive}
\end{figure}

Now we want to tie the optimal precision, $\hat{\Sigma}_{ii,*}^{-1} $to optimal quantity that can be invested to form a portfolio.  Let $\hat{\Sigma}_*^{-1}$ be the $N \times N$ precision matrix for all assets, where the $\hat{\Sigma}_{ii,*}^{-1}$ represents the $i$ th  asset's optimal precision. The expected optimal quantity to invest is given by (17) of \cite{vv2010} which is \footnote{ These optimal weights are corresponding to learning portfolio in \cite{vv2010}.}
\begin{equation}
    E w_q = \frac{1}{\rho} (\hat{\Sigma}_*^{-1} - \Sigma^{-1}) ( \mu - p r).\label{oq}
\end{equation}

Note that both $\hat{\Sigma}_{*}^{-1}, \Sigma^{-1}$ is diagonal. So clearly, accounting the idea in example in Figure \ref{fig:cara-additive} as well, investor holds more of the initial risky that improves its precision through information acquisition. This is called generalized learning. 
In our example in Figure \ref{fig:cara-additive}, since Asset 1 improved more that Asset 2 in precision, the investor will hold more of the risky asset Asset 1,  as long as Asset 1 has a positive expected return.
So generalized learning results in riskier portfolios. Generalist investor reduces uncertainty through information acquisition. Investor uses the information to buy assets with higher payoffs, and sell assets that have low payoffs. With larger capacity K, the investor can increase the returns by uncertainty reduction through information acquisition. Rather than a simple fully diversified, non-optimal portfolio, a screened portfolio is optimal through initial information acquisition about assets. 

\cite{vv2010} also tried mean-variance preference with entropy based technology-capacity constraint, as well as Constant Relative Risk Aversion  (CRRA) utility with entropy based technology. All these result even in much selective one asset corner solution. The asset chosen is the one with the highest Sharpe Ratio. The convex nature of the objective function results in one asset solution. There may be a knife-edge case of investor following CARA utility with entropy based capacity constraint, where the investor is indifferent among all asset choices.

\subsubsection{Financial Theory behind Agentic AI  }

Note that in Section 2.2 of \cite{kl2026}, the authors claim  investors can approximate or expand information set better with AI based systems. This proclamation also overlaps with Agentic AI based systems  larger information set reach, better coordination compared with single agent systems as shown in Section 2, and our empirics results. Also compared to human beings, simple Machine learning techniques, the reach of Agentic AI is much broader in terms of information. Our paper makes the suggestion that better approximation of an information set by Agentic AI can act as capacity, $K$ increase in (\ref{ft-6}).

Now we increase capacity compared to Figure \ref{fig:cara-additive} and call it Figure 5. This results in a higher utility, higher precision, possible larger number of riskier assets that can provide the investor larger returns.  Agentic-AI provides larger capacity, hence a higher utility compared to limited capacity approaches.  With the new larger capacity provided by Agentic AI, the investor can discover more assets that can provide higher return potential compared to a smaller capacity based approaches. 

Now let us give a simple example comparing humans with Agentic AI. First we start with the premise that $K_{Agentic AI} > K_{Human}$. 
Then assume that humans improve their information only in 2 assets out of N assets, with their capacity $K_{Human}$, and without losing any generality, these are assets 1 and 2. The optimal precision for assets 1 and 2 is $\bar{\sigma}_h^{-1}$. 
\[ \bar{\sigma}_h^{-1} + \bar{\sigma}_h^{-1} + \Sigma_{3,3}^{-1}+\cdots + \Sigma_{N,N}^{-1} = K_{human}.\]
So for humans $\hat{\Sigma}_{11,*}^{-1}= \hat{\Sigma}_{22,*}^{-1} = \bar{\sigma}_h^{-1}$.

For Agentic AI, with larger capacity, there may be an increase in precision of assets 1 and 2, or choice of more assets either with same precision or higher precision than human precision solution. Without losing any generality in our argument, the new precision solution is higher for assets 1-4, and $\bar{\sigma}_a^{-1} > \bar{\sigma}_h^{-1}$ due to a larger capacity constraint in Agentic AI,
$K_{Agentic AI}> K_{Human}$. This can be seen in (\ref{sol1})(\ref{sol2}).

Our binding capacity constraint is:
\[ \bar{\sigma}_a^{-1} + \bar{\sigma}_a^{-1} + \bar{\sigma}_a^{-1} + \bar{\sigma}_a^{-1}+ \Sigma_{5,5}^{-1}+\cdots + \Sigma_{N,N}^{-1} = K_{Agentic AI}.\]
So for Agentic AI, the optimal solution for information for assets 1-4
\[\hat{\Sigma}_{11,*}^{-1}= \hat{\Sigma}_{22,*}^{-1} = \hat{\Sigma}_{33,*}^{-1}= \hat{\Sigma}_{44,*}^{-1}= \bar{\sigma}_a^{-1} \]

Using this optimal information solution in (\ref{ft-5}) provides clearly higher utility for Agentic AI based investor through (\ref{ft-5}) compared to human based solution.

\begin{figure}

\begin{center}
\begin{tikzpicture}[x=7cm, y=7cm, every node/.style={font=\scriptsize}]
\draw[thick,->] (0,0) -- (2.0,0) node[below, font=\tiny] {Asset 1 posterior precision};
\draw[thick,->] (0,0) -- (0,2.0) node[above, font=\tiny, align=center] {Asset 2 posterior precision};

\fill[ncsured!12] (0.4,1.1) -- (0.4,1.4) -- (1.3,0.5) -- (1.0,0.5) -- cycle;
\fill[agentblue!20] (0.4,0.5) -- (0.4,1.1) -- (1.0,0.5) -- cycle;

\draw[dashed, gray] (0.4,0) -- (0.4,1.45);
\node[below, font=\tiny] at (0.4,0) {$0.4$};
\draw[dashed, gray] (0,0.5) -- (1.35,0.5);
\node[left, font=\tiny] at (0,0.5) {$0.5$};

\draw[dash dot, thick, agentblue] (0.15,1.35) -- (1.40,0.10);
\node[agentblue, font=\tiny, rotate=-45] at (1.22,0.17) {$K_{\text{human}}$};

\draw[dash dot, thick, ncsured] (0.20,1.60) -- (1.65,0.15);
\node[ncsured, font=\tiny, rotate=-45] at (1.47,0.22) {$K_{\text{Agentic AI}}$};

\draw[thick, agentblue, domain=0.30:1.85, smooth, variable=\myx, samples=60]
  plot ({\myx},{0.5625/\myx});

\draw[thick, ncsured, domain=0.45:1.85, smooth, variable=\myx, samples=60]
  plot ({\myx},{0.81/\myx});

\filldraw[agentblue] (0.75,0.75) circle (1.4pt);
\node[agentblue, font=\tiny, align=center] at (0.62,0.28) {human optimum\\ $(0.75,\,0.75)$};
\draw[-{Latex[length=1.3mm]}, agentblue] (0.65,0.38) -- (0.74,0.68);

\filldraw[ncsured] (0.9,0.9) circle (1.4pt);
\node[ncsured, font=\tiny, align=center] at (1.42,1.55) {Agentic AI optimum\\ $(0.9,\,0.9)$};
\draw[-{Latex[length=1.3mm]}, ncsured] (1.30,1.42) -- (0.95,0.97);
\end{tikzpicture}
\end{center}
\vspace{-0.2cm}
{\scriptsize  Figure 5. Human capacity $K_{\text{human}}$ ($x+y\le1.5$), utility curve, and optimum $(0.75,0.75)$ -- identical to Fig.\ 4.  Agentic AI capacity $K_{\text{Agentic AI}}>K_{\text{human}}$ ($x+y\le1.8$), shifted north-east; a higher utility curve is tangent to it at the new optimum $(0.9,0.9)$ -- higher precision on \emph{both} assets, higher utility, higher expected holdings $Ew_q$.}
\label{fig:agentic-pic}
\end{figure}

\subsection{LLM-S and FinBERT}\label{subsec:sensible_screening}

Our distinguishing approach compared to the rest of the literature is that we  provide a multi-agent  AI system that can screen stocks across a broader universe. 


We start by developing our own LLM agent, tailored to analyze firm fundamental data. We call this agent LLM-S, abbreviated for LLM-strategy. Using standard AI agent construction techniques, LLM-S produces a fundamentals-based scoring rule to screen stocks. The objective is to filter out firms that have undesirable fundamentals, leaving behind only high-quality stocks in the portfolio. We utilize LLMs in a novel way: we prompt the LLM to adopt the persona of a portfolio manager, and we ultimately let the LLM decide upon a specific buy and sell strategy based on fundamentals data. In doing so, we gain insight into the reasoning behind its chosen strategy. Ultimately, LLM-S is asked to justify its screening choices.

We next describe the main steps of our LLM-S agent. First, LLM-S uses three factors to initiate the algorithm: log firm size (mve), book-to-market ratio (bm), and 12-month momentum (mom12m). We select these specific factors because they represent the most robust and widely accepted predictors of cross-sectional stock returns in the empirical finance literature, capturing the size of the firm, its valuation, and its recent performance, respectively (see \cite{ff1992} and \cite{carhart1997persistence}). Furthermore, our preliminary testing revealed that using this parsimonious set of factors is crucial for the LLM. Expanding the factor set to include the Fama-French 5-factor variables plus momentum empirically degraded the model's performance. Therefore, it is critical to model the right LLM agent rather than simply using an LLM. This is very much related to the architecture of LLM agents and the aim here is to find agents that add economic or financial value to the firm. This architectural design is crucial to address criticism in the practical financial world regarding the value of LLMs (see Footnote 1 in Introduction).
The LLM-S is instructed to execute the following four steps on its own:

\begin{enumerate}
    \item  Explore the data to understand extreme values in factors, clustering or breakpoints, and correlations between factors.

\item Develop clear rules based on economic intuition.

\item Define specific thresholds for buy and sell.

\item Provide rationale for decisions.

\end{enumerate}

Each of these points is extensively detailed in Appendix \ref{additional material}, where a snippet of our code is also shown.  These above four rules can be thought of as the following pipeline: analyzing data, developing rules for buying and selling, deciding on exact thresholds, and finally, explaining its decision rationale. By using size, book-to-market, and momentum as features, our LLM-S filters the universe of stocks to retain only those exhibiting characteristics historically associated with risk premia. The strategy also inherently acts as a signal enhancer: by isolating firms with the most extreme signals, LLM-S constructs a subset of firms with high factor exposure. This provides the quantitative estimation method a highly informative candidate pool to optimize upon.

Our LLM-S deviates from traditional screening methods in a few fundamental ways. Firstly, traditional methods often rely on fixed heuristic rules that are static over time (see, for example, \cite{p2000} and \cite{mohanram2005separating}). In contrast, our LLM-S is rerun once every year to ensure the most recent  rules, and it has the ability to choose exact thresholds (i.e. it has the ability to contextualize the rule based on the distribution of the data for a specific date). 
While previous methods mostly rely on backward-looking training to determine scoring rules (see, for example, \cite{a1968} and \cite{asness2019quality}), we leverage the broad knowledge embedded in pre-trained modern LLMs. Because these models have already learned extensive financial principles from a massive set of data, we can use a ``zero-shot" approach. That is, the model generates the scoring rule directly from its internal knowledge without requiring any training. Lastly, the LLM provides the scoring rule \textit{plus} the economic intuition behind it, offering a greater level of interpretability.

As an example, our LLM-S provided the following strategy and rationale for buy and sell decisions in year 2024 (more details regarding reasoning of this choice are in Appendix \ref{additional material}). 

- BUY Rule: This rule targets undervalued (high bm), reasonably sized ($mve > 0.3$) companies with positive momentum ($mom12m > -0.5$).\footnote{Note the mve, bm, and mom12m features are standardized to have mean 0 and variance 1.} The economic intuition is to buy companies that are currently cheap but have shown some signs of recovery or positive market sentiment.

- SELL Rule: This rule aims to sell companies that are overvalued (low bm), have negative momentum ($mom12m < -0.55$), or are small in size ($mve < -0.75$). 


Note that buy and sell rules underscore the presence of interaction effects within financial markets, aligning with the theoretical framework established in Section 3.6 of \cite{kx2023}. 
In our example above, our buy decision is reminiscent of the interaction term between high book-to-market, reasonably sized, and positive momentum firms. Our sell decision can be similarly viewed as an interaction term. In sum, our buy and sell decisions can be seen as the decisions from a nonlinear framework involving factor interactions.
Our results can also be related  to an information theoretic approach by \cite{h2000}, who demonstrates that the momentum effect interacts with size of the firm and analyst coverage.

After introducing LLM-S, to handle human biases and account for short-term news, we utilize FinBERT, a natural language processing model specialized in analyzing financial text, to conduct sentiment analysis on recent stock news. The FinBERT agent also makes buy-sell decisions. 
After discussing their buy and sell decisions, two agents decide on a consensus of buy-sell decisions.



Once these LLMs decide which stocks to buy and sell, we feed the screened stocks into  quantitative high-dimensional weight formation techniques (see \cite{caner2019},  \cite{caner2022} using nodewise or residual nodewise regression, respectively). These methods will be described in detail in subsequent sections. In short, we combine LLM- and sentiment analysis-based screening with state-of-the-art high-dimensional quantitative portfolio formation techniques.

We also connect our screening idea econometrically to high-dimensional portfolio analysis. The screening process can be conceptualized not just as selecting stocks that can provide higher returns or Sharpe ratios, but also as choosing the correct number of stocks among a large universe of them. 
Clearly, this screening choice is a random variable and affects how we choose the weights in our portfolio.  
 First, let us denote the number of stocks chosen by the screening process as the random variable $\hat{p}$. Also, denote by  $p^*$ the target (optimal) number of stocks, which is non-random, and $p\ge  p^*\ge 1$, where $p$ is the size of the entire universe of stocks. It is worth emphasizing that, together with the number of stocks, we have the stocks themselves; for example, the optimal choice may be Nvidia and Palantir; hence giving us $p^*=2$. We have the following definition of the screening process. Let ${\cal U}$ represent the total universe of stocks.\\

{\bf Definition: Sensible Screening.} Let $\mathcal{S}^* \subseteq \mathcal{U}$ denote the optimal portfolio, with $|\mathcal{S}^*| = p^*$, and let $\hat{\mathcal{S}} \subseteq \mathcal{U}$ denote the set of stocks selected by the screening process, with $|\hat{\mathcal{S}}| = \hat{p}$. A screening process is \emph{sensible} if:
\begin{enumerate}[(i)]
    \item when $\hat{p} \ge p^*$, the optimal set is recovered within the screened set: $\mathcal{S}^* \subseteq \hat{\mathcal{S}}$;
    \item when $\hat{p} < p^*$, the screened set is entirely composed of optimal stocks: $\hat{\mathcal{S}} \subset \mathcal{S}^*$.
\end{enumerate}
A sensible screening process never selects a stock that does not belong to the optimal portfolio in (ii), and always includes the optimal portfolio when it selects enough stocks in (i). For example, suppose $p^* = 2$ and the optimal portfolio is $\mathcal{S}^* = \{\text{Nvidia, Palantir}\}$. If $\hat{p} = 3$, sensible screening requires that $\hat{\mathcal{S}} = \{\text{Nvidia, Palantir, X}\}$ for some third stock $\text{X} \notin \mathcal{S}^*$, i.e., the two optimal stocks are always included. If instead $\hat{p} = 1$, sensible screening requires that the single selected stock be either Nvidia or Palantir.

The main idea of our theoretical analysis is to treat the number of stocks in the portfolio as a random variable, one that is realized through the screening process before any weights are estimated. Under sensible screening, and robust to mild screening errors, we show that the Sharpe ratio of the screened portfolio consistently estimates the target Sharpe ratio of the optimal stock universe. This result appears to be new in the high-dimensional portfolio literature, where the composition of the investable set is typically taken as given. We next turn to the screening procedures used in our empirical analysis.

\subsection{Different Screening Methods}

As a benchmark, we also consider alternative screening methods. We  consider stock picking by human analysts, and  a logistic regression-based stock picking method. In the first method, we use the buy/sell decisions of human analysts from I/B/E/S (Institutional Brokers Estimates System, accessed via WRDS) recommendations, and in the second one, we use logistic regression to generate buy/sell signals as proposed by \cite{ckx2023}. The former also represents conventional human judgment, and the latter is a simple quantitative strategy for stock picking. Our benchmarking strategy is designed to disentangle three sources of potential performance gains: stock screening per se, the use of LLMs for screening, and the benefits of a multi-agent architecture. To this end, we compare our multi-agent system against four classes of alternatives. First, we benchmark against conventional screening methods, namely momentum and value screens, to assess whether any performance improvement stems from screening itself rather than from the specific screening technology. Second, we compare against single-agent LLM screeners, such as LLM-S alone or FinBERT combined with the quantitative method, to isolate the contribution of the multi-agent architecture. Third, we consider alternative multi-agent configurations, including hybrid systems that combine LLM-S, FinBERT, and human judgment screening, with portfolio decisions delegated to a quantitative optimization method. Finally, we compare against the quantitative-method-only baseline, which abstracts from screening entirely, to quantify the marginal value of LLM-based stock selection.


\subsection{Algorithm}

Our algorithm screens $p$ stocks and subsequently feeds the results to a quantitative precision matrix estimation method, which determines the portfolio weights. Our approach proceeds as follows:

\begin{enumerate}
    \item The LLM-S agent and the FinBERT agent independently decide which stocks to buy/sell. We rerun the LLM-S agent each year, and the FinBERT agent each month. We do this to track the slow-moving economic narrative that the LLM-S aims to identify, while also preserving the fast-moving news that FinBERT captures. This mirrors classical portfolio literature (see \cite{ff1992}) and modern ML approaches (see \cite{gu2020empirical}) that refit computationally heavy models annually. 
    \item According to their decision rule (which will be discussed in Subsection \ref{method}), the models reach a consensus, if applicable, on a subset of $\hat{p}$ stocks to buy/sell, where $1 \le \hat{p} \le p$.
    \item These $\hat{p}$ stocks are provided to the quantitative method, which computes optimal portfolio weights based on a strategy described in Section \ref{quant strategies}. These depend on a combination of a method and objective.
    \item Repeat this process for all months. We use the annualized out-of-sample Sharpe ratios as the primary measure of performance. 
\end{enumerate}

We illustrate the algorithm in Figure \ref{fig:algorithm}.

\begin{figure}
    \centering
    \includegraphics[width=0.8\linewidth]{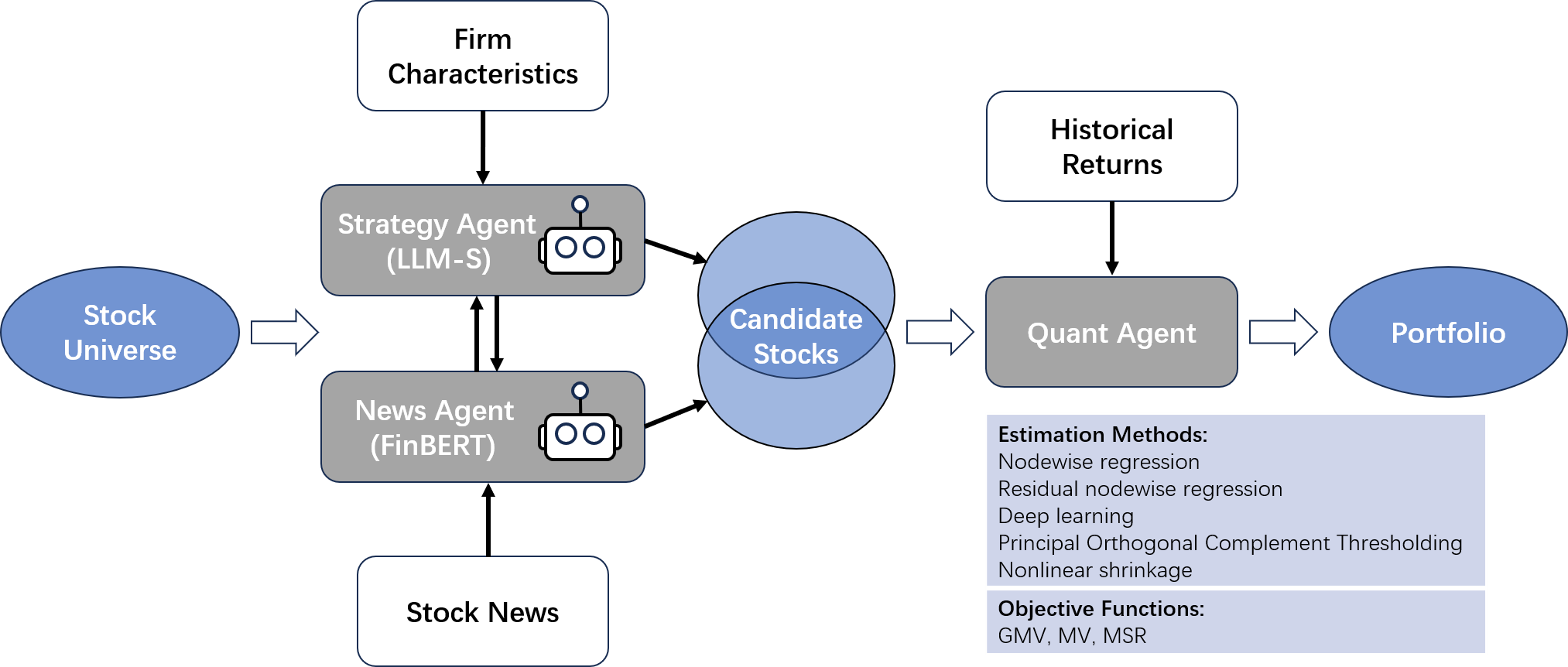}
    \caption{An illustration of the algorithm. LLM-S and FinBERT first screen for candidate stocks to invest in, and the precision matrix estimation technique assigns weights to these stocks.}
    \label{fig:algorithm}
\end{figure}

Denote the covariance matrix of outcomes as $\Sigma$.
As an illustration of the third step, consider the case where the portfolio objective is global minimum variance (GMV) and the precision matrix $\Gamma:=\Sigma^{-1} $,  which is the inverse of the covariance matrix of outcomes, is estimated via nodewise regression. The GMV weights, computed via the nodewise-GMV method, are then given by
\[ \hat{w}= \frac{\hat{\Gamma} 1_{\hat{p}}}{1_{\hat{p}}'\hat{\Gamma} 1_{\hat{p}}},\]
where $\hat{\Gamma}$ is the $\hat{p}\times\hat{p}$ estimated precision matrix of the screened stocks and $1_{\hat{p}}\in\mathbb{R}^{\hat{p}}$ is a vector of ones.

More details are described in Section \ref{quant strategies}. Note that the dimension of the precision matrix ($\hat{p}$) is a random variable, determined by the screening process prior to estimation. This stands in contrast to the existing high-dimensional portfolio literature, where the number of assets is treated as a deterministic sequence growing at a known rate; see, e.g., \cite{fan2011}, \cite{lw2017}, \cite{caner2022}, and \cite{cf2026}. 

\subsection{Screened Portfolio-Based Sharpe Ratio Analysis}

Let $\widehat{SR}_{\hat{p}}$ denote the Sharpe ratio estimator based on the screened portfolio of dimension $\hat{p}$, and let $SR_{p^*}$ denote the target Sharpe ratio of the optimally screened portfolio. In Theorem~\ref{thm:sr_convergence}, we establish that under sensible screening with mild errors (formalized in Assumption~\ref{assum1}),
\[
\left| \frac{\widehat{SR}_{\hat{p}}^2}{SR_{p^*}^2} - 1 \right| = o_p(1).
\]
That is, the squared Sharpe ratio of the screened GMV portfolio consistently estimates its target, even when the screening process is imperfect. To our knowledge, this is the first such result in the high-dimensional portfolio literature, where portfolio dimension is typically treated as fixed and known. The result holds for any sensible screening process satisfying Assumption~\ref{assum1}, and is therefore compatible with a broad class of precision matrix estimators, including those of \cite{fan2011} and \cite{caner2022}. The results will hold true in large portfolios as long as the precision matrix estimation of asset returns is consistent. Case-by-case analysis can be obtained under weaker conditions than the general-level assumptions in Appendix \ref{sec_A_proofs}.

\section{Quantitative-Based High-Dimensional Portfolio Weight Formation}\label{quant strategies}

One of the main empirical discoveries by \cite{cf2025} and \cite{caner2022} highlights the importance of jointly analyzing precision matrix estimation techniques with objective functions. The main insight is that investing in the optimal combination of a precision matrix estimation technique and an objective function yields the best performance. In all the quantitative-based techniques, we let the number of assets be $p$, and time series to be $n$. We index assets by $j=1,\dots,p$ and index time by $t=1,\dots,n$.

\subsection{Precision Matrix Estimation}
In this section, we discuss methods for high dimensional portfolio selection. We include brief descriptions of five different methodologies for precision matrix estimation, representing a variety of approaches ranging from statistical shrinkage to factor models to ML-based regression techniques. These methods address the challenge of estimating the precision matrix $\Gamma = \Sigma^{-1}$ when the number of assets exceeds the sample size. We use the estimated $\hat{\Gamma}$ of these methods as the precision matrix when making portfolio decisions in future sections.  In each of the methods below, we use a constant $p$ assets to determine the estimator for the precision matrix. In the appendix, we show how a screened universe of stocks changes the precision matrix estimation. A rigorous, in-depth description of each method can be found in Appendix \ref{sec:quant}.

\subsubsection{Nodewise Regression}\label{nodewise}

Introduced by \cite{mein2006} and applied to portfolio risk estimation by \cite{caner2019}, nodewise regression estimates the precision matrix directly via $p$ Lasso linear regressions. By modeling each asset's excess return as a linear combination of all other assets' excess returns (i.e. $y_{t,j} = y_{t,-j}' \gamma_j + \eta_{t,j}$), this method explicitly imposes sparsity on the rows of the precision matrix.
To account for high dimensions, the coefficients are estimated via a Lasso regression:$$\hat{\gamma}_j = \mathrm{argmin}_\gamma \left[ \frac{||y_j - Y_{-j} \gamma||_2^2}{n} +2\lambda_j ||\gamma||_1 \right],$$
where $y_j$ is the $n \times 1$ vector of returns for asset $j$, $Y_{-j}$ is the $n \times (p-1)$ matrix of returns for all other assets, and $||\cdot||_1$ and $||\cdot||_2$ denote the standard $\ell_1$ and $\ell_2$ norms. High-dimensional consistency is achieved by optimizing the penalty parameter $\lambda_j$ via a Generalized Information Criterion (GIC). The precision matrix can be constructed directly by using matrix algebra and defining  the diagonal elements $\hat{\Gamma}_{j,j} = \hat{\tau}_{j}^{-2}$ and the vector of off-diagonal elements for $j$-th row  $\hat{\Gamma}_{j,-j}=-\hat{\tau}_{j}^{-2} \hat{\gamma}_j'$, where $\hat{\tau}_j^2 = \frac{|| y_j - Y_{-j} \hat{\gamma}_j ||^2_2}{n} + \lambda_j || \hat{\gamma}_j ||_1$. We form each row $j$ of $\hat{\Gamma}$ using the main diagonal term $\hat{\Gamma}_{j,j}$ and the off-diagonal terms $\hat{\Gamma}_{j,-j}$. Stacking these rows yields $\hat{\Gamma}$, the nodewise estimator of the precision matrix.

\subsubsection{Residual Nodewise Regression}
Proposed by \cite{caner2022}, this approach extends standard nodewise regression by integrating factor models. The asset returns are modeled by $y_{t,j} = b_j' f_t + u_{t,j}$, where $f_t$ represents a $K \times 1$ vector of observable factors (e.g., the Fama-French three-factor model) and $b_j: K \times 1$ represents the factor loadings. Unlike standard nodewise regression, which assumes the precision matrix of returns is sparse, this method only assumes sparsity in the precision matrix of the unobserved idiosyncratic errors, $\Sigma_u^{-1}$. After the observable factor structures are removed via ordinary least squares, nodewise regression is applied to the residuals to estimate  $\Sigma_u^{-1}$ by $\hat{\Omega}$. The final precision matrix of returns is reconstructed analytically using the Sherman-Morrison-Woodbury formula:$$\hat{\Gamma} = \hat{\Omega} - \hat{\Omega} \hat{B} [\hat{\Sigma}_f^{-1} + \hat{B}' \hat{\Omega}_{sym} \hat{B}]^{-1} \hat{B}' \hat{\Omega},$$where $\hat{B}$ is the $p \times K$ matrix of OLS-estimated factor loadings, $\hat{\Sigma}_f$ is the sample covariance of the factors, and $\hat{\Omega}_{sym} = (\hat{\Omega} + \hat{\Omega}')/2$ is the symmetrized residual precision matrix estimate.

\subsubsection{Principal Orthogonal Complement Thresholding (POET)}
Developed by \cite{fan2013}, POET is designed for linear factor models ($y_{t,j} = b_j' f_t + u_{t,j}$) where the $K$ common factors $f_t$ are unobservable and must be estimated. POET first uses principal components analysis (PCA) to estimate the unobservable factors, then uses a thresholding method to estimate the covariance matrix of errors. Under the assumption that the covariance matrix of the remaining idiosyncratic errors ($\Sigma_u$) is sparse, POET applies a soft-thresholding technique to the error covariance matrix to eliminate spurious correlations, yielding $\hat{\Sigma}_{u,Th}$ (we leave technical details in Appendix \ref{sec:quant}). The final precision matrix is reconstructed from the thresholded error covariance and the estimated factor structure:$$\hat \Gamma=\hat{\Sigma}_{u,Th}^{-1}- \hat{\Sigma}_{u,Th}^{-1} \hat B(I_{K}+\hat B' \hat{\Sigma}_{u,Th}^{-1}\hat B)^{-1}\hat B' \hat{\Sigma}_{u,Th}^{-1},$$where $\hat{B}$ is the matrix of PCA-estimated factor loadings and $I_K$ is the $K \times K$ identity matrix.

\subsubsection{Deep Learning}
To capture complex, nonlinear relationships between asset returns and observable factors, \cite{c-d2025} introduce a deep learning-based estimator. The asset returns are modeled by a multi-layer neural network, $y_{t,j}=g_j(f_t)+u_{t,j}$, where $f_t$ is a $K$-dimensional observable column vector and $g_j(\cdot)$ is an unknown function. This effectively decomposes the total covariance matrix into a nonlinear function covariance component and an idiosyncratic error covariance component ($\Sigma_y = \Sigma_g + \Sigma_u$). Similar to POET, a targeted thresholding mechanism is applied to the error covariance matrix ($\hat{\Sigma}_{u,Th}$) to ensure stability. The final precision matrix is then recovered algebraically:$$\hat \Gamma =\hat \Sigma_{u,Th}^{-1}-\hat \Sigma_{u,Th}^{-1}\hat \Sigma_g(I_{K}+\hat \Sigma_{u,Th}^{-1}\hat \Sigma_g)^{-1}\hat \Sigma_{u,Th}^{-1},$$where $\hat{\Sigma}_g$ is the estimated covariance matrix of the neural network predictions. 

\subsubsection{Nonlinear Shrinkage (NLS)}
Introduced by \cite{lw2017} and  advanced by \cite{ledoit2020analytical}, Nonlinear Shrinkage (NLS) addresses the instability of the sample covariance matrix in high dimensions through modifying its eigenvalues. Starting with the spectral decomposition of the sample covariance matrix $S = U \Lambda U'$ (where $U$ is the orthogonal matrix of eigenvectors and $\Lambda$ is the diagonal matrix of sample eigenvalues), NLS systematically modifies the sample eigenvalues while preserving the eigenvectors. Utilizing the Hilbert transform of the sample spectral density (elaborated in Appendix \ref{sec:quant}), NLS derives an optimal, closed-form final estimator, which takes on the form:$$\hat{\Sigma} = U \hat{\Delta}^* U',$$ 
where $\hat{\Delta}^*$ is the diagonal matrix of optimally shrunk eigenvalues. Conceptually, this local shrinkage pulls dispersed sample eigenvalues toward each other, correcting the systemic over-dispersion inherent in high-dimensional settings and producing a well-conditioned precision matrix estimator. We then set $\hat{\Gamma} = \hat{\Sigma}^{-1}$.

\subsection{Objective Functions}

We introduce three objective functions that are heavily used in practice and provide the optimal weight of the portfolios attached to these functions. These are the global minimum variance portfolio (GMV), the Markowitz mean-variance portfolio (MV), and the maximum Sharpe ratio portfolio (MSR). 

The GMV weights are given by $$w^* := \mathrm{argmin}_{w \in \mathbb R^p} w' \Sigma w, \quad \text{such that } w' 1_p = 1. $$
The solution to the above is well-known and is given by $$w^* = \frac{\Gamma 1_p}{1_p' \Gamma 1_p}.$$

The GMV portfolio is designed to be very risk-averse, as it is designed to minimize the variance of a portfolio. It serves as an effective mechanism for generating modest returns with minimal risk. A different objective function that incorporates both return and variance would be the mean-variance portfolio, whose weights are given by

$$w^* := \mathrm{argmin}_{w \in \mathbb R^p} w' \Sigma w, \quad \text{such that } w' 1_p =1, \quad w' \mu = \rho .$$

Here, $\rho$ is the target monthly return (in our empirics we use $\rho = 0.01$). The mean return $\mu$ is estimated by the average return in our training window. The solution is also well-known and is given by 
$$w^* = \left( \frac{D-\rho F}{AD - F^2} \right) \Gamma 1_p + \left( \frac{\rho A-F}{AD-F^2} \right) \Gamma \mu,$$
where $A:= 1_p' \Gamma 1_p$, $F:=1_p'\Gamma \mu$, and $D:= \mu' \Gamma \mu$.

Finally, the maximum Sharpe ratio portfolio, which is heavily used in practice, has weights given by 
$$w^* := \mathrm{argmax}_{w \in \mathbb R^p} \frac{w'\mu}{\sqrt{w'\Sigma w}}, \quad \text{ such that } 1'_p w = 1.$$ The solution is given by $$w^* = \frac{\Gamma \mu}{1_p' \Gamma \mu}.$$


\section{Empirical Results}\label{sec:empirics}

This section addresses two empirical questions. First, does combining LLM-based screening with a quantitative weighting method produce superior out-of-sample Sharpe ratios relative to the quantitative method operating without any screening? This isolates the incremental value of the multi-agent screening stage. Second, do conventional screening approaches, specifically logistic regression-based screening and human analyst recommendations, achieve comparable performance when paired with the same quantitative method? We show that the multi-agent AI framework outperforms both alternatives across all specifications.

We also conduct a series of robustness checks. First, we evaluate a hybrid agent system (specifically, FinBERT, LLM-S, plus human analyst recommendations) with the quantitative method to determine if incorporating human judgment enhances or detracts from the Agentic AI system. Second, we assess whether a single-agent screening approach (pure FinBERT or LLM-S) and the quantitative method outperforms the two LLM agents together plus the quantitative method.
Lastly, we evaluate how human-only screening with the quantitative method performs compared to the quantitative method alone. This provides a measurable assessment of the value of traditional human screening.

We report the following performance measures: the mean (out-of-sample) monthly return, the out-of-sample variance, and the out-of-sample Sharpe ratio. All of these are calculated with a transaction cost of 10 basis points. Let $y_{P, t+1}=\hat{w}_t' y_{t+1}$ be the gross return of the portfolio at time $t+1$, and $\hat{w}_t$ be the weights of the portfolio at time $t$ with some method. Then the net returns are given by $$y_{P, t+1}^{net} = y_{P, t+1} -c(1+y_{P,t+1}) \sum_{j=1}^p \left|\hat{w}_{t+1,j} - \hat{w}_{t,j} \frac{1+y_{t+1,j}}{1+y_{P,t+1}}\right|,$$
where $c$ is the transaction cost. Turnover metrics, as well as leverage, concentration, and drawdown, can be found in Appendix \ref{turnover}.

Let the in-sample period span from $1$ to $n_I$.
Then the out-of-sample mean return is defined as $$\mu^{net} = \frac{1}{n-n_I} \sum_{t=n_I}^{n-1} y_{P,t+1}^{net}.$$

Note that $\mu^{net}$ is the average out-of-sample portfolio return over rolling windows, reported as ``Returns" in our tables.

The out-of-sample variance is defined as $$\hat{\sigma}^2_{net} = \frac{1}{n-n_I-1} \sum_{t=n_I}^{n-1} (y_{P, t+1}^{net} - \mu^{net})^2,$$
and reported as ``Variance" in our tables.

Lastly, the out-of-sample Sharpe ratio is defined as $$SR^{net}= \frac{\mu^{net}}{\hat{\sigma}_{net}}.$$

\subsection{Data}
Our monthly fundamentals dataset covers all S\&P~500 constituents from January 2005 
to April 2024, sourced from CRSP and Compustat, and includes three firm 
characteristics: size, book-to-market ratio, and 12-month momentum. This dataset 
is a subset of~\cite{green2017characteristics}, and the sample ends in April 2024 
to align with the coverage of our news dataset. Following~\cite{green2017characteristics}, 
we winsorize all characteristics at the 1st and 99th percentiles, standardize each 
to zero mean and unit standard deviation, and replace missing values with zero, 
implying that firms with missing data are assigned the cross-sectional average 
characteristic for that month.\footnote{Similar imputation methods for missing data are the standard in empirical asset pricing: used in \cite{gu2020empirical}, \cite{kozak2020shrinking}, and  \cite{kelly2025artificial}. \cite{chen2024missing} specifically recommend using mean imputation for ML studies. See footnote 5 in \cite{green2017characteristics} for additional justification.}

For each month $t$'s return, we calculate their features at the end of month $t-1$. We assume that annual accounting data are available at this time if the firm's fiscal year ended at least six months before $t-1$.

\subsection{Method}\label{method}

\subsubsection{Stage 1: Screening Agents}

The first stage of our model consists of a screening agent. As described above, this can be a conventional screening agent, such as human analysts, but we also consider LLMs as screening agents. The two such agents are LLM-S, a fundamentals-based LLM that generates buy, sell, and hold signals according to monthly firm fundamentals data, and FinBERT, a news-based agent that generates signals based on sentiment analysis on news each month.

We first describe the conventional screening approaches. The approach using human analysts is based on analyst recommendations on the I/B/E/S dataset. Since analyst recommendations may become stale quickly, we use an exponentially decreasing weight on each analyst's recommendation, based on how far their date of recommendation is to month-end (we use a half-life of $7$ days). We then take the weighted sum of all analysts' recommendations for each S\&P 500 stock, and calculate the change in the weighted sum from month to month. If the change in recommendations is greater than $0.5$ (respectively, less than $-0.5$), we interpret that as a sell (respectively, buy) signal. All others are hold signals. 

The approach using logistic regression runs a simple logistic regression cross-sectionally for all firms on a specific date. We then evaluate the out-of-sample probabilities for the next month, assign buy signals to firms with probabilities in the top decile, and assign sell signals to firms with probabilities in the bottom decile. Everything else is assigned a hold signal.

Now we describe the LLM screening approaches. FinBERT is a specialized, pre-trained natural language processing model designed for financial sentiment analysis (see \cite{araci2019finbert}). We use S\&P 500 news articles provided online at Hugging Face.\footnote{\url{https://huggingface.co/datasets/KrossKinetic/SP500-Financial-News-Articles-Time-Series}} For each stock-month combination, we use FinBERT to analyze all news articles in that month. We define a firm's sentiment score to be the positive FinBERT probability minus the negative FinBERT probability and using the same exponentially-decreasing weighted sum as above to account for potentially stale news. If the sentiment score is greater than $0.1$ (respectively less than $-0.1$), we assign a buy signal (respectively assign a sell signal). All other stocks are assigned hold signals.

Lastly, we describe LLM-S. We use ChatGPT-4o for short-term analysis and ChatGPT-3.5  for medium-term analysis.
We give LLM-S the necessary tools to access firm fundamentals data for a given date. 
It then considers the distribution of firm size, book-to-market, and momentum values for that date, and outputs a deterministic scoring rule for buy, hold, and sell signals. To leverage the zero-shot capabilities of LLMs, we only provide firm characteristics to the LLM at the end of the year, and ask it to provide buy, hold, and sell signals for the next year. Importantly, the agent does not use statistical learning on historical data. Instead, it relies on pretrained domain knowledge and in-context reasoning to analyze the current cross-section of firms and produce screening decisions. Finally, to ensure a sanity check, we also ask the agent to provide its economic intuition for why it chooses the scoring rule that it does. We have included snapshots of our prompts and an example output in the appendix. 

In the following sections, we also consider ensembles of agents (for example, LLM-S + FinBERT), geared at reducing hallucinations/inaccurate recommendations for each agent. The decision rule for the ensemble is given as follows: given an ensemble of agents $A_1$ and $A_2$, that provide sets of buys and sells $S_1$ and $S_2$, respectively, the ensemble $A_1+A_2$ recommends the set $S_1 \cap S_2$ to invest in. However, if the intersection is empty, then it recommends the set $S_1 \cup S_2$. This is saying that ideally, we would like to take the firms in the consensus of both agents. However, if the two agents cannot reach a consensus, we take all of their recommendations. If we have an ensemble of three agents $A_1, A_2$, and $A_3$, with recommendations $S_1$, $S_2$, and $S_3$, respectively, then the ensemble $A_1+A_2+A_3$ will recommend stocks that belong in at least two distinct sets. This is akin to taking the majority vote between three agents. The above particular decision rule defaults to the union of sets in case the intersection is trivial.

We explain our choice of decision rule further. Logically, the set of stocks with high alpha should be the ones that have exhibited both positive sentiment and positive fundamentals recently. This is precisely what the intersection is capturing - the set of stocks that have predictably high alpha. However, when there is no such intersection, it may be wiser to turn to diversification instead, which is what the union captures.


As a robustness check, we also changed our rule in the following way. First we use the intersection rule, and if intersection is empty, we default to the FinBERT agent's selection (\cite{ckx2023} show the strong effect of parsing news sentiments. Our own single-agent tables in Tables \ref{tab2} and \ref{tab5} also show FinBERT can deliver better performance than LLM-S.) These results can be found in Appendix \ref{consensus rules}.

\subsubsection{Stage 2: Quantitative Weighting Method}

Given a set of buy, sell, and hold signals from Stage 1, we run a variety of statistical techniques and portfolio weight formations to determine the combination that performs best out-of-sample. The statistical techniques used for high-dimensional portfolio formation are  nodewise regression, residual nodewise regression, deep learning, and nonlinear shrinkage, and also we use inverse of sample covariance matrix since the screened number of stocks are smaller than our fixed window size of $n=180$ (training sample size). Since the screened number of stocks are not small in comparison with training sample size, we still use other machine learning methods for precision matrix estimation. However, our econometric theory covers both screened number of stocks $\hat{p}$ smaller than equal to n, as well as $\hat{p}>n$ in case future data or a new method may find $\hat{p}>n$.
The portfolio weight formations are the global minimum variance portfolio, the mean-variance efficient portfolio, and the maximum Sharpe ratio portfolio. All of these methods are described in depth in Section \ref{quant strategies}. 

We only apply these techniques to firms that have either a buy or a sell signal in Stage 1. By considering the subset of non-hold firms as a whole, we allow the quantitative weighting method to correct potential mistakes made by the LLM screening agents. This prevents a cascade of errors (e.g., assigning a positive weight to a stock that the LLM screening agent recommended selling).

Lastly, to calculate optimal weights, we utilize a rolling window of 180 months (15 years) of historical returns data, stepping forward one month at a time. We  analyze the period from October 2021 through April 2024, the period after ChatGPT-3.5's knowledge cutoff date

In the second empirical exercise,  we focus our analysis here on the monthly runs from November 2023 through April 2024, the period after ChatGPT-4o's knowledge cutoff date. 

\subsection{Results: Medium-Term}\label{results_medium}

This section presents our empirical results, progressing from baseline comparisons to the evaluation of our full multi-agent framework. Our choice of LLM-S agent will be ChatGPT-3.5, since its knowledge cutoff date is September 2021. Hence our out-of-sample period will be October 2021-April 2024.
We first ask whether pairing an LLM screening agent with a quantitative method improves upon a purely quantitative baseline, and whether any such improvement is specifically attributable to LLM-based screening rather than screening per se. We then examine whether news-based sentiment screening via FinBERT contributes incremental performance gains, and whether integrating traditional analyst judgment with LLM-based screening generates higher Sharpe ratios than either approach alone. Finally, we evaluate whether our multi-agent AI system delivers superior risk-adjusted performance relative to all single-agent and hybrid alternatives, and whether the best-performing configurations outperform a passive market benchmark. To contextualize these performance metrics, we consistently benchmark our strategies against the S\&P 500 index, which generated a Sharpe ratio of 0.2685 between October 2021 and April 2024. 

Our main results are in Tables \ref{tab35-1}-\ref{tab35-9}, which detail annualized Sharpe ratios, returns, and variance metrics. The top Sharpe ratio for each method-portfolio combination is indicated in bold. All our results are after transaction costs of 10 basis points.

\subsubsection{Quantitative Method versus LLM-S Plus Quantitative Method}\label{ql_medium}

Table \ref{tab35-1} presents the baseline case of a purely quantitative strategy without any LLM agent or screening method. In this setup, the NLS/MSR portfolio performs best, achieving a Sharpe ratio of 1.0088, as well as the highest annualized return in the baseline model at 18.01\%.  Notably, there are only three configurations under the purely quantitative baseline that successfully outperform the broader market's Sharpe ratio of 0.2685; all other quantitative strategies fail to pass this benchmark.

The agent in Table \ref{tab35-2} consists of a two-stage model, where the LLM-S agent first generates buy/sell signals, and the quantitative method subsequently determines the portfolio weights. Under this framework, the Deep learning/MSR portfolio performs best, achieving a Sharpe ratio of 0.6108. When using the LLM-S agent for screening, all GMV and MSR portfolios (12 in total)  exceed the market threshold.

A critical question is whether incorporating LLM-S screening provides an advantage over the purely quantitative baseline. 
We see that in all method/portfolio combinations, using LLM-S screening improves all Sharpe ratios (SR henceforth) in 12 portfolios out of 15 (There is no Sample Cov in Baseline since baseline there is no screening and $p>n$ and Sample Covariance is singular-non-invertible). To illustrate the magnitude of this impact, consider the deep learning/MSR portfolio: without screening, the baseline quantitative method yields an SR of 0.0586 in Table \ref{tab35-1}, but there is a dramatic improvement to an SR of 0.6108 in Table \ref{tab35-2}. This illustration demonstrates how LLM-S screening can successfully transform a low-SR strategy into a profitable one.

\subsubsection{Is It Screening or LLM-S Screening That Makes the Difference?}

To determine whether our observed performance gains are driven by screening in general or uniquely by the LLM-S agent, we establish a benchmark using logistic regression-based screening. 
Once again, the NLS/MSR portfolio emerges as the top performer, achieving a Sharpe ratio of 1.0015 in Table \ref{tab35-3}. Against the market benchmark, only 11 configurations  manage to beat or match the S\&P 500 under this screening method. Comparing these results to the LLM-S screening in Table \ref{tab35-2}, we observe that the LLM-S results are mixed, although the highest SR achieved by logistic screening is larger than that of LLM-S.

As a second benchmark, we evaluate screening based on human analyst buy/sell recommendations that are used in conjunction with the quantitative strategies outlined above, with results in Table \ref{tab35-4}. Relying on human analyst recommendations yields the poorest overall performance, with every single method-portfolio combination failing to beat the market benchmark. Comparing this directly with Table \ref{tab35-2} provides a clear assessment of humans versus LLM-S screening capabilities. With the exception of the NW/MV method-portfolio combination, the Sharpe ratios generated by LLM-S screening (Table \ref{tab35-2}) are larger than the ones generated by human analysts (Table \ref{tab35-4}). This points towards the superiority of the LLM-S model for screening. We attribute this to the over- and under-reaction tendencies well documented in the behavioral finance literature in a period which has a large downturn in 2022.

We can also address whether any screening method together with a quantitative method improves upon the baseline quantitative strategy (Table \ref{tab35-1}). Comparing the logistic screening approach (Table \ref{tab35-3}) with the baseline (Table \ref{tab35-1}), we see that it does: 14 out of the 18 method-portfolio combinations exhibit improved Sharpe ratios when logistic regression screening is applied. 
Conversely, when comparing human analyst screening with the quantitative method (Table \ref{tab35-4}) to the baseline quantitative strategy, the results are mixed. It is not clear whether human judgment can improve upon the baseline quantitative strategy.
But top SR in Table \ref{tab35-4} is very low at 0.1910, compared to baseline top SR of 1.0088.


In summary, combining LLM-S based screening with quantitative strategies to form the portfolio dominates both the baseline quantitative method and human-based  screening methods coupled with quantitative methods.


\subsubsection{Is FinBERT-Based Screening Helpful?}

Next, we investigate whether pairing sentiment analysis with the quantitative strategy outperforms the purely quantitative baseline.
Comparing Table \ref{tab35-1} with Table \ref{tab35-5} reveals that screening stocks with FinBERT increases Sharpe ratios across all method-portfolio combinations, except in the three NLS portfolio specifications. Furthermore, sentiment-based screening via FinBERT proves highly effective against the broader market: in all method-portfolio combinations it beats the S\&P 500 Sharpe ratio of 0.2685.

We then assess which AI screening method is more effective: FinBERT or LLM-S? To answer that question, we compare Table \ref{tab35-2} with Table \ref{tab35-5}. The winner in Table \ref{tab35-5} (Residual NW/MSR) yields a Sharpe ratio of 0.6937 via FinBERT screening, and the winner in Table \ref{tab35-2} (deep learning/MSR) achieves a Sharpe ratio of 0.6108. Not only does FinBERT achieve a higher maximum Sharpe ratio overall, but a direct comparison across all method-portfolio combinations reveals that FinBERT outperforms LLM-S in 14 out of 18 cases.

Furthermore, we evaluate whether FinBERT-based screening is better than logistic regression-based screening. To that effect we compare Tables \ref{tab35-3} and \ref{tab35-5}. FinBERT dominates logistic screening in 10 out of 18 method-portfolio combinations. For instance, the deep learning/GMV portfolio achieves a 0.3254 Sharpe ratio under logistic regression screening, but it increases to 0.4262 SR using FinBERT screening - a 30.9\% increase. Lastly, comparing human analyst screening (Table \ref{tab35-4}) versus FinBERT (Table \ref{tab35-5}) demonstrates that FinBERT uniformly outperforms human judgment. Notably, the advantage FinBERT screening has over human analyst screening is larger than it has over logistic regression screening. 

Ultimately, a consistent pattern emerges: FinBERT-based screening easily outperforms both logistic regression and human-based screening, while also substantially improving upon the baseline quantitative strategy as well.

When we combine FinBERT with human analysts for screening and then use the quant agent (Table \ref{tab35-7}), the results are the worst, with all portfolios resulting in negative SRs. This clearly demonstrates that combining FinBERT with human judgment is detrimental; we hypothesize that human screening during market drawdowns yields lower Sharpe ratios due to well-documented behavioral biases

\subsubsection{Can Hybrid Screening Help?}
In this part, we analyze a hybrid approach that integrates the LLM-S agent with human analyst recommendations. The results of this ensemble are presented in Table \ref{tab35-6}, and further details can be found in Section \ref{method}.
First, we evaluate this ensemble against the baseline quantitative method by comparing Tables \ref{tab35-1} and \ref{tab35-6}. The results are mixed and there is no clear evidence that incorporating human judgment into the quantitative pipeline yields a performance advantage. The top SR of baseline quant table is 1.0088 and the top SR in Table \ref{tab35-6} is 0.3710. 

It could be possible that suboptimal decisions by human analysts could have degraded the ensemble's performance. To check that, we compare the hybrid results (Table \ref{tab35-6}) with the pure LLM-S screening results (Table \ref{tab35-2}). The highest Sharpe ratio in Table \ref{tab35-6} is 0.3710 (Residual NW/MV), which beats  the 0.2685 market benchmark but much lower than  the 0.6108 peak Sharpe ratio achieved by pure LLM-S in Table \ref{tab35-2}. The difference is substantial, indicating that human judgment may be suppressing the LLM-S Sharpe ratio during this evaluated time period. However, it seems that in all MV portfolios, humans add value to LLM screening. In the other portfolios, humans cannot add value.
For instance, the deep learning/MV portfolio achieves a Sharpe ratio of 0.1132 under LLM-S screening, but this increases to 0.2625 when human analysts are added.


\subsubsection{Multi-Agent AI versus Single-Agent Screening}

In the medium term, we want to understand whether there is value to Agentic AI-based screening (FinBERT plus LLM-S) over the alternative portfolio formations discussed so far. In particular, we evaluate whether the Agentic AI architecture outperforms the baseline quantitative strategy. Then, can Agentic AI do better than standalone single-agent models: LLM-S only or FinBERT only?

The results are presented in Table \ref{tab35-8}. The results are excellent when benchmarked against the broader market: twelve of the method-portfolio combinations exceed the S\&P 500 Sharpe ratio of 0.2685. MV portfolios perform poorly though.  The top-performing configuration is the residual NW/GMV portfolio with a 1.0946 Sharpe ratio, representing a 307\% improvement over the S\&P 500 SR of 0.2685. That is the largest SR for the same time period compared across all other tables. 

We then benchmark the Agentic AI-based approach against the baseline quantitative strategy baseline shown in Table \ref{tab35-1}. The performance difference is striking. For example, the deep learning/GMV portfolio yields a Sharpe ratio of -0.1147. Introducing Agentic AI screening elevates the Sharpe ratio to 0.9858, a remarkable turnaround. Moreover, with the exception of the MV portfolios and the NLS/MSR portfolio, the Agentic AI system dominates the baseline across other method-portfolio combination. 

Beyond Sharpe ratios, the Agentic AI framework generates impressive absolute returns. The peak annualized return is 37.31\% in our out-of-sample period, produced by the Residual NW/GMV portfolio in Table \ref{tab35-8}. The top return for the purely quantitative baseline without any screening is 18.01\% (NLS/MSR, Table \ref{tab35-1}).

Next, we answer whether a single agent by itself (FinBERT or LLM-S) can deliver better results than Agentic AI. To evaluate this, we compare the Agentic AI results (Table \ref{tab35-8}) against the FinBERT-only (Table \ref{tab35-5}) and LLM-S only (Table \ref{tab35-2}) frameworks. The evidence heavily favors Agentic AI. The FinBERT-only maximum Sharpe ratio is 0.6937 (Residual NW/MSR), compared to the Agentic AI winner of 1.0946 (Residual NW/GMV). Out of 18 combinations, in ten individual method-portfolio combinations, the Agentic AI system significantly outperforms the FinBERT-only model. The main exceptions are MV portfolios.  For instance, the deep learning/GMV portfolio has a Sharpe ratio of 0.4262 in Table \ref{tab35-5}, but this increases to 0.9858 when using Agentic AI. 
Similarly, the highest Sharpe ratio achieved by LLM-S is 0.6108, compared to 1.0946 with Agentic AI. With the main exception of MV portfolios, in 11 out of 18  method-portfolio combinations, Agentic AI consistently dominates LLM-S. For example, LLM-S delivers a Sharpe ratio of 0.3177 in the NW/GMV portfolio, and this increases to 0.8794 when using Agentic AI. The increase is large and can make a large difference in practice.

In summary, there is a substantial gain by using Agentic AI compared to single LLM agents.  Excluding the MV portfolios, the Agentic AI system also dominates in all logistic or human screening portfolios.

\subsubsection{Multi-Agent AI with Human Analyst Input}\label{aih_medium}

Finally, we investigate whether integrating human judgment into the Agentic AI framework yields any benefits. To evaluate this, Table \ref{tab35-9} presents the results of an ensemble consisting of FinBERT, LLM-S, and human analysts working together.  We benchmark this against the pure Agentic AI system detailed in Table \ref{tab35-8}. It is clear that adding human judgment significantly degrades performance, reducing the maximum SR from 1.0946 (residual NW/GMV) to 0.1117 (POET/GMV). Furthermore, in this three-agent screening ensemble, all resulting Sharpe ratios fall below the market benchmark, underscoring the detrimental effect of human judgment. 

We also investigated the factors contributing to the low performance of human analysts. We find that it may be related to overreaction and reversing positions quickly in volatile markets. 

\subsection{Results: Short-Term Performance}\label{results}

This section presents our empirical results in a short out-of-sample period, progressing from baseline comparisons to the evaluation of our full multi-agent framework. We first ask whether pairing an LLM screening agent with a quantitative method improves upon a purely quantitative baseline, and whether any such improvement is specifically attributable to LLM-based screening rather than screening per se. We then examine whether news-based sentiment screening via FinBERT contributes incremental performance gains compared to the purely quantitative baseline, and whether integrating traditional analyst judgment with LLM-based screening generates higher Sharpe ratios than either screening approach with the quant agent. Finally, we evaluate whether our full multi-agent AI system delivers superior risk-adjusted performance relative to all single-agent and hybrid alternatives, and whether the best-performing configurations outperform a passive market benchmark. To contextualize these performance metrics, we consistently benchmark our strategies against the S\&P 500 index, which generated a Sharpe ratio of 2.0857 between November  2023 and April 2024. 

Our main results are in Tables \ref{tab1}-\ref{tab9}, which detail annualized Sharpe ratios, returns, and variance metrics. The top Sharpe ratio for each method-portfolio combination is indicated in bold. All our results are after transaction costs of 10 basis points.

\subsubsection{Quantitative Method versus LLM-S Plus Quantitative Method}\label{ql}

Table \ref{tab1} presents the baseline case of a purely quantitative strategy without any LLM agent or screening method. In this setup, the NLS/MV portfolio performs best, achieving a Sharpe ratio of 3.0414 with a 19.29\% annualized return. Furthermore, operating on the entire universe of S\&P 500 stocks allows the baseline to maximally diversify away idiosyncratic risk, achieving the absolute lowest variances across all tables and models (0.0025 variance for the baseline Residual NW/MV portfolio). Notably, there are only two configurations under the purely quantitative baseline that successfully outperform the broader market's Sharpe ratio of 2.0857; all other quantitative strategies fail to pass this benchmark.

The agent in Table \ref{tab2} consists of a two-stage model, where the LLM-S agent first generates buy/sell signals, and the quantitative method subsequently determines the portfolio weights. Under this framework, the Sample Cov portfolio performs best, achieving a Sharpe ratio of 1.1629. When using the LLM-S agent for screening, none of the strategies exceed the market threshold.

A critical question is whether incorporating LLM-S screening provides an advantage over the purely quantitative baseline. 
We see that in all method/portfolio combinations, using LLM-S screening improves Sharpe ratios (SR henceforth) only in three deep learning portfolios, albeit a significant increase. To illustrate the magnitude of this impact, consider the deep learning/GMV portfolio: without screening, the baseline quantitative method yields an SR of -0.4666 in Table \ref{tab1}, but there is a dramatic improvement to a SR of 0.5857 in Table \ref{tab2}. This illustration demonstrates how LLM-S screening can successfully transform a negative-SR strategy into a profitable one in certain cases. 

However, it is clear that in the short run, LLM-S with the quantitative agent performs worse than the market, and it is also not necessarily better than the pure quant strategy. This is evidence towards the narrative that single-agent LLM strategies may not work well by themselves.

\subsubsection{Is Simple Logistic or Human Screening Better Than LLM-S?}

We establish a benchmark using logistic regression-based screening, detailed in Table \ref{tab3}. Specifically,  we use the prior  training window to fit a cross-sectional logistic regression to select the top and bottom deciles of firms to long and short, respectively. We then apply the quantitative agent on the output of the logistic regression to assign portfolio weights.
Once again, the NLS/MSR portfolio emerges as the top performer, achieving a Sharpe ratio of 3.3540 in Table \ref{tab3}. Against the market benchmark of S\&P 500 SR of 2.0857, only six configurations of the portfolio manage to beat the index SR  under this screening method. Comparing these results to the LLM-S screening in Table \ref{tab2}, logistic dominates LLM-S in all configurations of the portfolios that are formed.

As a second benchmark, we evaluate screening based on human analyst buy/sell recommendations that are used in conjunction with the quantitative agent, with results in Table \ref{tab4}. When relying on human analyst recommendations, only ten portfolios beat the market benchmark of 2.0857. But the best human-quant Sample Cov  portfolio has a SR of 3.5374, easily surpassing the S\&P 500 SR.

Comparing this directly with Table \ref{tab2} provides a clear assessment of humans versus LLM-S screening capabilities: across every single method/portfolio combination, the Sharpe ratios generated by human screening (Table \ref{tab4}) are larger than the ones generated by LLM-S (Table \ref{tab2}). This points towards the superiority of the human screening with the quantitative model for screening, compared to single-agent LLMs.

We can also answer whether a conventional screening method, together with a quantitative method, improves upon the baseline quantitative strategy (Table \ref{tab1}). Comparing the logistic screening approach (Table \ref{tab3}) with the baseline (Table \ref{tab1}), we see that it does: the best SR is 3.5374 in Table \ref{tab3} and the best SR is 3.0414 in Table \ref{tab1}. The robustness is more mixed: only 7 out of the 15 (no Sample Cov in baseline) method-portfolio combinations exhibit improved Sharpe ratios when logistic screening is applied and then  weighted with a quantitative method.
However, when comparing human analyst screening with the quantitative method (Table \ref{tab4}) to the baseline quantitative strategy, the results are more strongly in favor of adding human screening to the quant agent. Out of 15 portfolio combinations (no Sample Cov in baseline), human screening with the quantitative agent obtains 11 portfolios that are better than the counterparts in Table \ref{tab1}.

In summary, combining LLM-S based screening with quantitative strategies to form portfolios does not work well by itself, compared to both the baseline quantitative method and conventional screening methods (logistic or humans) coupled with quantitative methods. We hypothesize this is due to the short out-of-sample window: LLM-S is designed to capture long-term fundamentals and thus is not compatible with such a short evaluation window.

\subsubsection{Is FinBERT-Based Screening the Answer in the Short Run?}

Next, we investigate whether pairing sentiment analysis with the quantitative strategy outperforms the purely quantitative baseline.
Comparing Table \ref{tab1} with Table \ref{tab5} reveals that screening stocks with FinBERT increases Sharpe ratios in 10 out of 15 (no Sample cov in baseline) method-portfolio combinations.

We then assess which AI screening method is more effective: FinBERT or LLM-S? To answer that question, we compare Table \ref{tab2} with Table \ref{tab5}. The winner in Table \ref{tab5} (Sample Cov/MSR) yields a Sharpe ratio of 7.2046 via FinBERT screening, and the winner in Table \ref{tab2} (Sample Cov/MSR) achieves a Sharpe ratio of 1.1629
. This is a remarkable difference. 
While FinBERT achieves a higher maximum Sharpe ratio overall, also a direct comparison of each method-portfolio combination yields uniformly better  results in favor of FinBERT. It is clear that sentiment analysis in the short run plays a major role compared to more long-run oriented fundamental analysis.

Furthermore, we evaluate whether FinBERT-based screening is better than logistic regression-based screening. To that effect we compare Tables \ref{tab3} and \ref{tab5}. FinBERT dominates logistic screening in 15  method-portfolio combinations out of 18  possibilities. For instance, in logistic the best portfolio SR is the NLS/MSR portfolio and achieves a 3.3540 Sharpe ratio under logistic regression screening, but the best SR in FinBERT  is 7.2046  - a 99\% increase. Lastly, comparing human analyst screening (Table \ref{tab4}) versus FinBERT (Table \ref{tab5}) demonstrates that FinBERT  outperforms human judgment in 10 out of 18 combinations. But the difference between the top performers is stark: human analysts achieves a SR of 3.5374 (Sample Cov-MSR), but FinBERT achieves a SR (Sample Cov/MSR) of 7.2046. It is clear that the best method-objective combination of FinBERT is outperforming both logistic and human screeners by a large margin.

We also considered FinBERT with human screeners coupled with the quant agent in Table \ref{tab7}. Its top Sharpe ratio of 2.4582 in POET/GMV is much lower than the SR of FinBERT with the quant agent in Table \ref{tab5} (7.2046-Sample Cov/MSR). It is clear that in this short out-of-sample evaluation, sentiment analysis is the reason why FinBERT performed well.

\subsubsection{Can Hybrid Screening Help?}
In this part, we analyze a hybrid approach that integrates the LLM-S agent with human analyst recommendations. The results of this ensemble are presented in Table \ref{tab6}, and further details can be found in Section \ref{method}.
First, we evaluate this ensemble against the baseline quantitative method by comparing Tables \ref{tab1} and \ref{tab6}. The results are mixed, and there is no clear evidence that incorporating human judgment with LLM-S into the quantitative pipeline yields a performance advantage.

We then compare Table \ref{tab2} of only LLM-S screening with Table \ref{tab6} of human analyst plus LLM-S screening to see whether LLM-S  adds value to human screening.
The highest Sharpe ratio in Table \ref{tab6} is 2.7553 (NLS/MSR), which beats the 2.0857 market benchmark and the 1.1629 peak Sharpe ratio achieved by pure LLM-S in Table \ref{tab2}. The difference is very large, indicating that human judgment helps the LLM-S Sharpe ratio during this short out-of-sample time period. Furthermore, the results broadly hold (12 out of 18 combinations) across the method-portfolio comparisons. For instance, the deep learning/MSR portfolio achieves a Sharpe ratio of 0.6122 under LLM-S screening, but this increases to 2.0881 when human analysts are added, a dramatic increase.  We also compare Tables \ref{tab4} and \ref{tab6} to see whether LLM-S adds value to human screeners, but clearly this is not the case, since only 2 portfolio combination in Table \ref{tab6} beats its equivalent in Table \ref{tab4}. In the short-term, human screening with the quant agent works well and can do better than LLM-S.


\subsubsection{Multi-Agent AI versus Single-Agent Screening}

We will refer to the ensemble consisting of FinBERT and LLM-S as the Agentic AI system, as it relies on multiple specialized LLM-based agents coordinating with each other and executing different tasks.
We want to understand whether there is  value to Agentic AI-based screening over the alternative portfolio formations discussed so far. We first evaluate whether the Agentic AI architecture outperforms the baseline quantitative strategy. Then, can Agentic AI do better than standalone single-agent models: LLM-S only or FinBERT only? We finally analyze its performance against conventional screening such as logistic and human forecasters.

In this framework, the agent team operates collaboratively: FinBERT screens based on short-term sentiment analysis, while LLM-S screens based on a fundamentals-driven strategy. Their recommendations are combined together into a consensus signal, which is subsequently processed by the quantitative method to decide the portfolio weights. The results are presented in Table \ref{tab8}. The results are excellent when benchmarked against the broader market: the top Sharpe ratio (8.1612) exceeds the S\&P 500 SR (2.0857). 

We then benchmark the Agentic AI-based approach against the baseline quantitative strategy shown in Table \ref{tab1}. The performance difference is striking. With four exceptions, the Agentic AI system dominates the baseline across every method-portfolio combination. The differences between the two tables are large; for example, the deep learning/MV portfolio yields a Sharpe ratio of -0.4648. Introducing Agentic AI screening elevates the Sharpe ratio to 1.4812, a remarkable turnaround.

Beyond Sharpe ratios, the Agentic AI framework generates impressive absolute returns. The peak annualized return is 55.75\% in our out-of-sample period, produced by the Sample Cov/MSR portfolio in Table \ref{tab8}. Further, every method-objective combination in the baseline model generates lower returns compared to its direct counterpart in the Agentic AI ensemble. For instance, the Residual NW/GMV portfolio produces a 5.86\% annual return under the quantitative-only strategy, whereas the same method-portfolio combination produces a 24.58\% return annually with the Agentic AI framework—nearly a four-fold increase.

These returns/variance patterns illustrate a fundamental {\it predictability-diversification} tradeoff. The purely quantitative strategy operates on the entire universe of the S\&P 500 (no screening), which allows it to diversify away idiosyncratic risk and achieve the absolute lowest variances out of all tables and models, as seen in the 0.0025 variance by the Residual NW/GMV portfolio. In contrast, the Agentic AI framework acts as a high conviction screener, investing in an average of 20 stocks based on strong fundamentals and positive news sentiment. 
By holding such concentrated portfolios, the Agentic AI models naturally sacrifice 
some diversification benefits, resulting in higher variances. Notably, the lowest 
variance among all Agentic AI portfolios is $0.0042$, achieved by the NLS/MSR portfolio.   However, the high predictability and significant excess returns generated by the Agentic AI portfolios far outweigh the diversification penalty. For instance, while the highly diversified baseline Residual NW/MV portfolio minimizes variance (0.0025), it only yields an annualized return of 6.96\% and a corresponding Sharpe ratio of 1.3992. On the other hand, the Agentic AI Residual NW/MV portfolio results in a higher variance of 0.0201, but compensates with a 22.74\% annualized return, pushing its Sharpe ratio to an impressive 1.6028. All this demonstrates why the Agentic AI system delivers superior performance: the predictability that Agentic AI demonstrates outweighs the diversification benefits it sacrifices by taking a screened portfolio of stocks.

Next, we address whether a single agent by itself (FinBERT or LLM-S) can deliver better results than Agentic AI. To evaluate this, we compare the Agentic AI results (Table \ref{tab8}) against the FinBERT-only (Table \ref{tab5}) and LLM-S only (Table \ref{tab2}) frameworks. The evidence heavily favors Agentic AI. The FinBERT-only maximum Sharpe ratio is 7.2046 (Sample Cov/MSR), compared to the Agentic AI maximum of 8.1612 (NLS/MSR). Across all individual method-portfolio combination, the Agentic AI system significantly outperforms the FinBERT-only model, with only three exceptions, two of them  in the Deep learning/GMV-MSR portfolios. For instance, the Residual NW/MSR portfolio has a Sharpe ratio of 5.0523 in Table \ref{tab5}, but this increases to 5.7291 when using Agentic AI. Similarly, the highest Sharpe ratio achieved by LLM-S is 1.1629, compared to 8.1612 with Agentic AI. Across every method-portfolio combination, Agentic AI consistently dominates LLM-S. For example, LLM-S delivers a Sharpe ratio of 0.5962 in the NW/GMV portfolio, and this increases to 1.3801 when using Agentic AI. The increase is large and can make a large difference in practice.

The main reason why Agentic AI outperforms single agents is that it selects firms with desirable fundamentals \textit{and} news, with the intersection as an opportunity for the two separate agents to mitigate each other's errors. 
 In addition, we also see  evidence that dropping certain stocks from the portfolio can lead to a significant increase in the out-of-sample Sharpe ratio. This justifies our decision to screen for stocks first, and our Agentic AI model chooses those stocks with both promising long-term fundamentals and short-term news, making this a natural choice for a screening agent.

In summary, there is a substantial gain by using Agentic AI. It can easily surpass single agent systems, and even though we omitted a detailed comparison versus human analysts or logistic regression screening, the differences are large and can be seen by comparing Tables \ref{tab3} and \ref{tab4}  with Table \ref{tab8}.
The results can further be interpreted as the deviation from a benchmark with large number of assets. Our multi-agent AI ensemble selects a substantially lower number of stocks out of a large universe, and hence takes on a large tracking error, compared to the S\&P 500, to achieve larger returns and Sharpe ratios. Theoretical results on tracking error in portfolios have been recently established in high dimensions in \cite{cf2026}.

\subsubsection{Multi-Agent AI with Human Analyst Input}\label{aih}

Finally, we investigate whether integrating human judgment into the Agentic AI framework yields any benefits. To evaluate this, Table \ref{tab9} presents the results of an ensemble consisting of FinBERT, LLM-S, and human analysts working together.  We benchmark this against the pure Agentic AI system detailed in Table \ref{tab8}. It is clear that adding human judgment significantly degrades performance, reducing the maximum SR from 8.1612 (NLS/MSR) to 1.4193 (Sample Cov/MSR). Furthermore, in this three-agent ensemble, all resulting Sharpe ratios fall below the market benchmark, underscoring the detrimental effect of human judgment. This degradation is also consistently true when we compare each method-portfolio individually. To give a particularly stark example, the deep learning/MSR portfolio achieves a Sharpe ratio of 3.4411 when using Agentic AI, but this decreases to 1.0251 when human analysts are added. 
It is clear that while humans add value to the quant agent, but in an Agentic AI framework, the performance declines.

\begin{table}[h!]
    \centering
    {\bf BASELINE-ONLY WITH QUANTITATIVE WEIGHTING: Oct 2021- April 2024}\\
    \begin{tabular}{l|ccc|ccc|ccc}
\hline
 & \multicolumn{3}{c|}{\textbf{Sharpe Ratio}} & \multicolumn{3}{c|}{\textbf{Returns}} & \multicolumn{3}{c}{\textbf{Variance}} \\
\textbf{Method} & \textbf{GMV} & \textbf{MV} & \textbf{MSR} & \textbf{GMV} & \textbf{MV} & \textbf{MSR} & \textbf{GMV} & \textbf{MV} & \textbf{MSR} \\ \hline
NW & 0.2184 & 0.2169 & 0.2083 & 0.0340 & 0.0336 & 0.0325 & 0.0243 & 0.0241 & 0.0243 \\
Residual NW & 0.2625 & 0.2625 & 0.0401 & 0.0347 & 0.0330 & 0.0073 & 0.0175 & 0.0158 & 0.0329 \\
Deep learning & -0.1147 & -0.0873 & 0.0586 & -0.0245 & -0.0190 & 0.0180 & 0.0456 & 0.0474 & 0.0944 \\
POET & 0.1330 & 0.1620 & 0.1852 & 0.0172 & 0.0205 & 0.0228 & 0.0167 & 0.0161 & \textbf{0.0152} \\
NLS & 0.4979 & 0.4985 & \textbf{1.0088} & 0.0640 & 0.0635 &  \textbf{0.1801} & 0.0165 & 0.0162 & 0.0319 \\ 
\hline
\end{tabular}%
    \caption{Annualized Sharpe ratios, returns, and variance with different methods of estimating the precision matrix, with different objective functions, applied to all firms in the S\&P 500. GMV=Global minimum variance portfolio, MV=mean-variance portfolio with target returns as 1\% monthly, MSR=maximum Sharpe ratio portfolio.}\label{tab35-1}
\end{table}

\begin{table}[htbp]
    \centering
    {\bf LLM-S WITH QUANTITATIVE WEIGHTING: Oct 2021- April 2024}
    \begin{tabular}{l|ccc|ccc|ccc}
\hline
 & \multicolumn{3}{c|}{\textbf{Sharpe Ratio}} & \multicolumn{3}{c|}{\textbf{Returns}} & \multicolumn{3}{c}{\textbf{Variance}} \\
\textbf{Method} & \textbf{GMV} & \textbf{MV} & \textbf{MSR} & \textbf{GMV} & \textbf{MV} & \textbf{MSR} & \textbf{GMV} & \textbf{MV} & \textbf{MSR} \\ \hline
NW & 0.3177 & 0.0012 & 0.3549 & 0.0742 & 0.0003 & 0.0866 & 0.0546 & 0.0645 & 0.0595 \\
Residual NW & 0.5014 & 0.0400 & 0.5568 & 0.1118 & 0.0098 & \textbf{0.1619} & 0.0498 & 0.0597 & 0.0845 \\
Deep learning & 0.5026 & 0.1132 & \textbf{0.6108} & 0.1169 & 0.0294 & 0.1543 & 0.0541 & 0.0675 & 0.0638 \\
POET & 0.4125 & 0.2336 & 0.5520 &  0.1067 & 0.0628 & 0.1452 & 0.0668 & 0.0723 & 0.0691 \\
NLS & 0.5864 & 0.0964 & 0.5487 & 0.1261 & 0.0235 & 0.1606 & \textbf{0.0462} & 0.0591 & 0.0857 \\ 
Sample Cov & 0.6194 & 0.0960 & 0.5386 & 0.1338 & 0.0234 & 0.1655 & 0.0466 & 0.0596 & 0.0944 \\
\hline
\end{tabular}%
    \caption{Annualized Sharpe ratios, returns, and variance with different methods of estimating the precision matrix, with different objective functions, applied to firms that the LLM has screened. Portfolio abbreviations as in Table~\ref{tab35-1} and we also have Sample Cov which is inverting Sample Covariance to obtain precision matrix estimate when $\hat{p}\le n$.}\label{tab35-2}
\end{table}

\begin{table}[htbp]
    \centering
    {\bf LOGISTIC REGRESSION WITH QUANTITATIVE WEIGHTING: Oct 2021-April 2024}\\
    \begin{tabular}{l|ccc|ccc|ccc}
\hline
 & \multicolumn{3}{c|}{\textbf{Sharpe Ratio}} & \multicolumn{3}{c|}{\textbf{Returns}} & \multicolumn{3}{c}{\textbf{Variance}} \\
\textbf{Method} & \textbf{GMV} & \textbf{MV} & \textbf{MSR} & \textbf{GMV} & \textbf{MV} & \textbf{MSR} & \textbf{GMV} & \textbf{MV} & \textbf{MSR} \\ \hline
NW & 0.1747 & 0.1534 & 0.1682 & 0.0256 & 0.0224 & 0.0251 & 0.0214 & 0.0213 & 0.0223 \\
Residual NW & 0.4566 & 0.4129 & 0.1070 & 0.0577 & 0.0513 & 0.0200 & 0.0160 & \textbf{0.0155} & 0.0350 \\
Deep learning & 0.3254 & 0.2685 & 0.3476 & 0.0454 & 0.0376 & 0.0512 & 0.0195 & 0.0196 & 0.0217 \\
POET & 0.2185 & 0.2264 & 0.1634 & 0.0311 & 0.0316 & 0.0249 & 0.0202 & 0.0194 & 0.0232 \\
NLS & 0.6552 & 0.5474 & 0.8268 & 0.0824 & 0.0689 & 0.1562 & 0.0158 & 0.0158 & 0.0357 \\ 
Sample Cov & 0.8882 & 0.7285 & {\bf 1.0015} & 0.1145 & 0.0929 & {\bf 0.2283} & 0.0166 & 0.0163 & 0.0520 \\
\hline
\end{tabular}%
    \caption{Annualized Sharpe ratios, returns, and variance with different methods of estimating the precision matrix, with different objective functions, applied to firms that logistic regression has screened. Portfolio abbreviations as in Table~\ref{tab35-2}.}\label{tab35-3}
\end{table}

\begin{table}[htbp]
    \centering
    {\bf HUMAN ANALYSTS WITH QUANTITATIVE WEIGHTING: Oct 2021- April 2024}\\
    \begin{tabular}{l|ccc|ccc|ccc}
\hline
 & \multicolumn{3}{c|}{\textbf{Sharpe Ratio}} & \multicolumn{3}{c|}{\textbf{Returns}} & \multicolumn{3}{c}{\textbf{Variance}} \\
\textbf{Method} & \textbf{GMV} & \textbf{MV} & \textbf{MSR} & \textbf{GMV} & \textbf{MV} & \textbf{MSR} & \textbf{GMV} & \textbf{MV} & \textbf{MSR} \\ \hline
NW & 0.1359 & 0.0247 & 0.1229 & 0.0209 & 0.0037 & 0.0202 & 0.0237 & 0.0231 & 0.0270 \\
Residual NW & 0.0423 & -0.0685 & 0.1240 & 0.0055 & -0.0089 & 0.0260 & 0.0168 & 0.0169 & 0.0440 \\
Deep learning & 0.2070 & 0.0702 & 0.1844 & 0.0285 & 0.0096 & 0.0291 & 0.0190 & 0.0188 & 0.0250 \\
POET & 0.1903 & -0.0084 & \textbf{0.1910} & 0.0297 & -0.0013 & \textbf{0.0327} & 0.0243 & 0.0246 & 0.0294 \\
NLS & 0.1368 & -0.0061 & 0.0454 & 0.0165 & -0.0007 & 0.0109 & 0.0145 & \textbf{0.0144} & 0.0577 \\
Sample Cov & 0.1473 & -0.0147 & 0.0102 & 0.0185 & -0.0019 & 0.0029 & 0.0157 & 0.0160 & 0.0824 \\
\hline
\end{tabular}%
    \caption{Annualized Sharpe ratios, returns, and variance with different methods of estimating the precision matrix, with different objective functions, applied to firms that analysts have screened. Portfolio abbreviations as in Table~\ref{tab35-2}.}\label{tab35-4}
\end{table}

\begin{table}[htbp]
    \centering
    {\bf FINBERT WITH QUANTITATIVE WEIGHTING: Oct 2021-April 2024}\\
    \begin{tabular}{l|ccc|ccc|ccc}
\hline
 & \multicolumn{3}{c|}{\textbf{Sharpe Ratio}} & \multicolumn{3}{c|}{\textbf{Returns}} & \multicolumn{3}{c}{\textbf{Variance}} \\
\textbf{Method} & \textbf{GMV} & \textbf{MV} & \textbf{MSR} & \textbf{GMV} & \textbf{MV} & \textbf{MSR} & \textbf{GMV} & \textbf{MV} & \textbf{MSR} \\ \hline
NW & 0.3745 & 0.3615 & 0.4654 & 0.0727 & 0.0810 & 0.0867 & 0.0376 & 0.0503 & 0.0347 \\
Residual NW & 0.2890 & 0.4030 & \textbf{0.6937} & 0.0498 & 0.0913 & \textbf{0.1195} & 0.0297 & 0.0513 & 0.0297 \\
Deep learning & 0.4262 & 0.4190 & 0.6411 & 0.0777 & 0.0885 & 0.1115 & 0.0332 & 0.0446 & 0.0303 \\
POET & 0.4458 & 0.4343 & 0.5591 & 0.0905 & 0.0998 & 0.1093 & 0.0412 & 0.0528 & 0.0382 \\
NLS & 0.3011 & 0.4234 & 0.6415 & 0.0502 & 0.0902 & 0.1164 & \textbf{0.0278} & 0.0454 & 0.0329 \\ 
Sample Cov & 0.3407 & 0.4700 & 0.6891 & 0.0567 & 0.0993 & 0.1288 & 0.0277 & 0.0447 & 0.0350 \\
\hline
\end{tabular}%
    \caption{Annualized Sharpe ratios, returns, and variance with different methods of estimating the precision matrix, with different objective functions, applied to firms that FinBERT have screened. Portfolio abbreviations as in Table~\ref{tab35-2}.}\label{tab35-5}
\end{table}

\begin{table}[htbp]
    \centering
    {\bf LLM-S + HUMAN ANALYSTS WITH QUANTITATIVE WEIGHTING: Oct 2021-April 2024}\\
    \begin{tabular}{l|ccc|ccc|ccc}
\hline
 & \multicolumn{3}{c|}{\textbf{Sharpe Ratio}} & \multicolumn{3}{c|}{\textbf{Returns}} & \multicolumn{3}{c}{\textbf{Variance}} \\
\textbf{Method} & \textbf{GMV} & \textbf{MV} & \textbf{MSR} & \textbf{GMV} & \textbf{MV} & \textbf{MSR} & \textbf{GMV} & \textbf{MV} & \textbf{MSR} \\ \hline
NW & -0.0857 & 0.2094 & -0.0891 & -0.0161 & 0.0508 & -0.0218 & 0.0352 & 0.0587 & 0.0601 \\
Residual NW & -0.0323 & \textbf{0.3710} & 0.0270 & -0.0060 & \textbf{0.0890} & 0.0095 & 0.0346 & 0.0575 & 0.1242 \\
Deep learning & -0.0162 & 0.2625 & 0.0324 & -0.0030 & 0.0626 & 0.0082 & 0.0344 & 0.0568 & 0.0632 \\
POET & 0.2104 & 0.3526 & 0.2308 & 0.0382 & 0.0839 & 0.0461 & 0.0330 & 0.0566 & 0.0399 \\
NLS & -0.1105 & 0.3308 & 0.0329 & -0.0199 & 0.0762 & 0.0111 & \textbf{0.0323} & 0.0530 & 0.1135 \\ 
Sample Cov & -0.1101 & 0.3370 & 0.0350 & -0.0199 & 0.0778 & 0.0120 & 0.0327 & 0.0533 & 0.1181 \\
\hline
\end{tabular}%
    \caption{Annualized Sharpe ratios, returns, and variance with different methods of estimating the precision matrix, with different objective functions, applied to firms that LLM+analysts have screened, from I/B/E/S. Portfolio abbreviations as in Table~\ref{tab35-2}.}\label{tab35-6}
\end{table}

\begin{table}[htbp]
    \centering
    {\bf FINBERT + HUMAN ANALYSTS WITH QUANTITATIVE WEIGHTING: Oct 2021-April 2024}\\
    \begin{tabular}{l|ccc|ccc|ccc}
\hline
 & \multicolumn{3}{c|}{\textbf{Sharpe Ratio}} & \multicolumn{3}{c|}{\textbf{Returns}} & \multicolumn{3}{c}{\textbf{Variance}} \\
\textbf{Method} & \textbf{GMV} & \textbf{MV} & \textbf{MSR} & \textbf{GMV} & \textbf{MV} & \textbf{MSR} & \textbf{GMV} & \textbf{MV} & \textbf{MSR} \\ \hline
NW & -0.7262 & -0.7105 & -0.7036 & -0.1228 & -0.1597 & -0.1204 & 0.0286 & 0.0505 & 0.0293 \\
Residual NW & -0.7969 & -0.7343 & -0.8625 & -0.1291 & -0.1649 & -0.1650 & 0.0262 & 0.0504 & 0.0366 \\
Deep learning & -0.7069 & -0.6674 & -0.7215 & -0.1163 & -0.1500 & -0.1229 & 0.0271 & 0.0505 & 0.0290 \\
POET & -0.7250 & -0.7304 & -0.7197 & -0.1223 & -0.1624 & -0.1223 & 0.0285 & 0.0494 & 0.0289 \\
NLS & -0.5338 & -0.5130 & \textbf{-0.4832} & -0.0854 & -0.1152 & \textbf{-0.0840} & \textbf{0.0256} & 0.0504 & 0.0302 \\ 
Sample Cov & -0.6902 & -0.6567 & -0.8691 & -0.1109 & -0.1477 & -0.1783 & 0.0258 & 0.0506 & 0.0421 \\
\hline
\end{tabular}%
    \caption{Annualized Sharpe ratios, returns, and variance with different methods of estimating the precision matrix, with different objective functions, applied to firms that FinBERT+analysts have screened, from I/B/E/S. Portfolio abbreviations as in Table~\ref{tab35-2}.}\label{tab35-7}
\end{table}

\begin{table}[htbp]
    \centering
    {\bf LLM-S + FINBERT WITH QUANTITATIVE WEIGHTING: Oct 2021-April 2024}\\
    \begin{tabular}{l|ccc|ccc|ccc}
\hline
 & \multicolumn{3}{c|}{\textbf{Sharpe Ratio}} & \multicolumn{3}{c|}{\textbf{Returns}} & \multicolumn{3}{c}{\textbf{Variance}} \\
\textbf{Method} & \textbf{GMV} & \textbf{MV} & \textbf{MSR} & \textbf{GMV} & \textbf{MV} & \textbf{MSR} & \textbf{GMV} & \textbf{MV} & \textbf{MSR} \\ \hline
NW & 0.8794 & -0.2672 & 0.5380 & 0.3085 & -0.0712 & 0.1168 & 0.1230 & 0.0710 & 0.0472 \\
Residual NW & \textbf{1.0946} & -0.0282 & 0.7672 & \textbf{0.3731} & -0.0072 & 0.1839 & 0.1162 & 0.0648 & 0.0575 \\
Deep learning & 0.9858 & -0.1659 & 0.9342 & 0.3323 & -0.0424 & 0.1855 & 0.1136 & 0.0654 & \textbf{0.0394} \\
POET & 0.9376 & -0.1636 & 1.0090 & 0.3117 & -0.0436 & 0.2429 & 0.1105 & 0.0712 & 0.0579 \\
NLS & 1.0202 & -0.1199 & 0.6034 & 0.3486 & -0.0301 & 0.1557 & 0.1168 & 0.0629 & 0.0666 \\
Sample Cov & 1.0400 & -0.1049 & 0.5236 & 0.3561 & -0.0261 & 0.1384 & 0.1172 & 0.0620 & 0.0699 \\
\hline
\end{tabular}%
    \caption{Annualized Sharpe ratios, returns, and variance with different methods of estimating the precision matrix, with different objective functions, applied to firms that FinBERT+LLM has screened. Portfolio abbreviations as in Table~\ref{tab35-2}.}\label{tab35-8}
\end{table}

\begin{table}[htbp]
    \centering
    {\bf LLM-S + FINBERT + HUMAN ANALYSTS WITH QUANTITATIVE WEIGHTING: Oct 2021-April 2024}\\
    \begin{tabular}{l|ccc|ccc|ccc}
\hline
 & \multicolumn{3}{c|}{\textbf{Sharpe Ratio}} & \multicolumn{3}{c|}{\textbf{Returns}} & \multicolumn{3}{c}{\textbf{Variance}} \\
\textbf{Method} & \textbf{GMV} & \textbf{MV} & \textbf{MSR} & \textbf{GMV} & \textbf{MV} & \textbf{MSR} & \textbf{GMV} & \textbf{MV} & \textbf{MSR} \\ \hline
NW & -0.0486 & -0.2480 & -0.5475 & -0.0083 & -0.0533 & -0.1202 & 0.0295 & 0.0462 & 0.0482 \\
Residual NW & -0.2445 & -0.3935 & -0.6741 & -0.0410 & -0.0860 & -0.2345 & 0.0281 & 0.0478 & 0.1210 \\
Deep learning & -0.1105 & -0.2760 & -0.5915 & -0.0184 & -0.0593 & -0.1336 & 0.0276 & 0.0461 & 0.0510 \\
POET & \textbf{0.1117} & -0.1724 & -0.1623 & \textbf{0.0217} & -0.0392 & -0.0321 & 0.0377 & 0.0516 & 0.0390 \\
NLS & -0.2978 & -0.3918 & -0.6730 & -0.0480 & -0.0819 & -0.2274 & \textbf{0.0259} & 0.0437 & 0.1141 \\
Sample Cov & -0.3378 & -0.4296 & -0.7055 & -0.0535 & -0.0886 & -0.2435 & 0.0250 & 0.0426 & 0.1192 \\
\hline
\end{tabular}%
    \caption{Annualized Sharpe ratios, returns, and variance with different methods of estimating the precision matrix, with different objective functions, applied to firms that LLM+FinBERT+analysts have screened. Portfolio abbreviations as in Table~\ref{tab35-2}.}\label{tab35-9}
\end{table}

\section{Conclusion}\label{sec_conclusion}

This paper introduces a multi-agent Agentic AI platform for portfolio management. Our architecture coordinates three specialized
agents: an LLM-Strategy agent (LLM-S) that screens stocks annually on the basis of
fundamental firm characteristics, a FinBERT sentiment agent that screens monthly on
the basis of financial news, and a quantitative weighting method that applies
high-dimensional precision matrix estimation techniques to determine optimal portfolio
weights over the screened asset universe. We have shown that the full system, evaluated on the S\&P~500
over October 2023--April 2024, achieves a peak annualized Sharpe ratio of 8.1612, a substantial improvement over the market SR of 2.0857, while generating peak annualized returns of 52.63\%.

We have shown that performance gains depend
on the system's design. Adding human analyst recommendations to the AI ensemble
consistently degrades performance. This underscores that
Agentic AI is not simply a  powerful substitute for human judgment; rather, it is a qualitatively different conceptual framework, one that may avoid the behavioral and emotional
biases that systematically compromise human financial decision-making. 

Furthermore, we illustrate that the intersection-based decision rule, under which a stock is included in
the candidate portfolio only when both LLM-S and FinBERT agree, is responsible for
the performance gain. This demonstrates that the
hallucination-suppression role of multi-agent consensus  is a significant empirical driver of returns. The result also highlights the
economic mechanism at work: by requiring agreement between an agent that screens on
long-run fundamentals and one that screens on short-run sentiment, the system selects
firms that are simultaneously undervalued and positively perceived by the market,
combining two complementary sources of alpha. Through this economic mechanism, we show why a multi-agent system is superior to a single-agent approach.

From an econometric perspective, we make a contribution to the high-dimensional portfolio
literature. The number of assets in our portfolio is a random variable, realized
through the screening process before any weights are estimated. We have shown that under sensible screening, the squared Sharpe ratio of the screened 
portfolio consistently estimates its target, even under mild screening errors. The result holds for any precision matrix estimator
satisfying standard consistency conditions, and is therefore broadly applicable across the quantitative
portfolio formation literature.


\newpage

\setcounter{assumptionA}{0}
\setcounter{theorema}{0}
\setcounter{equation}{0}
\setcounter{lemma}{0}
\setcounter{table}{0}
\renewcommand{\theequation}{A.\arabic{equation}}
\renewcommand{\thelemma}{A.\arabic{lemma}}
\renewcommand{\thecorollary}{A.\arabic{corollary}}
\renewcommand{\thetable}{\thesection.\arabic{table}}
\renewcommand{\theassumptionA}{A.\arabic{assumptionA}}
\renewcommand{\thetheorema}{A.\arabic{theorema}}

\appendix
\section*{Appendix}
\addcontentsline{toc}{section}{Appendix}
\section{Proofs}\label{sec_A_proofs}

Investing is a two choice model, both the number of stocks as well as weights are key to this choice. When we screen the universe $p$ of stocks, our choice via screening will be a random variable $\hat{p}$. In other words,  
we estimate a bounded random variable $1\le \hat{p}\le p$ via screening.  One screening method we advocate here is Agentic AI. Screening amounts to choosing the names  of stocks, but as a simplification for our consistency analysis, this can be  done in a sensible way as defined in Section~\ref{subsec:sensible_screening}.

The aim of the screening is to select a certain stocks which may be optimal according to a strategy as in momentum based strategy or a metric like Sharpe Ratio. The optimal number of stocks in this strategy or a metric
is defined as, non-random, $p^*$, $1\le p^* \le p$.
If there is an optimal number (non-random) of stocks $p^*$,  we want to see that if our screening process is allowing mild mistakes, our estimated Sharpe Ratio with these screened stocks can achieve this target Sharpe Ratio of the optimal $p^*$ with probability approaching one.  This is a new idea to see whether a mildly successful screening can achieve Sharpe Ratio consistency. We use the following Assumption \ref{assum1} that allows mild mistake in selecting stocks.



.



\begin{assumptionA}\label{assum1}
With $\hat{p}, p^*$ growing with $n$
\[ \frac{|\hat{p} - p^*|}{p^*} = o_p (1)
.\]

    \end{assumptionA}
Since $\hat{p}, p^*$ are positive integers and grow with $n$, this assumption amounts to a slight mistake of a constant difference between 
$\hat{p},p^*$. Assumption \ref{assum1} is used instead of the restrictive
consistent integer estimation, which is  perfect selection of  the number of stocks in a portfolio
\[ \lim_{n \to \infty} P \left[ \left|\frac{\hat{p}}{p^*}=1\right| \right] =1.\]
Even though a mistake can be made in selecting the correct number of stocks (i.e. hence stock indexes), this may not affect downstream Sharpe 
Ratio estimation consistency. This is a novel result and a new proof. We show this in Lemma \ref{la1}-\ref{la2}, and the proof of Theorem~\ref{thm:sr_convergence} below.

Next we show that precision matrix, $\Gamma_{p^*}: p^* \times p^*$ can be estimated consistently by $\hat{\Gamma}_p^*: p^* \times p^*$ at a given $p^*$. Define the maximum row sum of a matrix $A: m \times n$ as 
$ \| A \|_{l_{\infty}}= \max_{1 \le i \le m} \sum_{j=1}^n |A_{i,j}|$,
where $A_{i,j}$ represents (i,j) th element of matrix A.

\begin{assumptionA}\label{assum2}
\[ \| \hat{\Gamma}_{p^*} - \Gamma_{p^*}\|_{l_{\infty}} =o_p(1).\]
    \end{assumptionA}

    For example both \cite{caner2019}, \cite{caner2022} prove Assumption \ref{assum2} under weaker assumptions. Spectral norm consistency such as in \cite{fan2013}, \cite{c-d2025} can also be used, and this will change the proofs but not the result of consistency we conjecture.

Next we analyze Global Minimum Variance portfolio and specifically its Sharpe Ratio (SR from now on). The aim of the portfolio is to minimize the variance of the portfolio of assets. We will show that if we do  screen stocks, the high-dimensional consistency of SR will  be obtained.  
Denote the precision matrix estimator after screening as 
$\hat{\Gamma}_{\hat{p}}: \hat{p} \times \hat{p}$, and the mean return estimator as $\hat{\mu}_{\hat{p}}:\hat{p} \times 1$. These estimators can come from any technique. 
Estimated Sharpe Ratio  of Global Minimum Variance Portfolio is 
    \[ \widehat{SR}_{\hat{p}}:= \sqrt{\hat{p}} \left( \frac{1_{\hat{p}}' \hat{\Gamma}_{\hat{p}} \hat{\mu}_{\hat{p}}}{\hat{p}}
\right)
\left( \frac{1_{\hat{p}}' \hat{\Gamma}_{\hat{p}} 1_{\hat{p}}}{\hat{p}}
\right)^{-1/2},
    \]
which is defined in (27) of \cite{caner2022} with non-random $p$ instead of $\hat{p}$. In \cite{caner2022}, as well as all the literature we know of, the estimated SR is using a fixed-non-random large universe of assets $p$. Here since there is a screening process before estimating SR, the dimension of the chosen portfolio is $\hat{p} \le p$. 

 Denote the target precision matrix   as 
$\Gamma_{p^*}: p^* \times p^*$, and the mean return is $\mu_{p^*}:p^* \times 1$. The target Sharpe Ratio is:
\[ SR_{p^*}:= \sqrt{p^*} \left( \frac{1_{p^*}' \Gamma_{p^*} \mu_{p^*}}{p^*}
\right)
\left( \frac{1_{p^*}' \Gamma_{p^*} 1_{p^*}}{p^*}
\right)^{-1/2},\]
which is defined in (26) of \cite{caner2022} but with $p$ the universe of assets, rather than target screened number of assets $p^*$. In \cite{caner2022}, as well as the other literature such as \cite{caner2019},\cite{c-d2025}, target SR is for $p$, universe of assets. Here we use the target number of screened stocks, $p^*$.

The main technical complication is even though screening process can allow mistakes, it is not clear a major financial metric such as Sharpe Ratio  can be consistently estimated in high dimensions.

In order to show the consistency of Sharpe Ratio estimate, we first need the following Lemma. It will show consistency of variance estimator for Global Minimum Variance portfolio. We have two possibilities, either random variable $\hat{p}$ is larger than or equal to target $p^*$, or smaller than equal to the target, $p^*$. For each possibility, we need a norm constraint on blocks of precision matrix estimates.

 First, it is possible to overshoot the target, implying  $\hat{p}\ge p^*$.
Decompose $\hat{\Gamma}_{\hat{p}}: \hat{p} \times \hat{p}$ into four blocks
\begin{equation}
 \hat{\Gamma}_{\hat{p}} = \left[ \begin{array}{cc}
    \hat{\Gamma}_{p^*} & \hat{\Gamma}_{p^*, 1} \\
    \hat{\Gamma}_{p^*,2} & \hat{\Gamma}_{p^*,3}
    \end{array}\right],\label{a0}
    \end{equation}
where $\hat{\Gamma}_{p^*}:p^* \times p^*$, $\hat{\Gamma}_{p^*,1}: p^* \times \hat{p}- p^*$, $\hat{\Gamma}_{p*,2}: \hat{p}-p^* \times p^* $, $\hat{\Gamma}_{p^*,3}:
\hat{p}- p^* \times \hat{p} - p^*$.

Other possibility is we may undershoot the target $\hat{p} \le p^*$,and for that we decompose the precision matrix estimator 
$\hat{\Gamma}_{p^*}: p^* \times p^*$ as 
\begin{equation} 
\hat{\Gamma}_{p^*} = \left[ \begin{array}{cc}
    \hat{\Gamma}_{\hat{p}} & \hat{\Gamma}_{\hat{p}, 1} \\
    \hat{\Gamma}_{\hat{p},2} & \hat{\Gamma}_{\hat{p},3}
    \end{array}\right].\label{a00}
    \end{equation}
Block dimensions are:
 $\hat{\Gamma}_{\hat{p}}: \hat{p} \times \hat{p}$,
$\hat{\Gamma}_{\hat{p},1}: \hat{p} \times (p^* - \hat{p})$, $\hat{\Gamma}_{\hat{p},2}: (p^* - \hat{p}) \times \hat{p}$, and 
$ \hat{\Gamma}_{\hat{p},3}: (p^* - \hat{p})\times (p^* - \hat{p})$.

Condition A.1 imposes that the blocks that emanate from over or undershooting in screening can have maximum row sums bounded by a constant. This is a reasonable condition imposed since either row or columns are small in number. However, we should note that either with high $\hat{p}$ or $p^*$, this constraint may become more strict. This condition can be relaxed to allow diverging sums if we know the rate of convergence of $\hat{p}$ to $p^*$ in consistency.\\

{\bf Condition A.1}.{\it 
When (i).$\hat{p} \ge p^*$ we need $\| \hat{\Gamma}_{p^*,m}\|_{l_{\infty}}\le C < \infty$, $m=1,2,3$.
When (ii). $\hat{p}\le  p^*$ we need 
$\| \hat{\Gamma}_{\hat{p},m}\|_{l_{\infty}}\le C < \infty$,$m=1,2,3$.
}\\

\noindent Note that Condition A.1 can be relaxed in the following way. For $\hat{p}\ge p^*$ we can have $\| \hat{\Gamma}_{p^*,m}\|_{l_{\infty}}\le C < \infty$ with probability approaching one, and same for (ii) too. The proofs will not change.

\noindent Now we introduce following Lemma \ref{la1}  that shows GMV variance consistency.

\begin{lemma}\label{la1}

   Under Assumptions \ref{assum1}-\ref{assum2} with Condition A.1
we have 
\[ \frac{| 1_{  \hat{p}}' \hat{\Gamma}_{\hat{p}} 1_{ \hat{p}} - 
1_{p^*}' \Gamma_{p^*} 1_{p^*}|}{p^*}=o_p(1).\]

\end{lemma}

{\bf Proof of Lemma \ref{la1}}

We start the proof by the analysis of 
\begin{equation}
 | 1_{  \hat{p}}' \hat{\Gamma}_{\hat{p}} 1_{ \hat{p}} - 
1_{p^*}' \Gamma_{p^*} 1_{p^*}|.\label{pla1-1}
   \end{equation}.

   {\it Step 1.} To consider (\ref{pla1-1}) we start with $\hat{p}\ge p^*$ case. In that scenario we can decompose the vector of ones  and the precision matrix estimator as  follows

\begin{equation}
 1_{\hat{p}} = \left( \begin{array}{c}
1_{p^*} \\
1_{\hat{p} - p^*}
    \end{array}
    \right),\quad 
    \hat{\Gamma}_{\hat{p}} = \left[ \begin{array}{cc}
    \hat{\Gamma}_{p^*} & \hat{\Gamma}_{p^*, 1} \\
    \hat{\Gamma}_{p^*,2} & \hat{\Gamma}_{p^*,3}
    \end{array}\right]\label{pla1-2}
\end{equation}
where vector of ones is decomposed into $p^*, \hat{p}-p^*$, parts $1_{p^*}, 1_{\hat{p}-p^*}$ respectively. Similarly the precision matrix estimator is decomposed into four blocks, explained in (\ref{a0}).
Then rewrite (\ref{pla1-1}) by adding and subtracting $1_{p^*}' \hat{\Gamma}_{p^*} 1_{p^*}$
\begin{eqnarray}
 | 1_{  \hat{p}}' \hat{\Gamma}_{\hat{p}} 1_{ \hat{p}} - 
1_{p^*}' \Gamma_{p^*} 1_{p^*}|
& = & | 1_{\hat{p}}' \hat{\Gamma}_{\hat{p}} 1_{\hat{p}}
- 1_{p^*}' \hat{\Gamma}_{p^*} 1_{p^*} + 1_{p^*}' \hat{\Gamma}_{p^*} 1_{p^*}
-1_{p^*}' \Gamma_{p^*} 1_{p^*}| \nonumber \\
& \le &
| 1_{  \hat{p}}' \hat{\Gamma}_{\hat{p}} 1_{ \hat{p}} - 
1_{p^*}' \hat{\Gamma}_{p^*} 1_{p^*}| + 
| 1_{p^*}' ( \hat{\Gamma}_{p^*}- \Gamma_{p^*}) 1_{p^*}|.
\label{pla1-3}
   \end{eqnarray}
Consider the analysis of the first term on the right side of (\ref{pla1-3}), use (\ref{a0})(\ref{pla1-2}) and triangle inequality

\begin{eqnarray}
| 1_{  \hat{p}}' \hat{\Gamma}_{\hat{p}} 1_{ \hat{p}}& -& 
1_{p^*}' \hat{\Gamma}_{p^*} 1_{p^*}|
\le | 1_{\hat{p}-p^*}' \hat{\Gamma}_{p^*,2} 1_{p^*} | \nonumber \\
& + & | 1_{p^*}' \hat{\Gamma}_{p^*,1} 1_{\hat{p} - p^*}|
+ | 1_{\hat{p}- p^*}' \hat{\Gamma}_{p^*,3} 1_{ \hat{p} - p^*}|.\label{pla1-4}
\end{eqnarray}

Analyze each term in (\ref{pla1-4}). Start with the first right side term
\begin{eqnarray}
 |1_{\hat{p}-p^*}' \hat{\Gamma}_{p^*,2} 1_{p^*} | & \le & 
 \| 1_{\hat{p}- p^*}'\|_1 \| \hat{\Gamma}_{p^*,2} 
  1_{p^*} \|_{\infty} \nonumber \\
 & \le & \| 1_{\hat{p}- p^*} \|_{1} \| \hat{\Gamma}_{p^*,2}
\|_{l_{\infty}}  \| 1_{p^*} \|_{\infty} \nonumber \\
& = & | \hat{p} - p^*| \| \hat{\Gamma}_{p^*,2} \|_{l_{\infty}},\label{pla1-5}
\end{eqnarray}
where we use Holder's inequality for the first inequality, and for the second inequality we use p.345 of \cite{hj2013}. 

Start with the second right side term, by taking a transpose of that term first in absolute terms
\begin{eqnarray}
 |1_{\hat{p}-p^*}' \hat{\Gamma}_{p^*,1}' 1_{p^*} | & \le & 
 \| 1_{\hat{p}- p^*}' \hat{\Gamma}_{p^*,1}' \|_1 
 \| 1_{p^*} \|_{\infty} \nonumber \\
 & \le & \| 1_{\hat{p}- p^*} \|_{1} \| \hat{\Gamma}_{p^*,1}
\|_{l_{\infty}}  \nonumber \\
& \le & | \hat{p} - p^*| \| \hat{\Gamma}_{p^*,2} \|_{l_{\infty}},\label{pla1-5a}
\end{eqnarray}
where we use Holder's inequality for the first inequality, and for the second inequality we use p.345 of \cite{hj2013}, and $\| A' \|_{l_1}= \| A \|_{l_{\infty}}$, for a generic matrix A.

The analysis for the third term right side term in (\ref{pla1-4}) is very same as in the first right side term analysis in (\ref{pla1-5}) , so we 
have 
\begin{equation}
\frac{1}{p^*}| 1_{  \hat{p}}' \hat{\Gamma}_{\hat{p}} 1_{ \hat{p}} -
1_{p^*}' \hat{\Gamma}_{p^*} 1_{p^*}|=o_p(1),\label{pla1-6}
\end{equation}
given Assumption \ref{assum1}, and by Condition A.1.
 For the second term on the right side of (\ref{pla1-3}).
\begin{eqnarray}
 | 1_{p^*}' ( \hat{\Gamma}_{p^*}- \Gamma_{p^*}) 1_{p^*}|
 & \le & \| 1_{p^*}' (\hat{\Gamma}_{p^*} - \Gamma_{p^*}) \|_{{\infty}}
\| 1_{p^*} \|_1 \nonumber \\
& = & \| 1_{p^*} \|_{\infty} \| \hat{\Gamma}_{p^*}
- \Gamma_{p^*} \|_{l_{\infty}} p^*,\label{pla1-7}
\end{eqnarray}
where we use Holder's inequality for the first inequality, and for the second inequality we use p.345 of \cite{hj2013}.  Then use Assumption \ref{assum2} to have 
\begin{equation}
\left| \frac{1_{p^*}' (\hat{\Gamma}_{p^*} - \Gamma_{p^*}) 1_{p^*}}
{p^*}\right| = o_p (1).\label{pla1-8}
\end{equation}
Combine (\ref{pla1-6})(\ref{pla1-8}) in (\ref{pla1-3}) to have 
\begin{equation}
\left|\frac{1_{\hat{p}}' \hat{\Gamma}_{\hat{p}} 1_{\hat{p}} - 
1_{p^*}' \Gamma_{p^*} 1_{p^*}}{p^*}\right| = o_p (1).\label{pla1-9}
    \end{equation}

{\it Step 2}. Now we consider the opposite case, $\hat{p}\le p^*$.
The indicator can be decomposed as

\begin{equation}
 1_{p^*} = \left( \begin{array}{c}
1_{\hat{p}} \\
1_{p^{*} -\hat{p} }
    \end{array}
    \right),\label{pla1-10}
\end{equation}
Note that the vector ones decomposed into two subvectors of dimension $\hat{p}, p^* - \hat{p}$ respectively.

\noindent By adding and subtracting  $1_{p^*}' \hat{\Gamma}_{p^*} 1_{p^*}$, we begin with a triangle inequality
\begin{equation}
| 1_{p^*}' \Gamma_{p^{*}} 1_{p^* } - 1_{\hat{p}}' \hat{\Gamma}_{\hat{p}}
1_{ \hat{p}}| \le 
| 1_{p^*}' \Gamma_{p^*} 1_{p^*}- 1_{p^*}' \hat{\Gamma}_{p^*} 1_{p^*}|
+ | 1_{p^*}' \hat{\Gamma}_{p^*} 1_{p^*} - 1_{\hat{p}}' \hat{\Gamma}_{\hat{p}} 1_{\hat{p}}|.\label{pla1-11}
\end{equation}
We simplify the second term on the right side in (\ref{pla1-11}) using (\ref{a00})(\ref{pla1-10})
\begin{eqnarray}
| 1_{p^*}' \hat{\Gamma}_{p^*} 1_{p^*} - 1_{\hat{p}}' \hat{\Gamma}_{\hat{p}} 1_{\hat{p}}| &=& 
| 1_{\hat{p}}' \hat{\Gamma}_{\hat{p}} 1_{\hat{p}} + 
1_{p^* - \hat{p}}' \hat{\Gamma}_{\hat{p},2} 1_{\hat{p}} + 
1_{\hat{p}}' \hat{\Gamma}_{\hat{p},1} 1_{p^* - \hat{p}} + 
1_{p^* - \hat{p}}' \hat{\Gamma}_{\hat{p},3} 1_{ p^* - \hat{p}} -
1_{\hat{p}}' \hat{\Gamma}_{\hat{p}} 1_{\hat{p}}| \nonumber \\
& \le & |1_{p^* - \hat{p}}' \hat{\Gamma}_{\hat{p},2} 1_{\hat{p}}| + 
|1_{\hat{p}}' \hat{\Gamma}_{\hat{p},1} 1_{p^* - \hat{p}} |+ 
|1_{p^* - \hat{p}}' \hat{\Gamma}_{\hat{p},3} 1_{ p^* - \hat{p}}|.\label{pla1-12}
\end{eqnarray}
where the first and fifth elements cancel in the equality in (\ref{pla1-12}).

\noindent In (\ref{pla1-12}) consider the first right side term in the last inequality

\begin{eqnarray}
    |1_{p^* - \hat{p}}' \hat{\Gamma}_{\hat{p},2} 1_{\hat{p}}|
    & \le & \|1_{p^* - \hat{p}}\|_1 \| \hat{\Gamma}_{\hat{p},2}     
    1_{\hat{p}}\|_{\infty} \nonumber \\
    & \le & \|1_{p^* - \hat{p}}\|_1 \| \hat{\Gamma}_{\hat{p},2} \|_{l_{\infty}}   
    \nonumber \\
        & \le & (p^* - \hat{p}) \|\hat{\Gamma}_{\hat{p},2} \|_{l_{\infty}} ,\label{pla1-13a}
    \end{eqnarray}
    where  we use Holder's inequality for the first inequality, and p.345 of \cite{hj2013} for the second inequality.
     The second term on the right side of (\ref{pla1-12})  is transposed
\begin{eqnarray}
|  1_{p^* - \hat{p}}' \hat{\Gamma}_{\hat{p},1}'1_{\hat{p}}|
& \le & \|  1_{p^* - \hat{p}}' \hat{\Gamma}_{\hat{p},1}'\|_1
\|   1_{\hat{p}}\|_{\infty} \nonumber \\
& \le & \|  1_{p^* - \hat{p}}' \|_1 \|\hat{\Gamma}_{\hat{p},1}'\|_{l_1}
\nonumber \\
& = & (p^* - \hat{p})  \| \hat{\Gamma}_{\hat{p},1}\|_{l_{\infty}},\label{pla13-ab}
\end{eqnarray} 
where the first inequality is by Holder's inequality, and the second one is by p.345 of \cite{hj2013}, and the equality is by $\| A' \|_{l_1}= \| A \|_{l_{\infty}}$.
     The third right side term in the last inequality in (\ref{pla1-12}) is  handled in the same way as in (\ref{pla1-13a}). So via Assumption \ref{assum1}  with 
     the Condition A.1 about block matrices 
   $\|\hat{\Gamma}_{\hat{p},m} \|_{l_{\infty}} \le C < \infty$, $m=1,2,3$
the scaled second term on the right side of (\ref{pla1-11}) is
\begin{equation}
    \frac{ | 1_{p^*}' \hat{\Gamma}_{p^*} 1_{p^*} - 1_{\hat{p}}' \hat{\Gamma}_{\hat{p}} 1_{\hat{p}}|}{p^*} = o_p (1).\label{pla1-13}
\end{equation}

Next, use (\ref{pla1-8}) on the first right side term in (\ref{pla1-11})
\begin{equation}
    \frac{| 1_{p^*}' \Gamma_{p^*} 1_{p^*}- 1_{p^*}' \hat{\Gamma}_{p^*} 1_{p^*}|}{p^*} = o_p(1).\label{pla1-14}
\end{equation}
Combine the last two equations on the left side term of (\ref{pla1-11})
\begin{equation}
\frac{| 1_{p^*}' \Gamma_{p^{*}} 1_{p^* } - 1_{\hat{p}}' \hat{\Gamma}_{\hat{p}}
1_{ \hat{p}}|}{p^*} = o_p (1). \label{pla1-15}
\end{equation}

{\it Step 3}. Combine Steps 1-2 (\ref{pla1-9})(\ref{pla1-15}) to have the desired result.
{\bf Q.E.D.}

We provide the following assumptions for  the next lemma.

\begin{assumptionA}\label{assum3}

(i). \[ \max_{1 \le j \le p} | \mu_j | \le C < \infty.\]

(ii). \[ \max_{1 \le j \le p} \| \hat{\mu}_{\hat{p}} - \mu_p \|_{\infty}
    =o_p(1).\]
\end{assumptionA}

This assumption is for all the universe of assets $p$, which is non-random.
Note that Assumption \ref{assum3}(i) is standard in the literature and can be seen in \cite{caner2022}.  Assumption \ref{assum3}(ii) can be proven under weaker conditions and a rate of convergence can be obtained. This is shown in \cite{caner2022}, and also in this assumption we let $\hat{\mu}_p$ any consistent estimator of $\mu_p$. Next assumption is a norm bound on the target precision  matrix.

\begin{assumptionA}\label{assum4}

\[ \| \Gamma_{p^*}\|_{l_{\infty}} \le C < \infty.\]

\end{assumptionA}

Assumption \ref{assum4} puts a finite bound on maximum row sums of a $p^* \times p^*$ target precision matrix. In case of $p^*$ growing with $n$ assumption is restrictive. But in that scenario, for a specific estimator like residual nodewise regression of \cite{caner2022} the proof may  depend on the joint product of estimation of error mean estimator as in Assumption \ref{assum3}(ii) (but with $p^*$ dimension, $p^*\le p$) multiplied by the target precision matrix in Assumption \ref{assum4}. Then we can have weaker assumptions  that shows divergence of the sum in Assumption \ref{assum4}, but multiplied by the rate in Assumption \ref{assum3}(ii) can converge to zero in probability. In other words, it is possible to have 
\[ \max_{1 \le j \le p^*} \| \hat{\mu}_{p^*} - \mu_{p^*} \|_{\infty}
= o_p ( d_n), \quad d_n \to 0,\]
\[ \| \Gamma_{p^*} \|_{l_{\infty}} = O (r_n), \quad r_n \to \infty \]
but $d_n r_n \to 0$. For details, see Sharpe Ratio estimation for non-random, no screening GMV portfolio of \cite{caner2022} proof.
Since we have  a general estimation setup, we impose more strict assumptions. Specific estimators can relax these assumptions.

\begin{lemma}\label{la2}

Under Assumptions \ref{assum1}-\ref{assum4} with Condition A.1

\[ \frac{ \left|1_{\hat{p}}' \hat{\Gamma}_{\hat{p}} \hat{\mu}_{\hat{p}}
- 1_{p^*}' \Gamma_{p^*} \mu_{p^*}\right|}{p^*} =o_p(1).
\]
    
\end{lemma}

{\bf Proof of Lemma \ref{la2}.}

We start with the numerator
by adding and subtracting $1_{p^*}' \hat{\Gamma}_{p^*} \hat{\mu}_{p^*}$
and triangle inequality

\begin{equation}\label{pla2-1}
| 1_{\hat{p}}' \hat{\Gamma}_{\hat{p}} \hat{\mu}_{\hat{p}}
- 1_{p^*}' \Gamma_{p^*} \mu_{p^*}|
 \le  
| 1_{\hat{p}}' \hat{\Gamma}_{\hat{p}} \hat{\mu}_{\hat{p}} -
1_{p^*}' \hat{\Gamma}_{p^*} \hat{\mu}_{p^*} | 
+ | 1_{p^*}' \hat{\Gamma}_{p^*} \hat{\mu}_{p^*}
- 1_{p^*}' \Gamma_{p^*} \mu_{p^*}|.
\end{equation}

In the next two steps, first we consider $\hat{p} \ge p^*$, and then $\hat{p}\le p^*$.

{\it Step 1}. We start with case $\hat{p}\ge p^*$. 

We can decompose  the estimator of the mean vector 
\[
\hat{\mu}_{\hat{p}}= \left[  \begin{array}{c}
\hat{\mu}_{p^*} \\
\hat{\mu}_{\hat{p} - p^*}
\end{array}
\right],
\]
where $\hat{\mu}_{p^*}: p^* \times 1$, and $\hat{\mu}_{\hat{p} - p^*}: \hat{p} - p^* \times 1$. Now substitute this decomposition immediately above, and for $\hat{\Gamma}_{\hat{p}}$ in Step 1 in Lemma \ref{la1} proof, using (\ref{a0})(\ref{pla1-2})
\begin{equation}\label{pla2-2}
1_{\hat{p}}' \hat{\Gamma}_{\hat{p}} \hat{\mu}_{\hat{p}}
= 1_{p^*}' \hat{\Gamma}_{p^*} \hat{\mu}_{p^*} + 1_{\hat{p}-p^*}'
\hat{\Gamma}_{p^*,2} \hat{\mu}_{p^*} + 
1_{p^*}' \hat{\Gamma}_{p^*,1} \hat{\mu}_{\hat{p}-p^*} 
+ 1_{\hat{p}-p^*} \hat{\Gamma}_{p^*,3} \hat{\mu}_{\hat{p}-p^*}.
\end{equation}

\noindent Consider the first term on the right side of (\ref{pla2-1}), with (\ref{pla2-2}) and triangle inequality
\begin{eqnarray}
| 1_{\hat{p}}' \hat{\Gamma}_{\hat{p}} \hat{\mu}_{\hat{p}} -
1_{p^*}' \hat{\Gamma}_{p^*} \hat{\mu}_{p^*} | & \le & 
| 1_{\hat{p} - p^*}' \hat{\Gamma}_{p^*,2} \hat{\mu}_{p^*}| \nonumber \\
& + & |1_{p^*}' \hat{\Gamma}_{p^*,1} \hat{\mu}_{\hat{p}-p^*} |
+ | 1_{\hat{p} - p^*}' \hat{\Gamma}_{p^*,3} \hat{\mu}_{\hat{p}-p^*}|.\label{pla2-3}
\end{eqnarray}

\noindent Before the next proof, 
\begin{equation}
   \| \hat{\mu}_{p^*}\|_{\infty} \le \| \hat{\mu}_p \|_{\infty} 
   \le \| \hat{\mu}_{p} - \mu_{p} \|_{\infty}
   + \|\mu_{p} \|_{\infty} = o_p (1) + O (1),\label{pla2-4}
   \end{equation}
where we use $p^* \le p$ and Assumption \ref{assum3}.

Analyze the first term on the right side of (\ref{pla2-3})
\begin{eqnarray}
| 1_{\hat{p} - p^*}' \hat{\Gamma}_{p^*,2} \hat{\mu}_{p^*}| & \le & 
\| 1_{\hat{p} - p^*} \|_1 \|\hat{\Gamma}_{p^*,2} \hat{\mu}_{p^*}\|_{\infty}
\nonumber \\
& \le & ( \hat{p} - p^*) \|\hat{\Gamma}_{p^*,2} \|_{l_{\infty}}
\| \hat{\mu}_{p^*}\|_{\infty} \nonumber \\
& = & (\hat{p}- p^*) O(1) o_p (1),\label{pla2-5}
\end{eqnarray}
where we use Holder's inequality for the first inequality, and p.345 of \cite{hj2013} for the second inequality, and the rates come from 
Condition A.1, and (\ref{pla2-4}). Then  via Assumption \ref{assum1}
\begin{equation}\label{pla2-6}
\frac{| 1_{\hat{p} - p^*}' \hat{\Gamma}_{p^*,2} \hat{\mu}_{p^*}|}{p^*}=o_p(1).
\end{equation}

Next, we consider second term on the right side of (\ref{pla2-3})
\begin{eqnarray}
|1_{p^*}' \hat{\Gamma}_{p^*,1} \hat{\mu}_{\hat{p}-p^*} | &= &
| \hat{\mu}_{\hat{p}-p^*}'\hat{\Gamma}_{p^*,1}'1_{p^*}'| \nonumber \\
& \le & 
\| \hat{\mu}_{\hat{p}-p^*}'\hat{\Gamma}_{p^*,1}'\|_1 
\|1_{p^*} \|_{\infty} \nonumber \\
& \le & \| \hat{\mu}_{\hat{p}-p^*}\|_1 \| \hat{\Gamma}_{p^*,1}'\|_{l_1}
= \| \hat{\mu}_{\hat{p}-p^*}\|_1 \| \hat{\Gamma}_{p^*,1}\|_{l_{\infty}} \nonumber \\
& \le & (\hat{p} - p^*) [ \max_{1 \le j \le \hat{p} - p^*} | \hat{\mu}_j |
] \| \hat{\Gamma}_{p^*,1} \|_{l_{\infty}} \nonumber \\
& \le & (\hat{p} - p^*)
[\max_{1 \le j \le p} | \hat{\mu}_j|]
\| \hat{\Gamma}_{p^*,1} \|_{l_{\infty}} \nonumber \\
& = & (\hat{p} - p^*) O_p (1) O(1),\label{pla2-7}
\end{eqnarray}
where we use Holder's inequality for the first inequality, and p.345 of \cite{hj2013} for the second inequality, and for the second equality we use 
$\| A' \|_{l_1} = \| A \|_{l_{\infty}}$, and for the third inequality we use $l_1- l_{\infty}$ inequality ($\| v \|_1 \le dim (v) \| v \|_{\infty}$, with dim(v) showing the dimension of vector v), and the fourth inequality by $\hat{p} - p^* \le p$, and the rates are by Condition A.1 and (\ref{pla2-4}). By (\ref{pla2-7}) and Assumption \ref{assum1}

\begin{equation}
    \frac{|1_{p^*}' \hat{\Gamma}_{p^*,1} \hat{\mu}_{\hat{p}-p^*} |}{p^*}=o_p (1).\label{pla2-8}
    \end{equation}
Third term on the right side of (\ref{pla2-3}) is considered in the same way as in (\ref{pla2-8}), so combining (\ref{pla2-6})(\ref{pla2-8}) with third term analysis we have in left side term of (\ref{pla2-3})
\begin{equation}
\frac{| 1_{\hat{p}}' \hat{\Gamma}_{\hat{p}} \hat{\mu}_{\hat{p}} -
1_{p^*}' \hat{\Gamma}_{p^*} \hat{\mu}_{p^*} |}{p^*} = o_p (1).\label{pla2-9}
\end{equation}

Now consider the second term on the right side of (\ref{pla2-1}). By adding and subtracting  $1_{p^*}' \Gamma_{p^*} \hat{\mu}_{p^*}$ via triangle inequality

\begin{equation}
    | 1_{p^*}' \hat{\Gamma}_{p^*} \hat{\mu}_{p^*} - 
    1_{p^*}' \Gamma_{p^*} \mu_{p^*} | 
    \le 
    | 1_{p^*}' (\hat{\Gamma}_{p^*} - \Gamma_{p^*} ) \hat{\mu}_{p^*}|
+ | 1_{p^*}' \Gamma_{p^*} (\hat{\mu}_{p^*} - \mu_{p^*})|.\label{pla2-10}
    \end{equation}

    First, we consider the first right side term in (\ref{pla2-10})
    \begin{eqnarray}
| 1_{p^*}' (\hat{\Gamma}_{p^*} - \Gamma_{p^*} ) \hat{\mu}_{p^*}|
& = & | \hat{\mu}_{p^*}' ( \hat{\Gamma}_{p^*} - \Gamma_{p^*})' 1_{p^*}| \nonumber \\
& \le & \| \hat{\mu}_{p^*}' ( \hat{\Gamma}_{p^*} - \Gamma_{p^*})' \|_1
\| 1_{p^*} \|_{\infty} \nonumber \\
& \le & \| \hat{\mu}_{p^*} \|_1  \| ( \hat{\Gamma}_{p^*} - \Gamma_{p^*})' \|_{l_1} \nonumber \\
& \le  & 
p^* [  \max_{1 \le j \le p^*} | \hat{\mu}_j |
] \| \hat{\Gamma}_{p^*} - \Gamma_{p^*} \|_{l_{\infty}}  \nonumber \\
& = & p^* O_p (1) o_p (1), \label{pla2-11}
\end{eqnarray}
where we use Holder's inequality for the first inequality, and p.345 of \cite{hj2013}  for the second inequality, and for the third inequality we use $\| A' \|_{l_1}= \| A \|_{l_{\infty}}$, and $l_1 - l_{\infty}$ vector norm inequality, and for the rates we use the same analysis in (\ref{pla2-4}) with Assumption \ref{assum2}.

Next, consider the second term on the right side of (\ref{pla2-10})

\begin{eqnarray}
    | 1_{p^*}' \Gamma_{p^*} (\hat{\mu}_{p^*} - \mu_{p^*})|
    & \le & \| 1_{p^*}' \Gamma_{p^*} \|_1 \| \hat{\mu}_{p^*} - 
    \mu_{p^*} \|_{\infty} \nonumber \\
    & \le & \| 1_{p^*} \|_1 \| \Gamma_{p^*} \|_{l_1}
\| \hat{\mu}_{p^*} - \mu_{p^*} \|_{\infty} \nonumber \\
& = & p^* \| \Gamma_{p^*} \|_{l_{\infty}} 
\| \hat{\mu}_{p^*} - \mu_{p^*} \|_{\infty} \nonumber \\
& = & p^* O(1) o_p (1),\label{pla2-12}
\end{eqnarray}
where we use Holder's inequality for the first inequality, and p.345 of \cite{hj2013} for the second inequality,  and the equality is by symmetricity of $\Gamma$,
and the rate is by Assumptions \ref{assum3}-\ref{assum4} and $p^* \le p$.
 Combine (\ref{pla2-11})(\ref{pla2-12}) into (\ref{pla2-10}) 
\begin{equation}
\frac{ | 1_{p^*}' \hat{\Gamma}_{p^*} \hat{\mu}_{p^*}
- 1_{p^*}' \Gamma_{p^*} \mu_{p^*}|}{p^*} = o_p (1).\label{pla2-13}
\end{equation}

\noindent Combine (\ref{pla2-9})(\ref{pla2-13}) into (\ref{pla2-1})
to have 

\begin{equation}
    \frac{| 1_{\hat{p}}' \hat{\Gamma}_{\hat{p}} \hat{\mu}_{\hat{p}}
    - 1_{p^*}' \Gamma_{p^*} \mu_{p^*}|
    }{p^*} = o_p (1).\label{pla2-14}
\end{equation}

{\it Step 2}. 

Consider (\ref{pla2-1}) but the case of $p^*\ge \hat{p}$. In that scenario decompose the mean estimator
\[
\hat{\mu}_{p^*}:= \left[ \begin{array}{c}
\hat{\mu}_{\hat{p}} \\
\hat{\mu}_{p^* - \hat{p}}
\end{array}
\right],
\]
where $\hat{\mu}_{\hat{p}}: \hat{p} \times 1$ and $\hat{\mu}_{p^*- \hat{p}}:p^* - \hat{p} \times 1$ dimension. We see that  by (\ref{pla1-10})
\begin{eqnarray}
1_{p^*}' \hat{\Gamma}_{p^*} \hat{\mu}_{p^*} & = & 
(1_{\hat{p}}', 1_{p^* - \hat{p}}')
\left[ \begin{array}{cc}
\hat{\Gamma}_{\hat{p}}, \hat{\Gamma}_{\hat{p},1} \\
\hat{\Gamma}_{\hat{p},2}, \hat{\Gamma}_{\hat{p},3}
    \end{array}
\right] 
\left[ \begin{array}{c}
\hat{\mu}_{\hat{p}} \\
\hat{\mu}_{p^* - \hat{p}}
\end{array}
\right] \nonumber \\
& = & 1_{\hat{p}}' \hat{\Gamma}_{\hat{p}} \hat{\mu}_{\hat{p}} 
+ 1_{p^* - \hat{p}}' \hat{\Gamma}_{\hat{p},2} \hat{\mu}_{\hat{p}} \nonumber \\
& + & 
1_{\hat{p}}' \hat{\Gamma}_{\hat{p},1} \hat{\mu}_{p^* - \hat{p}} + 
1_{ p^* - \hat{p}}' \hat{\Gamma}_{\hat{p},3} \hat{\mu}_{p^* - \hat{p}},\label{pla2-15}
 \end{eqnarray}
where we use $\hat{\Gamma}_{p^*}$ decomposition before Condition A.1.
Then the first term on the right side of (\ref{pla2-1}) simplifies via (\ref{pla2-15})
\begin{eqnarray}
    | 1_{p^*}' \hat{\Gamma}_{p^*} \hat{\mu}_{p^*} - 1_{\hat{p}}'
    \hat{\Gamma}_{\hat{p}} \hat{\mu}_{\hat{p}}| 
    & = & | 1_{p^* - \hat{p}}' \hat{\Gamma}_{\hat{p},2} \hat{\mu}_{\hat{p}}
  +  1_{\hat{p}}' \hat{\Gamma}_{\hat{p},1} \hat{\mu}_{p^* - \hat{p}} + 
1_{ p^* - \hat{p}}' \hat{\Gamma}_{\hat{p},3} \hat{\mu}_{p^* - \hat{p}}|
\nonumber \\
& \le & 
| 1_{p^* - \hat{p}}' \hat{\Gamma}_{\hat{p},2} \hat{\mu}_{\hat{p}}
| +  |1_{\hat{p}}' \hat{\Gamma}_{\hat{p},1} \hat{\mu}_{p^* - \hat{p}}|
\nonumber \\
& + & |1_{ p^* - \hat{p}}' \hat{\Gamma}_{\hat{p},3} \hat{\mu}_{p^* - \hat{p}}|.\label{pla2-16}
\end{eqnarray}
See that 
\begin{equation}
   \| \hat{\mu}_{\hat{p}} \|_{\infty}= \max_{1 \le j \le \hat{p}} | \hat{\mu}_j | \le \max_{1 \le j \le p} |\hat{\mu}_j |
   \le \max_{1 \le j \le p} | \mu_j| + \max_{1 \le j \le p}
   | \hat{\mu}_j - \mu_j|
   =O_p(1),\label{pla2-17}
\end{equation}
by $\hat{p}\le p^*\le p$ and Assumption \ref{assum3} via triangle inequality in the last inequality.
Also we have by the same analysis in (\ref{pla2-17})
\begin{equation}
    \| \hat{\mu}_{p^* - \hat{p}} \|_{\infty}= \max_{1 \le j \le p^* - \hat{p}} | \hat{\mu}_j | = O_p (1).\label{pla2-18}
\end{equation}

\noindent Next, consider the first term on the the right side of (\ref{pla2-16}).
\begin{eqnarray}
| 1_{p^* - \hat{p}}' \hat{\Gamma}_{\hat{p},2} \hat{\mu}_{\hat{p}}
| &=& | \hat{\mu}_{\hat{p}}' \hat{\Gamma}_{\hat{p},2}' 1_{p^*- \hat{p}}|
\le \| \hat{\mu}_{\hat{p}} \|_{\infty} 
\| \hat{\Gamma}_{\hat{p},2}' 1_{p^*- \hat{p}} \|_{1} \nonumber \\
& \le & \| \hat{\mu}_{\hat{p}} \|_{\infty}
\| \hat{\Gamma}_{\hat{p},2}'\|_{l_1} 
\| 1_{p^*- \hat{p}} \|_{1}  \nonumber \\
& = & \| \hat{\mu}_{\hat{p}} \|_{\infty}
\| \hat{\Gamma}_{\hat{p},2}\|_{l_{\infty}}  (p^* - \hat{p})\nonumber \\
& = & O_p(1) O(1) (p^* - \hat{p}),\label{pla2-19}
\end{eqnarray}
where the first inequality is by Holder's inequality, and the second inequality is by p.345 of \cite{hj2013}, and the second equality is 
by $\| A' \|_{l_1} = \| A \|_{l_{\infty}}$ for a generic matrix A, and the rates are by (\ref{pla2-17}), and Condition A.1.

Now we consider the second term on the right side of (\ref{pla2-16})
\begin{eqnarray}
| 1_{\hat{p}}' \hat{\Gamma}_{\hat{p},1} \hat{\mu}_{p^* - \hat{p}}
| & \le & 
\| 1_{\hat{p}}' \hat{\Gamma}_{\hat{p},1} \hat{\mu}_{p^* - \hat{p}} \|_{\infty} \| \hat{\mu}_{p^* - \hat{p}}\|_1
\le (p^* - \hat{p})\| 1_{p^*} \|_{\infty} \| \hat{\Gamma}_{\hat{p},1} \|_{l_{\infty}} \| \hat{\mu}_{p^* - \hat{p}} \|_{\infty} \nonumber \\
& = & (p^* - \hat{p}) O(1) O_p (1),\label{pla2-20} 
\end{eqnarray}
where we use Holder's inequality for the first inequality, and p.345 of \cite{hj2013} and $l_1-l_{\infty}$ vector norm inequality for the second inequality, and the rates are from Condition A.1 and (\ref{pla2-18}). The third term on the right side of (\ref{pla2-16}) is handled in the same way as in (\ref{pla2-20}). Combine (\ref{pla2-19})(\ref{pla2-20}) into the left side  term in (\ref{pla2-16}) (which is also the first right side term in (\ref{pla2-1})
\begin{equation}
    \frac{| 1_{p^*}' \hat{\Gamma}_{p^*} \hat{\mu}_{p^*}
    - 1_{\hat{p}}' \hat{\Gamma}_{\hat{p}} \hat{\mu}_{\hat{p}}|}{p^*}=o_p (1),
\label{pla2-21}
\end{equation}
by Assumption \ref{assum1}. Now second right term in (\ref{pla2-1}) is handled in the same way as in Step 1 proof- (\ref{pla2-13}). So via (\ref{pla2-13})(\ref{pla2-21}) we have 
\begin{equation}
    \frac{|1_{\hat{p}}' \hat{\Gamma}_{\hat{p}} \hat{\mu}_{\hat{p}}
    - 1_{p^*}' \Gamma_{p^*} \mu_{p^*}|
    }{p^*} =o_p(1).\label{pla2-22}
\end{equation}

{\it Step 3}. Steps 1-2 both show the same result so Lemma \ref{la2} is proved.
{\bf Q.E.D}\\

We need the following assumption, with $Eigmin (A)$ representing the minimum eigenvalue  for a generic square matrix $A$, and $c>0$ is a positive constant.

\begin{assumptionA}\label{assum5}

(i). \[ \frac{|1_{p^*}' \Gamma_{p^*} \mu_{p^*}|}{p^*} \ge c  > 0\]

(ii). \[ Eigmin (\Gamma_{p^*}) \ge c >0\]
    
\end{assumptionA}

Both assumptions are used in the literature, for (i) see \cite{caner2022}, \cite{cf2026}, and for (ii), see \cite{fan2011}, \cite{caner2019}. If we consider first , the numerator in Assumption (i), this  is a sum of returns, $\mu_{p^*}$, scaled by variance, $\Gamma_{p^*}$ since $p^*$ may grow with $n$, this sum may grow with $p^*$, so we scale with $p^*$ as well. We want the scaled returns to be larger than a positive constant. Assumption (ii) can be proved  in factor models as in \cite{fan2011}, \cite{caner2022}, \cite{cf2026}.

{\bf Theorem A.1}\refstepcounter{theorema}\label{thm:sr_convergence}
{\it 

Under Assumptions \ref{assum1}-\ref{assum5} 
\[ \left| \frac{\widehat{SR}_{\hat{p}}^2}{SR_{p^*}^2} -1  
\right| = o_p (1).
\]
}
Remarks. 1. This is a new result in high-dimensional Sharpe Ratio estimation for a GMV portfolio. The number of stocks to be selected is treated as a random variable $\hat{p}$, and selected via a screening method. It is possible to have mild mistakes, $\hat{p} \neq p^*$,  in this portfolio selection. After that selection, we show that this Sharpe ratio via the screened assets in the portfolio can be consistent estimator for a target Sharpe Ratio with a mildly different selection of stocks $p^*$. The literature before can only show high-dimensional consistency of Sharpe Ratio with fixed number of stocks, no screening process is used, or screened stocks are treated as a nonrandom quantity. For this literature with  unconstrained an constrained  high-dimensional portfolios see \cite{caner2022}, \cite{cf2026}.

2. We still allow $p>n$ consistency for screened Sharpe Ratio estimation as long as the precision matrix-mean return are estimated  consistently in the case of  $p>n$ as in \cite{fan2011}, \cite{caner2019}
\cite{caner2022}.

{\bf Proof of Theorem~\ref{thm:sr_convergence}.} Note that by Assumption \ref{assum5} 
\begin{equation}\label{pla3-1}
1_{p^*}' \Gamma_{p^*} 1_p^* \ge \| 1_{p^*} \|_2^2 Eigmin (\Gamma_{p^*}) \ge c p^*.
    \end{equation}

    Then we try to simplify the notation in the proof by defining 
    $\hat{y}:= 1_{\hat{p}}' \hat{\Gamma}_{\hat{p}} \hat{\mu}_{\hat{p}}$, $y= 1_{p^*}' \Gamma_{p^*} \mu_{p^*}$, 
    $\hat{x}= 1_{\hat{p}}' \hat{\Gamma}_{\hat{p}} 1_{\hat{p}} $, $x:= 1_{p^*}' \Gamma_{p^*} 1_{p^*}$.

    Then squared Sharpe Ratio estimation can be written as 
    \begin{eqnarray}\label{pla3-2}
        \frac{ \widehat{SR_{\hat{p}}}^2}{ SR_{p^*}^2} - 1 & = &  \frac{ \hat{p} (\hat{y}^2/\hat{x})}{p^* (y^2/x)} - 1 \nonumber \\
        & = & \frac{\hat{p}}{p^*} \left(
     \frac{\hat{y}^2}{y^2}    
        \right) \left(\frac{x}{\hat{x}}\right) -1 
    \end{eqnarray}

    We can rewrite the right side of (\ref{pla3-2}) 
    \begin{eqnarray}\label{pla3-3}
         \frac{\widehat{SR}_{\hat{p}}^2}{SR_{p^*}^2} - 1 & = & \frac{\hat{p}}{p^*} 
        \left( \frac{\hat{y}^2 - y^2 + y^2}{y^2}
        \right) \left( \frac{x}{\hat{x} - x + x }
        \right) - 1 \nonumber \\
        & = & \frac{\hat{p}}{p^*} \left( \frac{\hat{y}^2 - y^2}{y^2}
 + 1        \right)  \left( \frac{x}{\hat{x} - x + x }
        \right) - 1    \end{eqnarray}

        We consider first 
        \begin{eqnarray}
            \left|\frac{\hat{y}^2 - y^2}{y^2} \right|& \le & \frac{ |\hat{y} - y | | \hat{y} + y| }{y^2} 
             \nonumber \\
             & \le & \left[ \left| \frac{ \hat{y} - y}{y} \right|
                         \right] 
\left[ \left| \frac{ \hat{y} - y}{y} \right| +2 \right] \nonumber \\
& = & 
\left[ \left| \frac{ \hat{y} - y}{y} \right|
                         \right]^2 + 2  \left[ \left| \frac{ \hat{y} - y}{y} \right|
                         \right]   \nonumber \\
                         & = & \left[ \frac{p^* o_p (1)}{c p^*} 
                         \right]^2 + 2 \left[\frac{p^* o_p (1)}{c p^*}\right] = o_p (1),\label{pla3-4}
        \end{eqnarray}
where for the rate we use Lemma \ref{la2} and Assumption \ref{assum5}(i) with $y, \hat{y}$ definitions.
Next, consider the following term above in (\ref{pla3-3}), with $x, \hat{x}$ definitions, 

\begin{equation}
    \frac{x}{\hat{x}-x+x}\le\left| \frac{x}{x - | \hat{x}- x|}\right|=\left| \frac{1}{1 - \frac{| \hat{x}- x|}{x}}\right|\le \frac{1}{1-o_p(1)} \le 1+o_p(1),\label{pla3-5}
\end{equation}
where the result is derived from (\ref{pla3-1}), Lemma \ref{la1}. Use (\ref{pla3-3})(\ref{pla3-4}) on the right side of (\ref{pla3-3})
to have 
\begin{eqnarray}
\left|   \frac{\widehat{SR}_{\hat{p}}^2}{SR_{p^*}^2} -1 
\right| &\le & \frac{\hat{p}}{p^*} (1 + o_p (1))(1 + o_p (1)) - 1 \nonumber \\
& \le & \left| \frac{\hat{p} - p^*+p^*}{p^*}
\right| ( 1+ o_p (1))(1 + o_p (1)) - 1 \nonumber \\
& = & (1+ o_p (1)) (1 + o_p (1))(1 + o_p (1)) - 1 = o_p (1).\label{pla3-6}
\end{eqnarray}
{\bf Q.E.D.}

\vspace{1in}

\section{Additional Details on Precision matrix estimation techniques}\label{sec:quant}

In this section, we describe the precision matrix estimation techniques we use in detail.

\subsection{Nodewise regression}

Nodewise regression, first introduced by \cite{mein2006} and applied to risk-estimation in portfolio settings by \cite{caner2019}, estimates the precision matrix directly via lasso type  penalized linear regressions. \cite{mein2006} starts with precision matrix formulation. In a high-dimensional world, to achieve consistency, by imposing sparsity through zero elements in rows of the precision matrix, they use a penalized estimation to achieve that sparsity in the estimators.
Following the framework in \cite{caner2019}, to obtain each row of the precision matrix estimate, we consider $p$ separate regressions.

Let $y_t$ be the $p \times 1$ vector of all asset excess returns at time $t$. Nodewise regression runs $p$ different Lasso regressions to obtain the precision matrix.  In this method, there is exact sparsity assumption on the rows of the precision matrix. This sparsity is defined by the maximum number of nonzero elements across rows of the precision matrix as $\bar{s}$, and this number should be smaller than $n$. Since there is sparsity on the precision matrix, the estimator should match that. One of the best ways to impose sparsity on estimators is through lasso, as suggested by \cite{mein2006}.

Let $y_{t,j}$ represent excess asset return of  asset $j$. Let $y_{t,-j}$ represent all asset returns (excess) apart from asset $j$, and this is a $p-1 \times 1$ vector.
For each asset $j=1,\cdots,p$, we model its returns $y_{t,j}$ related to  other assets' returns $y_{t,-j}$ through the following (assuming all returns are time-demeaned):
$$y_{t,j} = y_{t,-j}' \gamma_j + \eta_{t,j},$$
where $\gamma_j$ are the regression coefficients and $\eta_{t,j}$ is the error term.  Note that sparsity assumption on the elements of each row of precision matrix is imposed on slope $\gamma_j$ basically. This is clear  from Step 3 of our Algorithm below.
Hence, we will run a lasso regression and penalizing the elements of $\gamma_j$.

The estimation algorithm can be described as follows:
\begin{enumerate}
    \item \textbf{Lasso regression}. To account for high dimensions, the coefficients $\gamma_j$ above are estimated via Lasso regression; that is:
$$\hat{\gamma}_j = \mathrm{argmin}_\gamma \left[ \frac{||y_j - Y_{-j} \gamma||_2^2}{n} +2\lambda_j ||\gamma||_1  \right].$$

Here, $\hat{\gamma}_j$ is a vector of length $p-1$ that are estimates of $\gamma_j$, and $\lambda_j$ is a positive tuning parameter that determines the size of the $l_1$ penalty of the estimates.

\item \textbf{Parameter selection}. The tuning parameter $\lambda_j$ is chosen to minimize the GIC criterion across  a set of possible $\lambda_j$ choices-a grid search- (see \cite{fan2013tuning}), defined as $$GIC(\lambda_j) := \log(\hat{\sigma}_{\lambda_j}^2) + |\hat{S}_{\lambda_j}| \frac{\log p}{n} \log(\log n),$$
where $\hat{\sigma}_{\lambda_j}^2 = \frac{||y_j - Y_{-j} \hat{\gamma}_j||_2^2}{n}$ is the mean squared error of the Lasso regression and $|\hat{S}_{\lambda_j}|$ is the cardinality of the set of nonzero parameters in $\hat{\gamma}_j$ using $\lambda_j$. \cite{fan2013tuning} show that the GIC selects the true model with probability approaching 1 both when $p<n$ and $p \geq n$.  The optimal $\lambda_j$ is denoted as $\lambda_j^*$, and the optimal slope is $\hat{\gamma}_j (\lambda_j^*)$, which means the lasso regression which uses $\lambda_j^*$ among the grid search in Step 1 provides the optimal slope estimate. For ease of notation denote $\hat{\gamma}_j:= \hat{\gamma_j} (\lambda_j^*)$.

\item \textbf{Precision matrix construction}. Repeat steps 1 and 2 for all $j = 1, \dots, p$. As derived in \cite{caner2019}, then the precision matrix can be constructed directly by using matrix algebra and defining  the diagonal elements $\hat{\Gamma}_{j,j} = \hat{\tau}_{j}^{-2}$ and the vector of off-diagonal elements for $j-th$ row  $\hat{\Gamma}_{j,-j}=-\hat{\tau}_{j}^{-2} \hat{\gamma}_j'$, where $\hat{\tau}_j^2 = \frac{|| y_j - Y_{-j} \hat{\gamma}_j ||^2_2}{n} + \lambda_j || \hat{\gamma}_j ||_1$. We form each row $j$ of $\hat{\Gamma}$ by using the main diagonal term $\hat{\Gamma}_{j,j}$ and the off-diagonal term in row $j$ as $\hat{\Gamma}_{j,-j}$. Then stacking each row one upon other we form $\hat{\Gamma}$.
$\hat{\Gamma}$ is then the nodewise estimator of the precision matrix.
 
\end{enumerate}

Under certain assumptions (see \cite{caner2019} and \cite{chang2018confidence}), the nodewise regression estimate $\hat{\Gamma}$ consistently estimates the precision matrix even when $p>n$. 

\subsection{Residual nodewise regression}

Proposed by \cite{caner2022}, residual nodewise regression is an extension on nodewise regression by integrating factor models. Unlike nodewise regression, which produces a sparse precision matrix of returns, residual nodewise regression only assumes a sparse precision matrix of the idiosyncratic errors, allowing for a dense precision matrix of returns.

The asset returns are modeled by $y_{t,j} = b_j' f_t + u_{t,j}$, where $f_t: K \times 1$ are observable factors. In our empirical analysis, we use the standard Fama-French three factor model. But in theory we allow growing number of factors. Here, $b_j: K \times 1$ are the factor loadings for $K$ factors, and $u_{t,j}$ are unobserved errors. The precision matrix of the returns, $\Gamma = \Sigma^{-1}$, can be estimated by the following steps:

\begin{enumerate}
    \item \textbf{Factor Removal}. Firstly, we estimate the residuals $\hat{u}_{t,j}$ via OLS. That is, $$\hat{u}_{j} = y_j - X' \hat{b}_j,$$ where $X = (f_1, \dots, f_n): K \times n$ is the matrix of the factors and the factor loadings are $\hat{b}_j = (XX')^{-1} Xy_j$, the OLS estimator. $y_j: n \times 1$ vector of returns for asset j, $y_j:= (y_{1,j}, \cdots, y_{t,j}, \cdots, y_{n,j})'$.
    In matrix form, let us write the matrix of factor loadings as $\hat{B} = (YX') (XX')^{-1} $, where $Y: p \times n$ matrix, where each row represents the asset returns, and the columns represent time periods.

    \item \textbf{Nodewise on Residuals}. Next, we apply nodewise regression, described in Section \ref{nodewise}, on the residuals $\hat{u}$ to estimate $\hat{\Omega}$ which will be the estimate for $ \Sigma_u^{-1}$, the precision matrix of the errors. 

    \item \textbf{Reconstruction}. We reconstruct the final precision matrix via the Sherman-Morrison-Woodbury formula:
$$\hat{\Gamma} = \hat{\Omega} - \hat{\Omega} \hat{B} [\hat{\Sigma}_f^{-1} + \hat{B}' \hat{\Omega}_{sym} \hat{B}]^{-1} \hat{B}' \hat{\Omega},$$
where 
 $\hat{\Omega}_{sym} = \frac{\hat{\Omega} + \hat{\Omega}' }{2}$ is the symmetric version of $\hat{\Omega}$, $\hat{B}$ is the matrix of estimated factor loadings, and 
    $\hat{\Sigma}_f = \frac{1}{n} XX'- \frac{1}{n^2} X 1_n 1_n' X'$ is the sample covariance of the factors.

\end{enumerate}

Likewise, under certain assumptions, there is a consistency guarantee of the precision matrix of asset returns when $p>n$, established in \cite{caner2022}. However, the main difference is in the number of factors. Compared to nodewise regression, convergence is slower for factor models with many factors, and is the reason we choose a factor model with a small number of factors.

\subsection{Principal Orthogonal Complement Thresholding (POET)}
Proposed by \cite{fan2013}, the principal orthogonal complement thresholding method (POET) is a powerful method to estimate the covariance matrix for linear factor when variables share common factors but are unobservable. The asset returns are modeled by a linear factor model $y_{t,j} = b_j' f_t + u_{t,j}$, however, the factors  at time t, $f_t: K \times 1$ are unknown and need to be estimated, $b_j: K \times 1$ factor loadings for asset $j$.  They assume that covariance matrix of errors is sparse, see (2.2) of \cite{fan2013}.
POET first uses principal components analysis (PCA) to estimate the unobservable factors, then uses a thresholding method to estimate the covariance matrix of errors. The detailed steps can be describe as follows:

\begin{enumerate}
    \item \textbf{Factor Estimation}. 
First, estimate the unobserved $F: n \times K$ matrix of factors by 
principal components. To that effect, denote $Y:p \times n$ matrix of asset returns. The largest $K$ eigenvalues of $n \times n$ matrix $Y'Y$ is found and the K eigenvectors  corresponding to them designated as the columns of $\hat{F}/\sqrt{n}$. The estimate of factor loadings is $\hat{B}= n^{-1} Y \hat{F}: p \times K$.


\item \textbf{Number of Factors}. To estimate the number of unknown factors, use \citet{bai2002determining} formula, let $\hat{F}_k: n \times k$
matrix of estimated factors. Let $\| A\|_F$ be the  Frobenius norm of matrix A.
    \[\hat{K}=\mathrm{argmin}_{1\le k\le K_1}\left[  \log(\frac{1}{pn} \| Y- \frac{1}{n} Y \hat{F}_k \hat{F}_k^{'} \|_F^2)+\frac{k(p+n)}{pn}\log(\min\{p,n\}) \right],\]
    where $K_1$ is a constant upper bound. 
    
    \item \textbf{Residual Estimation}. Set the residuals $\hat{u}_{jt} = y_{t,j} - \hat{b}_j' \hat{f}_t$ where $\hat{f}_t: \hat{K} \times 1$ is the $t$ th row of $\hat{F}$, (in column vector form) and $\hat{b}_j'$ is the $j$ th row of $\hat{B}$. Then estimate the components of the sample covariance matrix of errors $\hat \Sigma_u=\left(\hat{\sigma}_{ij}\right)$ by
    $$\hat{\sigma}_{ij}=\frac{1}{n}\sum_{t=1}^n \hat{u}_{t,i}\hat{u}_{t,j}.$$
    
    \item \textbf{Thresholding}. For each component $\hat{\sigma}_{ij}$, compute a threshold
    \[ \hat{\tau}_{ij}=\frac{1}{2}\left(\frac{1}{\sqrt{p}}+\sqrt{\frac{\log(p)}{n}}\right)\sqrt{\frac{1}{n}\sum_{t=1}^n(\hat{u}_{it}\hat{u}_{jt}-\hat{\sigma}_{ij})^2}.\]
    Denote the thresholded covariance matrix as $\hat{\Sigma}_{u, Th}$. If $\hat{\sigma}_{ij}<\hat{\tau}_{ij}$, set the $(i,j)$-th component of $\hat \Sigma_u$ to be $0$, otherwise, keep $\hat{\sigma}_{ij}$. Its invertibility is proved in \cite{fan2013}.
    
    \item \textbf{Precision Matrix Construction}. We reconstruct the final precision matrix by: 
    \[ \hat \Gamma=\hat{\Sigma}_{u,Th}^{-1}- \hat{\Sigma}_{u,Th}^{-1} \hat B(I_{K}+\hat B' \hat{\Sigma}_{u,Th}^{-1}\hat B)^{-1}\hat B' 
    \hat{\Sigma}_{u,Th}^{-1}.\]
\end{enumerate}

POET is high-dimensional consistent under pervasive factor assumption, with $p>n$ case too.


\subsection{Deep learning}
When the asset returns have nonlinear relationships with factors, estimating the correlation between different assets becomes challenging. Inspired by \cite{farrell2021deep}, \cite{c-d2025} introduced a deep learning-based method to estimate the precision matrix for nonlinear factor models. 

The asset returns are modeled by $y_{t,j}=g_j(f_t)+u_{t,j}$, where $f_t$ is a $K$-dimension observable column vector and $g_j(\cdot)$ is an unknown function. Then the covariance matrix of asset returns  $\Sigma_y$ can be decomposed to $\Sigma_g+\Sigma_u$, where $\Sigma_g$ representing the covariance of unknown function $g_j(.)$, and $\Sigma_u$ represents the covariance matrix of errors. In order to estimate the precision matrix, \cite{c-d2025} use a multi-layer neural network to capture the nonlinear relationships.
\begin{enumerate}
    \item \textbf{Nonlinear Factor estimation}. To estimate $g_j(\cdot)$, fit $f_t$ and $y_{t,j}$ into a deep neural network. Then for the $j$-th asset and feature $f_t$, we are able to get the prediction $\hat g_j$. The estimation of  $\Sigma_g$ is
    $$\hat \Sigma_g=\frac{1}{n}\sum_{t=1}^n(\hat{g}_j(f_t)-\bar g_j(f_t))(\hat{g}_j(f_t)-\bar g_j(f_t))',$$
    where $\bar g_j(f_t)=\frac{1}{n}\sum_{t=1}^n\hat{g}_j(f_t)$.
    \item \textbf{Residual Estimation}. Set the deep learning residuals $\hat u_{t,j}=y_{t,j}-\hat g_j(f_t)$. Then estimate the components of the sample covariance matrix of errors $\hat \Sigma_u=\left(\hat{\sigma}_{ij}\right)$ by
    $$\hat \sigma_{ij}=\frac{1}{n}\sum_{t=1}^n \hat{u}_{t,i}\hat u_{t,j}$$
    \item \textbf{Thresholding}. For each component $\hat{\sigma}_{ij}$, compute a threshold
    $$\hat{\tau}_{ij}=Cr_n\sum_{t=1}^n|\hat{u}_{t,i}\hat{u}_{t,j}-\hat{\sigma}_{ij}|,$$
    where $C$ is a positive constant, $r_n^{1/2}=n^{-\beta/2(\beta+K)}(\log n)^4$ and $\beta$ is a smoothness parameter of $g_j(\cdot)$.
    Note that function estimation error by deep learning estimator is $r_n$.
    Denote the thresholded covariance matrix as $\hat{\Sigma}_{u, Th}$. If $\hat{\sigma}_{ij}<\hat{\tau}_{ij}$, set the $(i,j)$-th component of $\hat \Sigma_{u,Th}$ to be $0$, otherwise, keep $\hat{\sigma}_{ij}$.
    \item \textbf{Precision Matrix Construction}. We reconstruct the final precision matrix by: $$\hat \Gamma =\hat \Sigma_{u,Th}^{-1}-\hat \Sigma_{u,Th}^{-1}\hat \Sigma_g(I_{k}+\hat \Sigma_{u,Th}^{-1}\hat \Sigma_g)^{-1}\hat \Sigma_{u,Th}^{-1}.$$
\end{enumerate}

The thresholding method here is different from POET due to a deep learning-nonlinear model. Under several assumptions on $g_j(\cdot)$, the eigenvalues of $\Sigma_g$ and $\Sigma_u$, deep learning estimator of the precision matrix is consistent but only when the number of assets, $p$ is smaller than the time period, $n$.  

\subsection{Nonlinear Shrinkage (NLS)}
Introduced by \cite{lw2017} and advanced by \cite{ledoit2020analytical}, nonlinear shrinkage address the problem of sample covariance matrix instability in high dimensions by modifying its eigenvalues. Compared to earlier approaches, this analytical method derives a closed-form solution using the Hilbert transform. Let $Y: n \times p$ matrix of excess asset returns, where the rows represent time series and the columns represent various assets.

\begin{enumerate}

\item Start with the sample covariance $S:= Y' Y/n$. It admits a spectral decomposition $S:= U \Lambda U'$, where $U: p \times p$ orthogonal matrix, that has eigenvectors columns of $U$. Obtain $U$.  Let $\Lambda$ be the diagonal matrix of eigenvalues of $S$.

\item  We then provide Hilbert transform estimators based on Epanechnikov kernel that will be used in shrinkage function estimator.
There will be three cases of interest. First, when $p <n$, all $\lambda_j>0$, and when $p>n$ there are two possibilities, either 
$\lambda_j>0$, or $\lambda_j=0$, since $p>n$, sample covariance matrix is singular.

(i). For $\lambda_j>0$, $p< n$ case, with definition of adaptive bandwidth adjusted for sample eigenvalues, 
$h_{n,k}:= n^{-1/3} \lambda_k$, $k=1,\cdots,p$
\begin{eqnarray*}
\hat{H}_f (\lambda_j)& := & \frac{1}{p} \sum_{k=1}^p \{ \frac{-3 (\lambda_j - \lambda_k)}{10 \pi h_{n,k}^2} \\
& + & \frac{3}{(4 \sqrt{5}) \pi h_{n,k}} \left[ 1 - \frac{1}{5} \left( \frac{(\lambda_j- \lambda_k)}{h_{n,k}}
\right)^2 \right] log \left| \frac{\sqrt{5} h_{n,k} - \lambda_j + \lambda_k}{\sqrt{5} h_{n,k} + \lambda_j - \lambda_k}
\right|\}
\end{eqnarray*}
(ii). In the case of $\lambda_j>0, p>n$ case
\begin{eqnarray*}
\hat{H}_{\bar{f}} (\lambda_j)& := & \frac{1}{n} \sum_{k=p-n+1}^p \{ \frac{-3 (\lambda_j - \lambda_k)}{10 \pi h_{n,k}^2} \\
& + & \frac{3}{(4 \sqrt{5}) \pi h_{n,k}} \left[ 1 - \frac{1}{5} \left( \frac{(\lambda_j- \lambda_k)}{h_{n,k}}
\right)^2 \right] log \left| \frac{\sqrt{5} h_{n,k} - \lambda_j + \lambda_k}{\sqrt{5} h_{n,k} + \lambda_j - \lambda_k}
\right|\}
\end{eqnarray*}
(iii). In case of $\lambda_j=0, p>n$ case
\begin{eqnarray*}
\hat{H}_{\bar{f}} (0) & = & \left[ \frac{3}{10 h_n^2} + \frac{3}{4 \sqrt{5} h_n} \left( 1 - \frac{1}{5 h_n^2}\right) log \left( \frac{1 + \sqrt{5} h_n}{1 - \sqrt{5} h_n}\right)  
\right] \nonumber \\
& \times & \frac{1}{\pi n } \sum_{j=p-n+1}^p \frac{1}{\lambda_j}.
\end{eqnarray*}
Section 4.7 of \cite{ledoit2020analytical} uses $h_n:= n^{-1/3}$ as the bandwidth.

\item In the case of $\lambda_j>0$,  for the optimal shrinkage function estimation in the next step we need an estimate of spectral density function.
First for the case with $p<n$, the spectral density, we use Epanechnikov kernel
\[ \hat{f} (\lambda_j):= \frac{1}{p} \sum_{k=1}^p \frac{3}{4 \sqrt{5} h_{n,k}} 
\left[ 1 - \frac{1}{5} (\frac{\lambda_j - \lambda_k}{h_{n,k}})^2
\right]^+,\] 
with $[.]^+$ denoting the positive part. In the case of $p>n$, without losing any generality, if the first $n$ eigenvalues are zero as in equation (C.6) of Supplement of \cite{ledoit2020analytical}, we have the following estimate for the spectral density
\[ \hat{f} (\lambda_j):= \frac{1}{p} \sum_{k=p-n+1}^p \frac{3}{4 \sqrt{5} h_{n,k}} 
\left[ 1 - \frac{1}{5} (\frac{\lambda_j - \lambda_k}{h_{n,k}})^2
\right]^+,\]

\item  Then set the diagonal matrix $\hat{\Delta}^*$
\[ \hat{\Delta}^*:= diag (\hat{\phi}^* (\lambda_1), \cdots, \hat{\phi}^* (\lambda_p)),\]
where if $\lambda_j=0$ in case of  $ p>n$
\begin{equation}
\hat{\phi}^* (0) = \frac{1}{\pi (\hat{c}-1) \hat{H}_{\bar{f}} (0)},\label{osa1}
\end{equation}
otherwise (With $\lambda_j>0$ in case of  $p<n$, or $\lambda_j>0$ in case of  $p>n$)
\begin{equation}
\hat{\phi}^* (\lambda_j):= \frac{\lambda_j}{[ \pi  \hat{c}  \hat{f}(\lambda_j)]^2 + [ 1 - \hat{c} - \pi c \lambda_j \hat{H}_f (\lambda_j)]^2}.\label{osa2}
\end{equation}


\item  Form the optimal shrinkage estimator for the covariance matrix as
\[ \hat{\Sigma}:= U \hat{\Delta}^* U'.\]


\item  Invert $\hat{\Sigma}$ to get the precision matrix estimate $\hat{\Gamma}:= (\hat{\Sigma})^{-1}$.

\end{enumerate}

\section{Additional Material}\label{additional material}

The following subsections considers a few additional robustness checks:  a different strategy for communication between agents, short caps analysis, equal-weighted no quant strategies, LLM-S prompts and outputs, and a more traditional way of screening stocks and empirical results related to that.

To address for potential randomness in each LLM run, we have also reran the Agentic AI screen twice and evaluated the downstream portfolio Sharpe ratios in the medium term. The results were similar to our main tables.

\subsection{Consensus rules}\label{consensus rules}

As a robustness check, we modify our consensus rule for the periods between November 2023-April 2024 and October 2021-April 2024. Instead of pivoting to the union of LLM-S and FinBERT when their intersection is null, we go with FinBERT's recommendation, since it achieves a higher Sharpe ratio by itself then LLM-S in  most cases in Tables \ref{tab2} and \ref{tab5}. Tables \ref{tab:default finbert 4o} and \ref{tab:default finbert 3.5} present the results. The Sharpe ratios in Table \ref{tab:default finbert 4o} are higher than in our main results, showing the importance of FinBERT and sentiment analysis in the short term. However, when we consider the longer horizon, the winner in Table \ref{tab:default finbert 3.5} does not achieve as high of a SR as the main table in Table \ref{tab8} (with union), as there are 14 method-portfolio combinations that do better. 

Overall, the magnitude of the Sharpe ratios in both horizons are relatively close to the default union, showing robustness.

\begin{table}[htbp]
    \centering
    {\bf CHATGPT-4o: CONSENSUS RULES}\\
    \begin{tabular}{l|ccc|ccc|ccc}
\hline
 & \multicolumn{3}{c|}{\textbf{Sharpe Ratio}} & \multicolumn{3}{c|}{\textbf{Returns}} & \multicolumn{3}{c}{\textbf{Variance}} \\
\textbf{Method} & \textbf{GMV} & \textbf{MV} & \textbf{MSR} & \textbf{GMV} & \textbf{MV} & \textbf{MSR} & \textbf{GMV} & \textbf{MV} & \textbf{MSR} \\ \hline
NW  & 1.4787 & 1.0388 & 1.9917 & 0.2090 & 0.1640 & 0.2536 & 0.0200 & 0.0249 & 0.0162 \\
Residual NW  & 2.0444 & 1.5707 & 5.6356 & 0.2470 & 0.2173 & 0.4584 & 0.0146 & 0.0191 & 0.0066 \\
Deep learning & 2.3050 & 1.6008 & 4.0585 & 0.2738 & 0.2166 & 0.3790 & 0.0141 & 0.0183 & 0.0087 \\
POET & 1.9646 & 1.3365 & 2.7122 & 0.2727 & 0.2089 & 0.3339 & 0.0193 & 0.0244 & 0.0152 \\
NLS & 3.5733 & 2.2318 & \textbf{8.8279} & 0.3223 & 0.2605 & \textbf{0.5481} & 0.0081 & 0.0136 & \textbf{0.0039} \\ \hline
\end{tabular}%
    \caption{Annualized Sharpe ratios, returns, and variance with different methods of estimating the precision matrix, with different objective functions, applied to firms that FinBERT+LLM has screened. Here, the consensus rule is to use the intersection, and defaulting to FinBERT's recommendation in case of a null intersection. Portfolio abbreviations as in Table~\ref{tab1}.}
    \label{tab:default finbert 4o}
\end{table}

\begin{table}[htbp]
    \centering
    {\bf CHATGPT-3.5: CONSENSUS RULES}\\
    \begin{tabular}{l|ccc|ccc|ccc}
\hline
 & \multicolumn{3}{c|}{\textbf{Sharpe Ratio}} & \multicolumn{3}{c|}{\textbf{Returns}} & \multicolumn{3}{c}{\textbf{Variance}} \\
\textbf{Method} & \textbf{GMV} & \textbf{MV} & \textbf{MSR} & \textbf{GMV} & \textbf{MV} & \textbf{MSR} & \textbf{GMV} & \textbf{MV} & \textbf{MSR} \\ \hline
NW & 0.8456 & -0.3718 & 0.5557 & 0.2933 & -0.0966 & 0.1175 & 0.1203 & 0.0675 & 0.0447 \\
Residual NW & 0.8659 & -0.3873 & 0.6218 & 0.2965 & -0.0966 & 0.1397 & 0.1173 & 0.0622 & 0.0505 \\
Deep learning & 0.8988 & -0.3509 & 0.8751 & \textbf{0.3039} & -0.0888 & 0.1736 & 0.1143 & 0.0641 & \textbf{0.0393} \\
POET & 0.8839 & -0.3117 & \textbf{0.9960} & 0.2900 & -0.0810 & 0.2378 & 0.1077 & 0.0675 & 0.0570 \\
NLS & 0.8571 & -0.4136 & 0.4752 & 0.2971 & -0.1033 & 0.1184 & 0.1202 & 0.0623 & 0.0621 \\ \hline
\end{tabular}%
    \caption{Annualized Sharpe ratios, returns, and variance with different methods of estimating the precision matrix, with different objective functions, applied to firms that FinBERT+LLM has screened. Here, the consensus rule is to use the intersection, and defaulting to FinBERT's recommendation in case of a null intersection. Portfolio abbreviations as in Table~\ref{tab1}.}
    \label{tab:default finbert 3.5}
\end{table}

\subsection{Small Caps}
We have also considered the question of how small caps factor into our strategy and Sharpe ratio between November 2023-April 2024 out-of-sample period. We have also analyzed the bottom 100 stocks by market capitalization in our universe, and applied our buy-sell decisions corresponding to them only in our main table, which contains LLM-S, FinBERT and the quantitative estimation method with the Agentic AI structure. The best method-objective Sharpe ratios decrease.
With ChatGPT-4o, the best SR decreases from 8.1612 to 2.1316, and with ChatGPT-3.5, the best SR decreases from 1.0946 to 0.4722.

Our explanation for this behavior is that since we are now restricting to firms with the 100 smallest market capitalizations in the S\&P 500, this contradicts with most buy signals from LLM-S (since most buy signals invest in firms with high market caps, as illustrated in the example in Section~\ref{subsec:sensible_screening}), which may explain why the Sharpe ratio  decreases.

\begin{table}[htbp]
    \centering
    {\bf CHATGPT-4o: SMALL CAPS}\\
    \begin{tabular}{l|ccc|ccc|ccc}
\hline
 & \multicolumn{3}{c|}{\textbf{Sharpe Ratio}} & \multicolumn{3}{c|}{\textbf{Returns}} & \multicolumn{3}{c}{\textbf{Variance}} \\
\textbf{Method} & \textbf{GMV} & \textbf{MV} & \textbf{MSR} & \textbf{GMV} & \textbf{MV} & \textbf{MSR} & \textbf{GMV} & \textbf{MV} & \textbf{MSR} \\ \hline
NW & 2.0762 & 2.1057 & 2.0518 & 0.5356 & \textbf{0.5420} & 0.5167 & 0.0665 & 0.0662 & 0.0634 \\
Residual NW  & 2.1049 & 2.1213 & 2.0529 & 0.5223 & 0.5267 & 0.5103 & 0.0616 & 0.0616 & 0.0618 \\
Deep learning & 2.1000 & 2.1258 & 2.0795 & 0.5132 & 0.5222 & 0.4812 & 0.0597 & 0.0603 & \textbf{0.0535} \\
POET & 2.0732 & 2.1069 & 2.0674 & 0.5235 & 0.5408 & 0.4998 & 0.0638 & 0.0659 & 0.0585 \\
NLS & 2.1125 & \textbf{2.1316} & 2.0542 & 0.5110 & 0.5162 & 0.5069 & 0.0585 & 0.0587 & 0.0609 \\ \hline
\end{tabular}%
    \caption{Annualized Sharpe ratios, returns, and variance with different methods of estimating the precision matrix, with different objective functions, applied to small-cap firms that FinBERT+LLM has screened. Portfolio abbreviations as in Table~\ref{tab1}.}
    \label{tab:small cap finbert 4o}
\end{table}

\begin{table}[htbp]
    \centering
    {\bf CHATGPT-3.5: SMALL CAPS}\\
    \begin{tabular}{l|ccc|ccc|ccc}
\hline
 & \multicolumn{3}{c|}{\textbf{Sharpe Ratio}} & \multicolumn{3}{c|}{\textbf{Returns}} & \multicolumn{3}{c}{\textbf{Variance}} \\
\textbf{Method} & \textbf{GMV} & \textbf{MV} & \textbf{MSR} & \textbf{GMV} & \textbf{MV} & \textbf{MSR} & \textbf{GMV} & \textbf{MV} & \textbf{MSR} \\ \hline
NW & 0.4297 & 0.3943 & 0.0559 & 0.0798 & 0.1008 & 0.0153 & 0.0345 & 0.0653 & 0.0751 \\
Residual NW & 0.4015 & 0.4127 & 0.1322 & 0.0725 & \textbf{0.1067} & 0.0498 & \textbf{0.0327} & 0.0669 & 0.1421 \\
Deep learning & 0.4466 & 0.3997 & 0.1594 & 0.0842 & 0.1029 & 0.0491 & 0.0355 & 0.0662 & 0.0947 \\
POET & \textbf{0.4722} & 0.3915 & 0.3613 & 0.0928 & 0.1005 & 0.0798 & 0.0386 & 0.0659 & 0.0487 \\
NLS & 0.4047 & 0.4047 & 0.1197 & 0.0735 & 0.1041 & 0.0449 & 0.0330 & 0.0662 & 0.1408 \\ \hline
\end{tabular}%
    \caption{Annualized Sharpe ratios, returns, and variance with different methods of estimating the precision matrix, with different objective functions, applied to small-cap firms that FinBERT+LLM has screened. Portfolio abbreviations as in Table~\ref{tab1}.}
    \label{tab:small cap finbert 3.5}
\end{table}

\subsection{Equal-weighted Portfolios}

In this section, we compare our best method-objective models against the $1/N$ equal-weighted portfolios from November 2023-April 2024 and October 2021-April, respectively. Beating the $1/N$ baseline is non-trivial. An extensive body of literature \citep{demiguel2009optimal, duchin2009markowitz} demonstrate that complex models are often inferior to naive equal-weighting out-of-sample. We will demonstrate that our Agentic AI screening method achieves a higher Sharpe ratio than the $1/N$ portfolio within a $10\%$ statistically significant rate.

The equal-weighted S\&P 500 portfolio achieved a Sharpe ratio of 1.5147 for the former horizon and 0.1633 for the latter horizon, respectively. These are far below the best method-objective combinations that utilize Agentic AI screening. Moreover, 11 out of 15 method-objective combinations beat the $1/N$ baseline in the short-term, and 10 out of 15 method-objective combinations beat it in the medium-term. 

We can test statistical significance of our best Agentic-AI model against the $1/N$ portfolio due to a method by \citet{ledoit2008robust}, who evaluates the statistical significance of the Sharpe ratio of two paired time series. We compare the Agentic AI Residual NW/GMV portfolio against the equal-weighted S\&P 500 portfolio in October 2021-April 2024\footnote{The time horizon from November 2023-April 2024 was too short to draw statistically significance results.}. To isolate for the impact of the quantitative agent, we also test our Agentic AI portfolio against the equal-weighted Agentic AI-screened portfolio. The event that the Agentic AI portfolio achieves a higher Sharpe ratio than the equal-weighted S\&P 500 portfolio has a p-value of 0.067, and the event that the Agentic AI portfolio is better than the equal-weighted Agentic AI-screened portfolio has a p-value of 0.062, allowing us to conclude 7\% significance of our Shape ratios.



\vspace{1in}

\subsection{LLM-S Prompts and Outputs}

Below is a code snippet describing the LLM-S agent. This snippet is specifically for test dates in 2024.

\begin{lstlisting}[breaklines=true]
strategy_agent = Agent(
    role="Quantitative Strategy Developer",
    goal="Develop systematic BUY/HOLD/SELL rules based on firm characteristics that can be applied to all S&P 500 firms",
    backstory="""
    You are an expert quantitative strategist who creates systematic, rule-based trading strategies.
    
    CRITICAL DATA UNDERSTANDING:
    - 'mve' = log(market value of equity), represents log firm size
    - 'bm' = book-to-market ratio (value factor). Understand that a high book-to-market value means undervalued, and a low book-to-market value means overvalued.
    - 'mom12m' = 12-month momentum
    - ALL features are standardized: mean = 0, standard deviation = 1
    - Values are z-scores showing standard deviations from mean
    
    Your task is to develop EXPLICIT, SYSTEMATIC RULES for generating trading signals. Understand that when doing the following, you must use causal masking from December 2023 to prevent any look-ahead bias.
    
    1. EXPLORE THE DATA (December 2023):
       - Identify what constitutes "extreme" values for mve, bm, and mom12m
       - Look for natural clustering or breakpoints in the data
       - Consider correlations between characteristics
    
    2. DEVELOP CLEAR RULES based on economic intuition:
       - Keep in mind the market conditions at this date - this might influence the rules you choose.
       - The following are example questions you can consider, BUT THEY ARE NOT EXHAUSTIVE: 
           - Value stocks: Should low bm (cheap) be BUY or SELL?
           - Momentum: Should high mom12m (strong performance) be BUY or SELL?
           - Size: Should mve matter for the strategy?
           - Combinations: What about value + momentum + size together?
    
    3. DEFINE SPECIFIC THRESHOLDS:
     Your output must include exact rules. You can use test_complex_condition to test different combinations.
       Your output must include exact rules like:
       - "BUY if: bm < -0.71 AND mom12m > 0.85 AND mve > 0.23" or "BUY if: (bm > 0.57 AND mom12m < 0.82) OR mve > -0.88" or "BUY if: bm > 0.63 OR mve < 0.98"
       - "SELL if: bm > 1.25 OR mom12m < -0.98" or "SELL if: (bm < -0.84 OR mve < 0) AND mom12m < -0.97" or "SELL if: (bm > 0.94 AND mom12m < -1.06) OR mve <-0.68"
       - "HOLD: all other cases"
       - However, the above is only AN EXAMPLE - so do not simply copy the format above. You are free to include/exclude as many conditions in the if statements. You are also free to make the conditions as complicated or as simple as you like.
       - Be PRECISE in your thresholds, do not choose numbers that are nice or round - you are a quantative strategy developer.
    
    4. PROVIDE RATIONALE:
       - Why these thresholds?
       - What's the economic intuition?
       - What patterns did you observe in the data?
    
    CRITICAL REQUIREMENTS:
    - Rules must be DETERMINISTIC (same inputs -> same output)
    - Use ONLY z-score comparisons (>, <, AND, OR), but these may be impacted by market conditions.
    - Define thresholds for BUY, SELL, and HOLD
    - Rules should be implementable as: if (condition) then signal
    - Keep in mind that we will use the buy signals to construct a portfolio, so it is better to give too many signals, rather than too few signals.
    
    OUTPUT FORMAT:
    ===========================================
    SYSTEMATIC TRADING RULES
    ===========================================
    
    Data Exploration Summary:
    - [Key statistics and patterns observed]
    
    BUY RULE:
    if [ANY complex z-score condition using AND/OR/NOT]:
        signal = BUY
    
    Examples of valid BUY rules:
    - "bm < -1.15 AND mom12m > 0.73 AND mve < 1.11" (simple AND)
    - "bm < -1.56 OR mom12m > 1.28 OR mve > 1.52" (simple OR)
    - "(bm < -1.02 AND mom12m > 0.53) OR mve > 1.59" (combination)
    - "bm < -0.83 AND (mom12m > 0.77 OR mve > 1.08)" (nested conditions)
    
    SELL RULE:
    if [ANY complex z-score condition]:
        signal = SELL
    
    HOLD RULE:
    else:
        signal = HOLD
    
    Rationale:
    - [Economic reasoning for BUY rule]
    - [Economic reasoning for SELL rule]
    - [Expected signal distribution]
    
    ```
    ===========================================
    
    Be precise, systematic, and data-driven. Your rules will be applied to ~500 firms.
    """
)
\end{lstlisting}

Further, below is code snippet on the task description of our LLM-S model.

\begin{lstlisting}[breaklines=true]
strategy_task = Task(
    description="""
    Develop systematic BUY/HOLD/SELL rules for S&P 500 firms at December 2023.
    
    Available characteristics (all are z-scores):
    - mve: log firm size
    - bm: log book-to-market (value)
    - mom12m: 12-month momentum
    
    Your process:
    1. Get database schema and understand available data for December 2023
    2. Explore extreme values for each characteristic
    3. Look for patterns and correlations
    4. Use test_complex_condition to test different rule combinations
    5. Develop systematic rules with specific z-score thresholds

    You have COMPLETE FLEXIBILITY in creating rules. Test different combinations using AND, OR, NOT.
    
    CRITICAL: Your output must be EXPLICIT RULES with exact thresholds that can be 
    implemented in Python/pandas. You must give PRECISE THRESHOLDS - do not give thresholds that are only nice or round numbers. Further, use causal masking from December 2023 to prevent any look-ahead bias.
    
    Focus on:
    - Economic intuition (value, momentum, size effects)
    - Clear, implementable thresholds
    - Balance between signal strength and diversification
    - Rules that make sense for ~500 firms
    
    Output systematic rules that I can directly implement in code.
    """,
    expected_output="""
    Complete strategy document with:
    1. Data exploration summary
    2. Explicit BUY rule with z-score thresholds
    3. Explicit SELL rule with z-score thresholds  
    4. HOLD rule (default case)
    5. Economic rationale for each rule
    
    Rules must be deterministic and implementable.
    """,
    agent=strategy_agent,
    tools=[
        get_database_schema,
        query_firm_database,
        get_extreme_firms,
        test_complex_condition
    ]
)
\end{lstlisting}

Lastly, below is an example of the corresponding output of the LLM-S model. This particular year, the LLM likes buying undervalued firms with above-average size and non-small momentum values. Conversely, it likes selling firms that are overvalued, small firms, or negative-momentum firms (but if they have not shown significantly positive momentum). 

\begin{lstlisting}[language= , breaklines=true]
===========================================
SYSTEMATIC TRADING RULES
===========================================

Data Exploration Summary:
- The data includes standardized book-to-market (bm), momentum (mom12m), and market value of equity (mve) for S&P 500 firms as of December 2023.
- All features are standardized to have a mean of 0 and a standard deviation of 1.
- Extreme values were explored to understand the distribution of characteristics.

BUY RULE:
if bm > 0.95 AND mve > 0.3 AND mom12m > -0.5:
    signal = BUY

SELL RULE:
if (bm < -0.75 OR mom12m < -0.55 OR mve < -0.75) AND NOT (mom12m > 1.5):
    signal = SELL

HOLD RULE:
else:
    signal = HOLD

Rationale:
- BUY Rule: This rule targets undervalued (high bm), reasonably sized (mve > 0.3) companies with positive momentum (mom12m > -0.5). The economic intuition is to buy companies that are currently cheap but have shown some signs of recovery or positive market sentiment.
- SELL Rule: This rule aims to sell companies that are overvalued (low bm), have negative momentum (mom12m < -0.55), or are small in size (mve < -0.75). The `NOT (mom12m > 1.5)` condition prevents selling companies with extremely high momentum, even if they meet the other criteria, as these might be temporary situations or represent significant growth opportunities. The economic intuition is to avoid holding onto companies that are losing value or are too small to provide substantial returns, unless they are exhibiting exceptional positive momentum.
- Expected signal distribution: The BUY rule is expected to generate signals for a small percentage of firms (around 2%), focusing on higher-conviction value opportunities. The SELL rule is expected to affect a larger percentage of firms (around 50%), filtering out less desirable investments and managing risk. The remaining firms will be held.
\end{lstlisting}

\subsection{Novy-Marx Screening}

In this part, we analyze  \cite{nm2013} based screening effect on Sharpe Ratio analysis.
Table \ref{tab:novy marx + finbert 4o} contains Sharpe ratios, returns, and variances with the screening method described in \cite{nm2013}, with out-of-sample test periods from November 2023 to April 2024. Table \ref{tab:novy marx + finbert 3.5} uses out-of-sample test periods from October 2021 to April 2024. For Novy-Marx screening, out of 500 firms, we take the 150 stocks with the highest combined profitability and value ranks, and the 150 stocks with the lowest, to be the screened set of stocks each year.

To see the effect that LLM-S has in the screening ensemble containing LLM-S and FinBERT, we replace LLM-S instead with the method described in \cite{nm2013}. Notice that in the sample in Table \ref{tab:novy marx + finbert 4o}, the Sharpe ratio is lower than the one obtained by the LLM-S plus FinBERT ensemble when compared with  Table \ref{tab8}.
The same result holds true for the ChatGPT-3.5 sample in Table \ref{tab:novy marx + finbert 3.5} in comparison with  Table \ref{tab35-8}
as well. This suggests another contribution of our LLM-S screening model: that it has better synergy with FinBERT than with other screening methods. 

In short, the ensemble containing both LLM-S and FinBERT beat Novy-Marx in both windows, showing that FinBERT can also decrease mistakes made by LLM-S when acting as an ensemble.

\begin{table}[h!]
    \centering
    \begin{tabular}{l|ccc|ccc|ccc}
\hline
 & \multicolumn{3}{c|}{\textbf{Sharpe Ratio}} & \multicolumn{3}{c|}{\textbf{Returns}} & \multicolumn{3}{c}{\textbf{Variance}} \\
\textbf{Method} & \textbf{GMV} & \textbf{MV} & \textbf{MSR} & \textbf{GMV} & \textbf{MV} & \textbf{MSR} & \textbf{GMV} & \textbf{MV} & \textbf{MSR} \\ \hline
NW & 1.1247 & 0.6692 & 1.6025 & 0.1894 & 0.1131 & 0.2504 & 0.0283 & 0.0286 & 0.0244 \\
Residual NW & 1.2579 & 0.9681 & \textbf{4.5409} & 0.1956 & 0.1571 & \textbf{0.4691} & 0.0242 & 0.0263 & \textbf{0.0107} \\
Deep learning & 1.4625 & 1.0795 & 2.4936 & 0.2255 & 0.1684 & 0.3365 & 0.0238 & 0.0243 & 0.0182 \\
POET & 1.2240 & 0.7361 & 1.5136 & 0.1963 & 0.1164 & 0.2401 & 0.0257 & 0.0250 & 0.0252 \\
NLS & 1.5171 & 1.0169 & 4.2105 & 0.2121 & 0.1523 & 0.4430 & 0.0195 & 0.0224 & 0.0111 \\ \hline
\end{tabular}%
    \caption{Annualized Sharpe ratios with different methods of estimating the precision matrix, with different objective functions, applied to firms that Novy-Marx + FinBERT has screened. Test period is from Nov. 2023-April 2024. Portfolio abbreviations as in Table~\ref{tab1}.}
    \label{tab:novy marx + finbert 4o}
\end{table}

\begin{table}[h!]
    \centering
    \begin{tabular}{l|ccc|ccc|ccc}
\hline
 & \multicolumn{3}{c|}{\textbf{Sharpe Ratio}} & \multicolumn{3}{c|}{\textbf{Returns}} & \multicolumn{3}{c}{\textbf{Variance}} \\
\textbf{Method} & \textbf{GMV} & \textbf{MV} & \textbf{MSR} & \textbf{GMV} & \textbf{MV} & \textbf{MSR} & \textbf{GMV} & \textbf{MV} & \textbf{MSR} \\ \hline
NW & 0.0179 & 0.4580 & -0.1263 & 0.0041 & 0.1248 & -0.0261 & 0.0516 & 0.0742 & \textbf{0.0425} \\
Residual NW & -0.0178 & \textbf{0.5791} & -0.0255 & -0.0038 & \textbf{0.1689} & -0.0060 & 0.0454 & 0.0851 & 0.0551 \\
Deep learning & 0.0672 & 0.5276 & -0.0323 & 0.0149 & 0.1415 & -0.0068 & 0.0494 & 0.0719 & 0.0442 \\
POET & 0.0529 & 0.4201 & -0.0090 & 0.0125 & 0.1112 & -0.0019 & 0.0554 & 0.0700 & 0.0452 \\
NLS & -0.0089 & 0.5506 & -0.0461 & -0.0019 & 0.1532 & -0.0112 & 0.0445 & 0.0774 & 0.0588 \\ \hline
\end{tabular}%
    \caption{Annualized Sharpe ratios with different methods of estimating the precision matrix, with different objective functions, applied to firms that Novy-Marx + FinBERT has screened. Test period is from Oct. 2021-April 2024. Portfolio abbreviations as in Table~\ref{tab1}.}
    \label{tab:novy marx + finbert 3.5}
\end{table}

\clearpage


\clearpage

\subsection{ChatGPT-4-o}

As a robustness test, we address fully the potential problem of look-ahead bias using another model: ChatGPT-4-o, with a knowledge cutoff date of October 2023. When applied to test dates of October November 2023 to April 2024, we obtain the following tables. The S\&P 500 achieved an annualized Sharpe ratio of 2.0857 during this time period.

\begin{table}[h!]
    \centering
    {\bf BASELINE-ONLY WITH QUANTITATIVE WEIGHTING: Nov 2023-April 2024}\\
    \begin{tabular}{l|ccc|ccc|ccc}
\hline
 & \multicolumn{3}{c|}{\textbf{Sharpe Ratio}} & \multicolumn{3}{c|}{\textbf{Returns}} & \multicolumn{3}{c}{\textbf{Variance}} \\
\textbf{Method} & \textbf{GMV} & \textbf{MV} & \textbf{MSR} & \textbf{GMV} & \textbf{MV} & \textbf{MSR} & \textbf{GMV} & \textbf{MV} & \textbf{MSR} \\ \hline
NW & 1.5912 & 1.6171 & 1.6021 & 0.1557 & 0.1537 & 0.1633 & 0.0096 & 0.0090 & 0.0104 \\
Residual NW & 1.1037 & 1.3992 & 0.7649 & 0.0586 & 0.0696 & 0.1293 & 0.0028 & \textbf{0.0025} & 0.0286 \\
Deep learning & -0.4666 & -0.4648 & -0.2548 & -0.0704 & -0.0715 & -0.0689 & 0.0228 & 0.0237 & 0.0731 \\
POET & 1.4204 & 1.4359 & 1.3758 & 0.0885 & 0.0903 & 0.0891 & 0.0039 & 0.0040 & 0.0042 \\
NLS & 2.9373 & \textbf{3.0414} & 1.8596 & 0.2025 & 0.1929 & \textbf{0.3475} & 0.0048 & 0.0040 & 0.0349 \\ \hline
\end{tabular}%
    \caption{Annualized Sharpe ratios, returns, and variance with different methods of estimating the precision matrix, with different objective functions, applied to all firms in the S\&P 500. GMV=Global minimum variance portfolio, MV=mean-variance portfolio with target returns as 1\% monthly, MSR=maximum Sharpe ratio portfolio.}
    \label{tab1}
\end{table}

\begin{table}[htbp]
    \centering
    {\bf LLM-S WITH QUANTITATIVE WEIGHTING: Nov 2023-April 2024}
    \begin{tabular}{l|ccc|ccc|ccc}
\hline
 & \multicolumn{3}{c|}{\textbf{Sharpe Ratio}} & \multicolumn{3}{c|}{\textbf{Returns}} & \multicolumn{3}{c}{\textbf{Variance}} \\
\textbf{Method} & \textbf{GMV} & \textbf{MV} & \textbf{MSR} & \textbf{GMV} & \textbf{MV} & \textbf{MSR} & \textbf{GMV} & \textbf{MV} & \textbf{MSR} \\ \hline
NW & 0.5962 & 0.5941 & 0.5708 & 0.1354 & 0.1240 & 0.1388 & 0.0516 & 0.0546 & 0.0472 \\
Residual NW & 0.5155 & 0.5367 & 0.6531 & 0.1186 & 0.1226 & 0.1377 & 0.0529 & 0.0522 & 0.0445 \\
Deep learning & 0.5857 & 0.5618 & 0.6122 & 0.1338 & 0.1302 & 0.1320 & 0.0522 & 0.0537 & 0.0465 \\
POET & 0.3861 & 0.4986 & 0.3039 & 0.0829 & 0.1137 & 0.0625 & 0.0461 & 0.0520 & \textbf{0.0423} \\
NLS & 0.6394 & 0.6083 & 0.9585 & 0.1478 & 0.1407 & 0.2164 & 0.0534 & 0.0535 & 0.0510 \\ 
Sample Cov & 0.5122 & 0.4601 & {\bf 1.1629} & 0.1200 & 0.1075 & {\bf 0.2782} & 0.0549 & 0.0546 & 0.0572 \\
\hline
\end{tabular}%
    \caption{Annualized Sharpe ratios, returns, and variance with different methods of estimating the precision matrix, with different objective functions, applied to firms that the LLM has screened. Portfolio abbreviations as in Table~\ref{tab1}.}
    \label{tab2}
\end{table}

\begin{table}[htbp]
    \centering
    {\bf LOGISTIC REGRESSION WITH QUANTITATIVE WEIGHTING: Nov 2023-April 2024}\\
    \begin{tabular}{l|ccc|ccc|ccc}
\hline
 & \multicolumn{3}{c|}{\textbf{Sharpe Ratio}} & \multicolumn{3}{c|}{\textbf{Returns}} & \multicolumn{3}{c}{\textbf{Variance}} \\
\textbf{Method} & \textbf{GMV} & \textbf{MV} & \textbf{MSR} & \textbf{GMV} & \textbf{MV} & \textbf{MSR} & \textbf{GMV} & \textbf{MV} & \textbf{MSR} \\ \hline
NW & 1.2405 & 0.8946 & 1.5106 & 0.1154 & 0.0858 & 0.1452 & 0.0087 & 0.0092 & 0.0092 \\
Residual NW & 0.9554 & 1.0912 & 1.8282 & 0.0830 & 0.0937 & 0.2777 & 0.0075 & 0.0074 & 0.0231 \\
Deep learning & 1.7043 & 1.0285 & 2.7159 & 0.1421 & 0.0949 & 0.2393 & 0.0070 & 0.0085 & 0.0078 \\
POET & 1.7210 & 1.2661 & 2.2568 & 0.1237 & 0.0975 & 0.1727 & \textbf{0.0052} & 0.0059 & 0.0059 \\
NLS & 2.1484 & 1.3866 & \textbf{3.3540} & 0.2022 & 0.1425 & 0.5055 & 0.0089 & 0.0106 & 0.0227 \\
Sample Cov & 2.2564 & 1.5294 & 2.9575 & 0.2327 & 0.1470 & {\bf 0.6181 }& 0.0106 & 0.0092 & 0.0437 \\
\hline
\end{tabular}%
    \caption{Annualized Sharpe ratios, returns, and variance with different methods of estimating the precision matrix, with different objective functions, applied to firms that logistic regression has screened. Portfolio abbreviations as in Table~\ref{tab1}.}
    \label{tab3}
\end{table}

\begin{table}[htbp]
    \centering
    {\bf HUMAN ANALYSTS WITH QUANTITATIVE WEIGHTING: Nov 2023-April 2024}\\
    \begin{tabular}{l|ccc|ccc|ccc}
\hline
 & \multicolumn{3}{c|}{\textbf{Sharpe Ratio}} & \multicolumn{3}{c|}{\textbf{Returns}} & \multicolumn{3}{c}{\textbf{Variance}} \\
\textbf{Method} & \textbf{GMV} & \textbf{MV} & \textbf{MSR} & \textbf{GMV} & \textbf{MV} & \textbf{MSR} & \textbf{GMV} & \textbf{MV} & \textbf{MSR} \\ \hline
NW & 1.9050 & 1.3314 & 2.1859 & 0.1676 & 0.1174 & 0.2033 & 0.0077 & 0.0078 & 0.0086 \\
Residual NW & 1.5191 & 1.5564 & 2.3350 & 0.1143 & 0.1198 & 0.3817 & 0.0057 & 0.0059 & 0.0267 \\
Deep learning & 1.8167 & 1.3572 & 2.7316 & 0.1561 & 0.1152 & 0.2581 & 0.0074 & 0.0072 & 0.0089 \\
POET & 1.8452 & 1.3944 & 2.7487 & 0.1680 & 0.1278 & 0.2428 & 0.0083 & 0.0084 & 0.0078 \\
NLS & 2.4773 & 2.2076 & 3.1745 & 0.1948 & 0.1626 & 0.6060 & 0.0062 & \textbf{0.0054} & 0.0364 \\ 
Sample Cov & 2.4661 & 2.2742 & {\bf 3.5374} & 0.2057 & 0.1785 & {\bf 0.8802} & 0.0070 & 0.0062 & 0.0619 \\
\hline
\end{tabular}%
    \caption{Annualized Sharpe ratios, returns, and variance with different methods of estimating the precision matrix, with different objective functions, applied to firms that analysts have screened. Portfolio abbreviations as in Table~\ref{tab1}.}
    \label{tab4}
\end{table}

\begin{table}[htbp]
    \centering
    {\bf FINBERT WITH QUANTITATIVE WEIGHTING: Nov 2023-April 2024}\\
    \begin{tabular}{l|ccc|ccc|ccc}
\hline
 & \multicolumn{3}{c|}{\textbf{Sharpe Ratio}} & \multicolumn{3}{c|}{\textbf{Returns}} & \multicolumn{3}{c}{\textbf{Variance}} \\
\textbf{Method} & \textbf{GMV} & \textbf{MV} & \textbf{MSR} & \textbf{GMV} & \textbf{MV} & \textbf{MSR} & \textbf{GMV} & \textbf{MV} & \textbf{MSR} \\ \hline
NW & 1.2999 & 0.8588 & 1.7441 & 0.1842 & 0.1328 & 0.2284 & 0.0201 & 0.0239 & 0.0172 \\
Residual NW & 1.6709 & 1.2760 & 5.0523 & 0.1937 & 0.1692 & 0.3982 & 0.0134 & 0.0176 & 0.0062 \\
Deep learning & 2.0878 & 1.4275 & 3.6581 & 0.2400 & 0.1831 & 0.3448 & 0.0132 & 0.0164 & 0.0089 \\
POET & 1.7141 & 1.2507 & 2.3143 & 0.2304 & 0.1888 & 0.2758 & 0.0181 & 0.0228 & 0.0142 \\
NLS & 2.9408 & 1.8150 & 6.6874 & 0.2551 & 0.1991 & 0.4507 & 0.0075 & 0.0120 & \textbf{0.0045} \\ 
Sample Cov & 4.0391 & 2.4555 & {\bf 7.2046} & 0.2986 & 0.2443 & {\bf 0.4920} & 0.0055 & 0.0099 & 0.0047 \\
\hline
\end{tabular}%
    \caption{Annualized Sharpe ratios, returns, and variance with different methods of estimating the precision matrix, with different objective functions, applied to firms that FinBERT have screened. Portfolio abbreviations as in Table~\ref{tab1}.}
    \label{tab5}
\end{table}

\begin{table}[htbp]
    \centering
    {\bf LLM-S + HUMAN ANALYSTS WITH QUANTITATIVE WEIGHTING: Nov 2023-April 2024}\\
    \begin{tabular}{l|ccc|ccc|ccc}
\hline
 & \multicolumn{3}{c|}{\textbf{Sharpe Ratio}} & \multicolumn{3}{c|}{\textbf{Returns}} & \multicolumn{3}{c}{\textbf{Variance}} \\
\textbf{Method} & \textbf{GMV} & \textbf{MV} & \textbf{MSR} & \textbf{GMV} & \textbf{MV} & \textbf{MSR} & \textbf{GMV} & \textbf{MV} & \textbf{MSR} \\ \hline
NW & 1.1550 & 0.3502 & 1.6014 & 0.1046 & 0.0315 & 0.1630 & 0.0082 & 0.0081 & 0.0104 \\
Residual NW & 0.2685 & 0.2037 & 2.5673 & 0.0230 & 0.0193 & 0.3746 & 0.0073 & 0.0090 & 0.0213 \\
Deep learning & 1.0221 & 0.3292 & 2.0881 & 0.0936 & 0.0293 & 0.2240 & 0.0084 & 0.0079 & 0.0115 \\
POET & 1.2053 & 0.3307 & 1.7578 & 0.1028 & 0.0295 & 0.1629 & 0.0073 & 0.0080 & 0.0086 \\
NLS & 1.2285 & 0.6378 & \textbf{2.7553} & 0.0946 & 0.0513 & 0.5166 & {\bf 0.0059} & 0.0065 & 0.0352 \\
Sample Cov & 0.8955 & 0.3824 & 2.4230 & 0.0627 & 0.0263 & {\bf 0.7007} & 0.0049 & 0.0047 & 0.0836 \\
\hline
\end{tabular}%
    \caption{Annualized Sharpe ratios, returns, and variance with different methods of estimating the precision matrix, with different objective functions, applied to firms that LLM+analysts have screened, from I/B/E/S. Portfolio abbreviations as in Table~\ref{tab1}.}
    \label{tab6}
\end{table}

\begin{table}[htbp]
    \centering
    {\bf FINBERT + HUMAN ANALYSTS WITH QUANTITATIVE WEIGHTING: Nov 2023-April 2024}\\
    \begin{tabular}{l|ccc|ccc|ccc}
\hline
 & \multicolumn{3}{c|}{\textbf{Sharpe Ratio}} & \multicolumn{3}{c|}{\textbf{Returns}} & \multicolumn{3}{c}{\textbf{Variance}} \\
\textbf{Method} & \textbf{GMV} & \textbf{MV} & \textbf{MSR} & \textbf{GMV} & \textbf{MV} & \textbf{MSR} & \textbf{GMV} & \textbf{MV} & \textbf{MSR} \\ \hline
NW & 2.2250 & 1.5567 & 1.3170 & 0.1466 & 0.1187 & 0.1179 & 0.0043 & 0.0058 & 0.0080 \\
Residual NW & 1.8127 & 1.2533 & 1.0236 & 0.1394 & 0.1167 & 0.1130 & 0.0059 & 0.0087 & 0.0122 \\
Deep learning & 2.1149 & 1.4905 & 1.2579 & 0.1489 & 0.1241 & 0.1230 & 0.0050 & 0.0069 & 0.0096 \\
POET & \textbf{2.4582} & 1.5254 & 1.6655 & \textbf{0.1545} & 0.1234 & 0.1354 & \textbf{0.0040} & 0.0065 & 0.0066 \\
NLS & 1.9685 & 1.4311 & 1.2129 & 0.1439 & 0.1232 & 0.1296 & 0.0053 & 0.0074 & 0.0114 \\ 
Sample Cov & 1.8717 & 1.3487 & 1.1807 & 0.1401 & 0.1186 & 0.1275 & 0.0056 & 0.0077 & 0.0117 \\
\hline
\end{tabular}%
    \caption{Annualized Sharpe ratios, returns, and variance with different methods of estimating the precision matrix, with different objective functions, applied to firms that FinBERT+analysts have screened, from I/B/E/S. Portfolio abbreviations as in Table~\ref{tab1}.}
    \label{tab7}
\end{table}

\begin{table}[htbp]
    \centering
    {\bf LLM-S + FINBERT WITH QUANTITATIVE WEIGHTING: Nov 2023-April 2024}\\
    \begin{tabular}{l|ccc|ccc|ccc}
\hline
 & \multicolumn{3}{c|}{\textbf{Sharpe Ratio}} & \multicolumn{3}{c|}{\textbf{Returns}} & \multicolumn{3}{c}{\textbf{Variance}} \\
\textbf{Method} & \textbf{GMV} & \textbf{MV} & \textbf{MSR} & \textbf{GMV} & \textbf{MV} & \textbf{MSR} & \textbf{GMV} & \textbf{MV} & \textbf{MSR} \\ \hline
NW & 1.3801 & 0.9793 & 1.8579 & 0.2007 & 0.1564 & 0.2443 & 0.0211 & 0.0255 & 0.0173 \\
Residual NW & 1.8382 & 1.6028 & 5.7291 & 0.2458 & 0.2274 & 0.4604 & 0.0179 & 0.0201 & 0.0065 \\
Deep learning & 2.0552 & 1.4812 & 3.4411 & 0.2648 & 0.2119 & 0.3631 & 0.0166 & 0.0205 & 0.0111 \\
POET & 1.7976 & 1.3027 & 2.5453 & 0.2616 & 0.2047 & 0.3265 & 0.0212 & 0.0247 & 0.0165 \\
NLS & 3.0202 & 2.0658 & \textbf{8.1612} & 0.3057 & 0.2516 & 0.5263 & 0.0102 & 0.0148 & \textbf{0.0042} \\ 
Sample Cov & 3.9676 & 2.6332 & 7.2993 & 0.3390 & 0.2855 & {\bf 0.5575} & 0.0073 & 0.0118 & 0.0058 \\
\hline
\end{tabular}%
    \caption{Annualized Sharpe ratios, returns, and variance with different methods of estimating the precision matrix, with different objective functions, applied to firms that FinBERT+LLM has screened. Portfolio abbreviations as in Table~\ref{tab1}.}
    \label{tab8}
\end{table}

\begin{table}[htbp]
    \centering
    {\bf LLM-S + FINBERT + HUMAN ANALYSTS WITH QUANTITATIVE WEIGHTING: Nov 2023-April 2024}\\
    \begin{tabular}{l|ccc|ccc|ccc}
\hline
 & \multicolumn{3}{c|}{\textbf{Sharpe Ratio}} & \multicolumn{3}{c|}{\textbf{Returns}} & \multicolumn{3}{c}{\textbf{Variance}} \\
\textbf{Method} & \textbf{GMV} & \textbf{MV} & \textbf{MSR} & \textbf{GMV} & \textbf{MV} & \textbf{MSR} & \textbf{GMV} & \textbf{MV} & \textbf{MSR} \\ \hline
NW & 0.5310 & 0.3095 & 0.6816 & 0.0691 & 0.0411 & 0.0922 & 0.0169 & 0.0176 & 0.0183 \\
Residual NW & 0.3113 & 0.3212 & 1.1117 & 0.0422 & 0.0443 & 0.1835 & 0.0184 & 0.0190 & 0.0273 \\
Deep learning & 0.6753 & 0.4439 & 1.0251 & 0.0925 & 0.0608 & 0.1512 & 0.0188 & 0.0188 & 0.0217 \\
POET & 0.3727 & 0.2973 & 0.6593 & 0.0476 & 0.0391 & 0.0851 & \textbf{0.0163} & 0.0173 & 0.0167 \\
NLS & 0.5763 & 0.4292 & 1.3912 & 0.0743 & 0.0556 & 0.2558 & 0.0166 & 0.0168 & 0.0338 \\
Sample Cov & 0.5067 & 0.3828 & {\bf 1.4193} & 0.0672 & 0.0510 & {\bf 0.2657} & 0.0176 & 0.0178 & 0.0350 \\
\hline
\end{tabular}%
    \caption{Annualized Sharpe ratios, returns, and variance with different methods of estimating the precision matrix, with different objective functions, applied to firms that LLM+FinBERT+analysts have screened. Portfolio abbreviations as in Table~\ref{tab1}.}\label{tab9}
\end{table}

\subsection{Additional metrics for medium term}\label{turnover}

This section presents the average turnover, concentration, leverage, and drawdown for the portfolios in the medium term, under baseline (no screening) and the Agentic AI portfolios.

Leverage is calculated as the sum of the absolute value of all weights. Note that occasionally in our empirical backtests, there might not be enough training data for all stocks in the portfolio (especially true for small portfolios consisting of 2 stocks, for example). In this case, we treat this as a portfolio liquidation and return an empty portfolio. Hence, leverage may be less than 1 in some cases. Concentration is defined as the HHI index --- the sum of the squares of all weights.
Max drawdown is defined as the maximum peak-to-trough drop of the portfolio over time.

\begin{table}[htbp]
    \centering
    {\bf BASELINE QUANTITATIVE WEIGHTING: OCT 2021-2024}\\
    \vspace{0.2cm} 
    \begin{tabular}{l|cccc}
    \hline
    \textbf{Method} & \textbf{Turnover} & \textbf{Leverage} & \textbf{Concentration} & \textbf{Drawdown} \\ \hline
    \multicolumn{5}{c}{\textbf{Panel A: Global Minimum Variance (GMV)}} \\ \hline
    NW            & 0.1134 & 1.1600 & 0.0064 & -0.1485 \\
    Residual NW   & 0.4053 & 2.6745 & 0.0438 & -0.1826 \\
    Deep learning & 2.1449 & 8.9213 & 0.5458 & -0.1612 \\
    POET          & 0.2069 & 1.4847 & 0.0145 & -0.1443 \\
    NLS           & 0.6862 & 3.7686 & 0.0828 & -0.1578 \\ \hline
    \multicolumn{5}{c}{\textbf{Panel B: Mean-Variance (MV)}} \\ \hline
    NW            & 0.1502 & 1.1631 & 0.0067 & -0.1457 \\
    Residual NW   & 0.4353 & 2.7019 & 0.0454 & -0.1707 \\
    Deep learning & 2.2813 & 9.2858 & 0.5882 & -0.1695 \\
    POET          & 0.2276 & 1.5108 & 0.0147 & -0.1365 \\
    NLS           & 0.6855 & 3.7765 & 0.0831 & -0.1560 \\ \hline
    \multicolumn{5}{c}{\textbf{Panel C: Maximum Sharpe Ratio (MSR)}} \\ \hline
    NW            & 0.1174 & 1.1806 & 0.0070 & -0.1697 \\
    Residual NW   & 1.2037 & 5.2355 & 0.1809 & -0.2362 \\
    Deep learning & 4.1655 & 13.7410 & 1.2604 & -0.2936 \\
    POET          & 0.3186 & 1.7959 & 0.0200 & -0.1175 \\
    NLS           & 1.6740 & 7.8410 & 0.3636 & -0.1180 \\ \hline
    \end{tabular}%
    \caption{Turnover, leverage, concentration, and drawdown of portfolios with different methods of estimating the precision matrix.}
\end{table}

\begin{table}[htbp]
    \centering
    {\bf LLM-S + FINBERT WITH QUANTITATIVE WEIGHTING: OCT 2021-2024}\\
    \vspace{0.2cm} 
    \begin{tabular}{l|cccc}
    \hline
    \textbf{Method} & \textbf{Turnover} & \textbf{Leverage} & \textbf{Concentration} & \textbf{Drawdown} \\ \hline
    \multicolumn{5}{c}{\textbf{Panel A: Global Minimum Variance (GMV)}} \\ \hline
    NW            & 1.2463 & 0.9148 & 0.1598 & -0.1341 \\
    Residual NW   & 1.6768 & 1.2479 & 0.3050 & -0.1131 \\
    Deep learning & 1.3630 & 0.9336 & 0.1945 & -0.1159 \\
    POET          & 1.3385 & 0.9182 & 0.1536 & -0.1126 \\
    NLS           & 1.7377 & 1.2567 & 0.2884 & -0.1168 \\ \hline
    \multicolumn{5}{c}{\textbf{Panel B: Mean-Variance (MV)}} \\ \hline
    NW            & 2.2057 & 1.2827 & 1.6299 & -0.4146 \\
    Residual NW   & 2.7210 & 1.6593 & 1.7787 & -0.4299 \\
    Deep learning & 2.3354 & 1.3149 & 1.6660 & -0.4176 \\
    POET          & 2.2640 & 1.2494 & 1.6169 & -0.3957 \\
    NLS           & 2.7463 & 1.6455 & 1.7588 & -0.4300 \\ \hline
    \multicolumn{5}{c}{\textbf{Panel C: Maximum Sharpe Ratio (MSR)}} \\ \hline
    NW            & 1.3597 & 0.9387 & 0.1967 & -0.1475 \\
    Residual NW   & 2.5890 & 1.7878 & 0.5182 & -0.1972 \\
    Deep learning & 1.6387 & 1.0625 & 0.2505 & -0.1170 \\
    POET          & 1.4348 & 0.9279 & 0.1632 & -0.1097 \\
    NLS           & 2.8375 & 1.8906 & 0.5622 & -0.2559 \\ \hline
    \end{tabular}%
    \caption{Turnover, leverage, concentration, and drawdown of portfolios with different methods of estimating the precision matrix.}
\end{table}

\clearpage

\bibliographystyle{chicagoa}
\bibliography{agentic}
\end{document}